%% file: CorrCoeReProCS_full.tex
\documentclass[10pt, twoside]{IEEEtran} \onecolumn 
\usepackage{epsfig,algorithm,amsmath,amssymb,color,algorithmic}
\usepackage{epstopdf}
\usepackage{caption}
\usepackage{subcaption}
\usepackage[hidelinks]{hyperref}
\hypersetup{draft}
\begin{document}

\input{zcom}

\input{ReProCSCommands}
\newcommand{\Pb}{\mathbf{P}}
\newcommand{\Hc}{\mathcal{H}}
\newcommand{\Ic}{\mathcal{I}}
\newcommand{\tb}{\tilde{b}}
\newcommand{\tf}{\tilde{f}}
\allowdisplaybreaks

\title{Performance Guarantees for ReProCS -- Correlated Low-Rank Matrix Entries Case}


\author{
Jinchun Zhan, Namrata Vaswani and Chenlu Qiu
\thanks{
J. Zhan, N. Vaswani are with the ECE dept at Iowa State University and C. Qiu is with Traffic Management Research Institute of the Ministry of Public Security, China. Email: \{jzhan,namrata\}@iastate.edu, echotosusan@gmail.com.
}
}
\maketitle

\begin{abstract}
Online or recursive robust PCA can be posed as a problem of recovering a sparse vector, $S_t$, and a dense vector, $L_t$, which lies in a slowly changing low-dimensional subspace, from $M_t:= S_t + L_t$ on-the-fly as new data comes in. For initialization, it is assumed that an accurate knowledge of the subspace in which $L_0$ lies is available. In recent works, Qiu et al proposed and analyzed a novel solution to this problem called recursive projected compressed sensing or ReProCS. In this work, we relax one limiting assumption of Qiu et al's result. Their work required that the $L_t$'s be mutually independent over time. However this is not a practical assumption, e.g., in the video application, $L_t$ is the background image sequence and one would expect it to be correlated over time. In this work we relax this and allow the $L_t$'s to follow an autoregressive model. We are able to show that under mild assumptions and under a denseness assumption on the unestimated part of the changed subspace, with high probability (w.h.p.), ReProCS can exactly recover the support set of $S_t$ at all times; the reconstruction errors of both $S_t$ and $L_t$ are upper bounded by a time invariant and small value; and the subspace recovery error decays to a small value within a finite delay of a subspace change time. Because the last assumption depends on an algorithm estimate, this result cannot be interpreted as a correctness result but only a useful step towards it.
\end{abstract}

\section{Introduction}
Principal Components Analysis (PCA) is a widely used dimension reduction technique that finds a small number of orthogonal basis vectors, called principal components (PCs), along which most of the variability of the dataset lies. Often, for time series data, the PCs space changes gradually over time. Updating it on-the-fly (recursively) in the presence of outliers, as more data comes in is referred to as online or recursive robust PCA \cite{ipca_weightedand,Li03anintegrated}. As noted in earlier work, an outlier is well modeled as a sparse vector.
With this, as will be evident, this problem can also be interpreted as one of recursive sparse recovery in large but structured (low-dimensional) noise.

A key application where the robust PCA problem occurs is in video analysis where the goal is to separate a slowly changing background from moving foreground objects \cite{Torre03aframework,rpca}. 
If we stack each frame as a column vector, the background is well modeled as being dense and lying in a low dimensional subspace that may gradually change over time, while the moving foreground objects constitute the sparse outliers \cite{rpca}. Other applications include detection of brain activation patterns from functional MRI sequences  or detection of anomalous behavior in dynamic networks \cite{mateos_anomaly}.
There has been a large amount of earlier work on robust PCA, e.g. see  \cite{Torre03aframework}. In recent works \cite{rpca,rpca2}, the batch robust PCA problem has been posed as one of separating a low rank matrix, ${\cal L}_t := [L_0, L_1,\dots,L_t]$, from a sparse matrix, ${\cal S}_t := [S_0, S_1,\dots,S_t]$, using the measurement matrix, ${\cal M}_t := [M_0, M_1,\dots,M_t] = {\cal L}_t+ {\cal S}_t$. A novel convex optimization solution called principal components' pursuit (PCP) has been proposed and shown to achieve exact recovery under mild assumptions. Since then, the batch robust PCA problem, or what is now also often called the sparse+low-rank recovery problem, has been studied extensively. We do not discuss all the works here due to limited space.


Online or recursive robust PCA can be posed as a problem of recovering a sparse vector, $S_t$, and a dense vector, $L_t$, which lies in a slowly changing low-dimensional subspace, from $M_t:= S_t + L_t$ on-the-fly as new data comes in. For initialization, it is assumed that an accurate knowledge of the subspace in which $L_0$ lies is available. In recent works, Qiu et al proposed and analyzed a novel solution to this problem called ReProCS \cite{rrpcp_allerton,rrpcp_perf,rrpcp_allerton11}. 

{\em Contribution: }
In this work, we relax one limiting assumption of Qiu et al's result. Their work required that the $L_t$'s be mutually independent over time. However this is not a practical assumption, e.g., in the video application, $L_t$ is the background image sequence and one would expect it to be correlated over time. In this work we relax this and allow the $L_t$'s to follow an autoregressive (AR) model. We are able to show that, under mild assumptions and a denseness assumption on the currently unestimated subspace, with high probability (w.h.p.), ReProCS (with subspace change model known) can exactly recover the support set of $S_t$ at all times; the reconstruction errors of both $S_t$ and $L_t$ are upper bounded by a time invariant and small value; and the subspace recovery error decays to a small value within a finite delay of a subspace change time. The last assumption depends on an algorithm estimate and hence this result also cannot be interpreted as a correctness result but only a useful step towards it.

%
To the best of our knowledge, the result of Qiu et al and this follow-up work are among the first results  for any recursive (online) robust PCA approach, and also for recursive sparse recovery in large but structured (low-dimensional) noise.


Other very recent work on algorithms for recursive / online robust PCA includes \cite{grass_undersampled,rpcal0sur,frpca,xu_nips2013_1,xu_nips2013_2,mateos_anomaly}. In  \cite{xu_nips2013_1,xu_nips2013_2}, two online algorithms for robust PCA (that do not model the outlier as a sparse vector but only as a vector that is ``far" from the data subspace) have been shown to approximate the batch solution and do so only asymptotically. 

\section{Notation and Background}
\label{bgnd}
\subsection{Notation}

For a set $T \subset \{1,2,\dots, n\}$, we use $|T|$ to denote its cardinality, i.e., the number of elements in $T$. We use $T^c$ to denote its complement w.r.t. $\{1,2,\dots n\}$, i.e. $T^c:= \{i \in \{1,2,\dots n\}: i \notin T \}$.

We use the interval notation, $[t_1, t_2]$, to denote the set of all integers between and including $t_1$ to $t_2$, i.e. $[t_1, t_2]:=\{t_1, t_1+1, \dots, t_2\}$. For a vector $v$, $v_i$ denotes the $i$th entry of $v$ and $v_T$ denotes a vector consisting of the entries of $v$ indexed by $T$. We use $\|v\|_p$ to denote the $\ell_p$ norm of $v$. The support of $v$, $\text{supp}(v)$, is the set of indices at which $v$ is nonzero, $\text{supp}(v) := \{i : v_i\neq 0\}$. We say that $v$ is s-sparse if $|\text{supp}(v)| \leq  s$.

For a matrix $B$, $B'$ denotes its transpose, and $B^{\dag}$ its pseudo-inverse. For a matrix with linearly independent columns, $B^{\dag} = (B'B)^{-1}B'$.
We use $\|B\|_2:= \max_{x \neq 0} \|Bx\|_2/\|x\|_2$ to denote the induced 2-norm of the matrix. Also, $\|B\|_*$ is the nuclear norm (sum of singular values) and $\|B\|_{\max}$ denotes the maximum over the absolute values of all its entries.
We let $\sigma_i(B)$ denotes the $i$th largest singular value of $B$. For a Hermitian matrix, $B$, we use the notation $B \overset{EVD}{=} U \Lambda U'$ to denote the eigenvalue decomposition of $B$. Here $U$ is an orthonormal matrix and $\Lambda$ is a diagonal matrix with entries arranged in decreasing order. Also, we use $\lambda_i(B)$ to denote the $i$th largest eigenvalue of a Hermitian matrix $B$ and we use $\lambda_{\max}(B)$ and $\lambda_{\min}(B)$ denote its maximum and minimum eigenvalues. If $B$ is Hermitian positive semi-definite (p.s.d.), then $\lambda_i(B) = \sigma_i(B)$. For Hermitian matrices $B_1$ and $B_2$, the notation $B_1 \preceq B_2$ means that $B_2-B_1$ is p.s.d. Similarly, $B_1 \succeq B_2$ means that $B_1-B_2$ is p.s.d.

For a Hermitian matrix $B$, $\|B\|_2 = \sqrt{\max( \lambda_{\max}^2(B), \lambda_{\min}^2(B) )}$ and thus, $\|B\|_2 \le b$ implies that $-b \le \lambda_{\min}(B) \le \lambda_{\max}(B) \le b$.


We use $I$ to denote an identity matrix of appropriate size. For an index set $T$ and a matrix $B$, $B_T$ is the sub-matrix of $B$ containing columns with indices in the set $T$. Notice that $B_T = B I_T$. Given a matrix $B$ of size $m \times n$ and $B_2$ of size $m \times n_2$, $[B \ B_2]$ constructs a new matrix by concatenating matrices $B$ and $B_2$ in the horizontal direction. Let $B_{\text{rem}}$ be a matrix containing some columns of $B$. Then $B \setminus B_{\text{rem}}$ is the matrix $B$ with columns in $B_{\text{rem}}$ removed.

For a tall matrix $P$, $\Span(P)$ denotes the subspace spanned by the column vectors of $P$.

The notation $[.]$ denotes an empty matrix.

\begin{definition}
We refer to a tall matrix $P$ as a {\em basis matrix} if it satisfies $P'P=I$.
\end{definition}

\begin{definition}
We use the notation $Q = \mathrm{basis}(B)$ to mean that $Q$ is a basis matrix and $\Span(Q) = \Span(B)$.  In other words, the columns of $Q$ form an orthonormal basis for the range of $B$.
\end{definition}


\begin{definition}\label{defn_delta}
The {\em $s$-restricted isometry constant (RIC)} \cite{decodinglp}, $\delta_s$, for an $n \times m$ matrix $\Psi$ is the smallest real number satisfying $(1-\delta_s) \|x\|_2^2 \leq \|\Psi_T x\|_2^2 \leq (1+\delta_s) \|x\|_2^2$
for all sets $T \subseteq \{1,2,\dots n \}$ with $|T| \leq s$ and all real vectors $x$ of length $|T|$.%
\end{definition} 

It is easy to see that $\max_{T:|T| \le s} \|({\Psi_T}'\Psi_T)^{-1}\|_2 \le \frac{1}{1-\delta_s(\Psi)}$ \cite{decodinglp}.


\begin{definition}
We give some notation for random variables in this definition.
\ben
\item We let $\E[Z]$ denote the expectation of a random variable (r.v.) $Z$ and $\E[Z|X]$ denote its conditional expectation given another r.v. $X$.
\item Let $\calb$ be a set of values that a r.v. $Z$ can take. We use $\ecalb$ to denote the {\em event} $Z \in \calb$, i.e. $\ecalb:= \{Z \in \calb\}$.
\item The probability of any event $\ecalb$ can be expressed as \cite{grimmett},
$$\mathbf{P}(\ecalb) := \E[\mathbb{I}_\calb(Z)].$$
where
\bea
\mathbb{I}_\calb(Z): = \left\{ \begin{array}{cc}
                                 1 \ & \ \text{if} \ Z \in \calb \nn \\
                                 0  \ & \ \text{otherwise}
                                 \end{array} \right. \nn
\eea
is the indicator function on the set $\calb$.

\item For two events $\ecalb, \tilde\ecalb$,  $\mathbf{P}(\ecalb|\tilde\ecalb)$ refers to the conditional probability of $\ecalb$ given $\tilde\ecalb$, i.e. $\mathbf{P}(\ecalb|\tilde\ecalb) := \mathbf{P}(\ecalb, \tilde\ecalb)/\mathbf{P}(\tilde\ecalb)$.

\item For a r.v. $X$, and a set $\calb$ of values that the r.v. $Z$ can take, the notation $\mathbf{P}(\ecalb|X)$ is defined as 
$$\mathbf{P}(\ecalb|X) := \E[\mathbb{I}_{\calb}(Z)|X].$$
Notice that $\mathbf{P}(\ecalb|X)$ is a r.v. (it is a function of the r.v. $X$) that always lies between zero and one.
%
\een
\label{probdefs}
\end{definition}

Finally, RHS refers to the right hand side of an equation or inequality; w.p. means ``with probability"; and w.h.p. means ``with high probability".

\subsection{Compressive Sensing result}
The error bound for noisy compressive sensing (CS) based on the RIC is as follows \cite{candes_rip}.
%
\begin{theorem}[\cite{candes_rip}]
\label{candes_csbound}
Suppose we observe
\beq
y := \Psi x + z \nn
\eeq
where $z$ is the noise. Let $\hat{x}$ be the solution to following problem
\beq
\min_{x} \|x\|_1 \ \text{subject to} \  \|y - \Psi x\|_2 \leq \xi  \label{*}
\eeq
Assume that $x$ is $s$-sparse, $\|z\|_2 \leq \xi$, and $\delta_{2s}(\Psi) < b (\sqrt{2}-1)$ for some $0 \le b < 1$. Then the solution of (\ref{*}) obeys
$$\|\hat{x} - x\|_2 \leq C_1 \xi$$
with $\ds C_1 = \frac{4\sqrt{1+\delta_{2s}(\Psi)}}{1-(\sqrt{2}+1)\delta_{2s}(\Psi)} \le \frac{4\sqrt{1+b (\sqrt{2}-1)}}{1-(\sqrt{2}+1)b (\sqrt{2}-1)}$.
\end{theorem}

\begin{remark}
Notice that if $b$ is small enough, $C_1$ is a small constant but $C_1 >1$. For example, if $\delta_{2s}(\Psi) \le 0.15$, then $C_1 \le 7$.
If $C_1 \xi  > \|x\|_2$, the normalized reconstruction error bound would be greater than $1$, making the result useless. Hence, (\ref{*}) gives a small reconstruction error bound only for the small noise case, i.e., the case where $\|z\|_2 \leq \xi \ll \|x\|_2$. 
\end{remark}

\subsection{Results from linear algebra}

Davis and Kahan's $\sin \theta$ theorem \cite{davis_kahan} studies the rotation of eigenvectors by perturbation.

\begin{theorem}[$\sin \theta$ theorem \cite{davis_kahan}] \label{sin_theta}
Given two Hermitian matrices $\mathcal{A}$ and $\mathcal{H}$ satisfying
\beq \label{sindecomp}
\mathcal{A} = \left[ \begin{array}{cc} E & E_{\perp} \\ \end{array} \right]
\left[\begin{array}{cc} A\ & 0\ \\ 0 \ & A_{\perp} \\ \end{array} \right]
\left[ \begin{array}{c} E' \\ {E_{\perp}}' \\ \end{array} \right], \
\mathcal{H} = \left[ \begin{array}{cc} E & E_{\perp} \\ \end{array} \right]
\left[\begin{array}{cc} H \ & B'\ \\ B \ & H_{\perp} \\ \end{array} \right]
\left[ \begin{array}{c} E' \\ {E_{\perp}}' \\ \end{array} \right]
\eeq
where $[E \ E_{\perp}]$ is an orthonormal matrix. The two ways of representing $\mathcal{A}+\mathcal{H}$ are
\beq
\mathcal{A} + \mathcal{H}  = \left[ \begin{array}{cc} E & E_{\perp} \\ \end{array} \right]
\left[\begin{array}{cc} A + H \ & B'\ \\ B \ & A_{\perp} + H_{\perp} \\ \end{array} \right]
\left[ \begin{array}{c} E' \\ {E_{\perp}}' \\ \end{array} \right]
= \left[ \begin{array}{cc} F & F_{\perp} \\ \end{array} \right]
\left[\begin{array}{cc} \Lambda\ & 0\ \\ 0 \ & \Lambda_{\perp} \\ \end{array} \right]
\left[ \begin{array}{c} F' \\ {F_{\perp}}' \\ \end{array} \right] \nn
\eeq
where $[F\ F_{\perp}]$ is another orthonormal matrix. Let $\mathcal{R} := (\mathcal{A}+\mathcal{H}) E - \mathcal{A}E = \mathcal{H} E $. If $ \lambda_{\min}(A) >\lambda_{\max}(\Lambda_{\perp})$, then
\beq
\|(I-F F')E \|_2 \leq \frac{\|\mathcal{R}\|_2}{\lambda_{\min}(A) - \lambda_{\max}(\Lambda_{\perp})} \nn
\eeq

\end{theorem}
The above result bounds the amount by which the two subspaces $\Span(E)$ and $\Span(F)$ differ as a function of the norm of the perturbation $\|\mathcal{R}\|_2$ and of the gap between the minimum eigenvalue of $A$ and the maximum eigenvalue of $\Lambda_{\perp}$. Next, we state Weyl's theorem which bounds the eigenvalues of a perturbed Hermitian matrix, followed by Ostrowski's theorem.

\begin{theorem}[Weyl \cite{hornjohnson}]\label{weyl}
Let $\mathcal{A}$ and $\mathcal{H}$ be two  $n \times n$ Hermitian matrices. For each $i = 1,2,\dots,n$ we have
$$\lambda_i(\mathcal{A}) + \lambda_{\min}(\mathcal{H}) \leq \lambda_i(\mathcal{A}+\mathcal{H}) \leq \lambda_i(\mathcal{A}) + \lambda_{\max}(\mathcal{H})$$
\end{theorem}

\begin{theorem}[Ostrowski \cite{hornjohnson}]\label{ost}
Let $H$ and $W$ be $n \times n$ matrices, with $H$ Hermitian and $W$ nonsingular. For each $i=1,2 \dots n$, there exists a positive real number $\theta_i$ such that $\lambda_{\min} (WW') \leq \theta_i \leq \lambda_{\max}(W{W}')$ and $\lambda_i(W H {W}') = \theta_i \lambda_i(H)$. Therefore,
$$\lambda_{\min}(W H {W}') \geq \lambda_{\min} (W{W}') \lambda_{\min} (H)$$
\end{theorem}

The following lemma proves some simple linear algebra facts.
\begin{lem} \label{lemma0}\label{hatswitch}
Suppose that $P$, $\Phat$ and $Q$ are three basis matrices. Also, $P$ and $\Phat$ are of the same size, ${Q}'P = 0$ and $\|(I-\Phat{\Phat}')P\|_2 = \zeta_*$. Then,
\ben
  \item $\|(I-\Phat{\Phat}')PP'\|_2 =\|(I - P{P}')\Phat{\Phat}'\|_2 =  \|(I - P P')\Phat\|_2 = \|(I - \Phat \Phat')P\|_2 =  \zeta_*$
  \item $\|P{P}' - \Phat {\Phat}'\|_2 \leq 2 \|(I-\Phat{\Phat}')P\|_2 = 2 \zeta_*$
  \item $\|{\Phat}' Q\|_2 \leq \zeta_*$ \label{lem_cross}
  \item $ \sqrt{1-\zeta_*^2} \leq \sigma_i((I-\Phat \Phat')Q)\leq 1 $
\een
\end{lem}
The proof is in the Appendix.

\subsection{High probability tail bounds for sums of random matrices}

The following lemma follows easily using Definition \ref{probdefs}. We will use this at various places in the paper.
\begin{lem}
Suppose that $\calb$ is the set of values that the r.v.s $X,Y$ can take. Suppose that $\calc$ is a set of values that the r.v. $X$ can take.
For a $0 \le p \le 1$, if $\mathbf{P}(\ecalb|X) \ge p$ for all $X \in \calc$,  then $\mathbf{P}(\ecalb|\ecalc) \ge p$ as long as $\mathbf{P}(\ecalc)> 0$.
\label{rem_prob}
\end{lem}
The proof is in the Appendix.

The following lemma is an easy consequence of the chain rule of probability applied to a contracting sequence of events.
\begin{lem} \label{subset_lem}
For a sequence of events $E_0^e, E_1^e, \dots E_m^e$ that satisfy $E_0^e \supseteq E_1^e  \supseteq E_2^e \dots  \supseteq E_m^e$, the following holds
$$\mathbf{P}(E_m^e|E_0^e) = \prod_{k=1}^{m} \mathbf{P}(E_k^e | E_{k-1}^e).$$
\end{lem}
\begin{proof}
$\mathbf{P}(E_m^e|E_0^e) = \mathbf{P}(E_m^e, E_{m-1}^e, \dots E_0^e | E_0^e) = \prod_{k=1}^{m} \mathbf{P}(E_k^e | E_{k-1}^e, E_{k-2}^e, \dots E_0^e) = \prod_{k=1}^{m} \mathbf{P}(E_k^e | E_{k-1}^e)$.
\end{proof}

Next, we state the matrix Azuma inequality which gives tail bounds for sums of random matrices.

\begin{theorem}
(Matrix Azuma) Consider a finite adapted sequence $\{Z_t\}_{t\geq1}$ of Hermitian matrices in dimension $n$, and a fixed sequence $A_t$ of self-adjoint matrices that satisfy
$$\E_{t-1}(Z_t)=\E(Z_t|Z_1,\cdots,Z_{t-1})=0, \text{ and }Z_t^2\preceq A_t^2 \text{ almost surely.}$$
Then, for all $\epsilon> 0$,  $$\Pb\left(\lambda_{\max}\left(\sum_t Z_t\right)\leq \epsilon\right)\geq 1- n \exp\left(-\frac{\epsilon^2}{8\sigma^2}\right), \text{ where } \sigma^2=\left\|\sum_t{A_t^2}\right\|.$$
\label{azuma}
\end{theorem}

The proof is in the appendix.

\begin{corollary}[Matrix Azuma conditioned on another random variable for any zero mean matrix]\label{azuma_rec}
Consider an $\alpha$-length adapted sequence $\{Z_t\}$ of random matrices of size $n_1\times n_2$ given a random variable $X$. Assume that, for all $X\in\mathcal{C}$, (i) $\mathbf{P}(\|Z_t\|_2 \le b_1|X) = 1$ and (ii) $\E_{t-1}( Z_t|X)=0$. Then, for all $\epsilon >0$,
$$\mathbf{P} \left(\left\|\frac{1}{\alpha}\sum_{t=1}^{\alpha} Z_t \right\|\leq \epsilon|X\right) \geq 1-(n_1+n_2) \exp\left(-\frac{\alpha \epsilon^2}{32 b_1^2}\right)$$
\end{corollary}

The proof is in the appendix.

\begin{corollary}[Matrix Azuma conditioned on another random variable for a nonzero mean Hermitian matrix]\label{azuma_nonzero}
Consider an $\alpha$-length sequence $\{Z_t\}_{1\leq t\leq \alpha}$ of random Hermitian matrices of size $n\times n$ given a random variable $X$. Assume that, for all $X\in\mathcal{C}$, (i) $\mathbf{P}(b_1 I \preceq Z_t \preceq b_2 I|X) = 1, 1\leq t\leq \alpha$ and (ii) $b_3 I \preceq \frac{1}{\alpha}\sum_{t=1}^{\alpha} \E_{t-1}(Z_t|X) \preceq b_4 I $. Then for all $\epsilon > 0$,
\bea
&&\mathbf{P} \left(\lambda_{\max}\left(\frac{1}{\alpha}\sum_{t=1}^{\alpha} Z_t|X\right) \leq b_4 + \epsilon\right) \geq 1- n \exp\left(-\frac{\alpha \epsilon^2}{8(b_2-b_1)^2}\right)
\nn \\
&&\mathbf{P} \left(\lambda_{\min}\left(\frac{1}{\alpha}\sum_{t=1}^{\alpha} Z_t|X\right) \geq b_3 -\epsilon\right) \geq  1- n \exp\left(-\frac{\alpha \epsilon^2}{8(b_2-b_1)^2}\right) \nn
\eea
\end{corollary}

The proof is in the appendix.

\begin{definition}
Let function $\digamma(\alpha,\epsilon,b_2-b_1):=\exp(-\frac{\alpha \epsilon^2}{8(b_2-b_1)^2})$.
\end{definition}

\begin{lem}
If random variable $X$ and $Y$ are independent, $h(\cdot)$, $g(\cdot)$ are some functions, then $\E(Xh(Y)|g(Y))=\E(X)\E(h(Y)|g(Y))$.
\label{ind_expec}
\end{lem}
The proof is in the appendix.

\section{Problem Definition}
The measurement vector at time $t$, $M_t$, satisfies
\bea
M_t:=L_t+S_t
\eea
where $S_t$ is a sparse vector and $L_t$ is a dense vector that satisfies the model given below. Denote by $P_0$ a basis matrix for ${\cal L}_{t_{\train}}=[L_0, L_1, \cdots, L_{t_{\train}}]$, i.e., $\Span(P_0)=\Span({\cal L}_{t_{\train}})$. We are given an accurate enough estimate $\hat{P}_0$ for $P_0$, i.e., $\|(I-\hat{P}_0\hat{P}_0')P_0\|_2$ is small.
The goal is
\ben
\item to estimate both $S_t$ and $L_t$ at each time $t > t_\train$, and
\item to estimate $\Span({\cal L}_t)$ every so often. 
\een

\subsection{Signal Model}
\label{sig_mod}

\begin{sigmodel}[Model on $L_t$] \label{model_lt} \

\ben
\item We assume that $L_t = P_{(t)} a_t$ with $P_{(t)} = P_j$ for all $t_j \leq t <t_{j+1}$, $j=0,1,2 \cdots J$, where $P_j$ is an $n \times r_j$ basis matrix with $r_j  \ll \min(n, (t_{j+1} - t_j))$. At the change times, $t_j$, $P_j$ changes as $$P_j = [P_{j-1} \ P_{j,\new}].$$ Here, $P_{j,\new}$ is an $n \times c_{j,\new}$ basis matrix with $P_{j,\new}'P_{j-1} = 0$. Thus $r_j = r_{j-1} + c_{j,\new}$.   We let $t_0=0$ and $t_{J+1}$ equal the sequence length or $t_{J+1}=\infty$.



\item The vector of coefficients, $a_t:={P_{(t)}}'L_t$, satisfies the following autoregressive model
$$a_t = ba_{t-1}+\nu_t$$
where $b < 1$ is a scalar, $\E[\nu_t]=0$, $\nu_t$'s are mutually independent over time $t$ and the entries of any $\nu_t$ are pairwise uncorrelated, i.e. $\E[(\nu_t)_i (\nu_t)_j]=0$ when $i \neq j$.
\een
\end{sigmodel}

\begin{definition}
Define the covariance matrices of $\nu_t$ and $a_t$ to be the diagonal matrices $\Lambda_{\nu,t}:=\text{Cov}[\nu_t]=\E(\nu_t\nu_t')$ and $\Lambda_{a,t}:=\text{Cov}[a_t]=\E(a_ta_t').$
Then clearly,
\bea \label{armod}
\Lambda_{a,t} = b^2 \Lambda_{a,t-1} + \Lambda_{\nu,t}
\eea
Also, for $t_j \le t < t_{j+1}$, $a_t$ is an $r_j$ length vector which can be split as
  $$a_t ={P_j}'L_t = \vect{a_{t,*}}{a_{t,\new}}$$
where $a_{t,*}: = {P_{j-1}}'L_t$ is an $r_{j-1}$ length vector and $a_{t,\new}: = {P_{j,\new}}'L_t$ is a $c_{j,\new}$ length vector.
Thus, for this interval, $L_t$ can be rewritten as
\[
L_t = \left[ P_{j-1} \ P_{j,\new}\right] \vect{a_{t,*}}{a_{t,\new}} = P_{j-1} a_{t,*} + P_{j,\new} a_{t,\new}
\]
Also, $\Lambda_{a,t}$ can be split as
$$\Lambda_{a,t} = \left[ \begin{array}{cc}
(\Lambda_{a,t})_* &  0 \nn \\
0  & (\Lambda_{a,t})_\new  \nn \\
\end{array}
\right]
$$
where $(\Lambda_{a,t})_* = \text{Cov}[a_{t,*}] $ and $(\Lambda_{a,t})_\new = \text{Cov}[a_{t,\new}]$ are diagonal matrices.

\end{definition}

\begin{ass}
Assume that $L_t$ satisfies Signal Model \ref{model_lt} with
\ben
\item $0\leq c_{j,\new}\leq c$ for all $j$ (thus $r_j\leq r_{\max}:=r_0+Jc$) 

\item $\|\nu_{t}\|_{\infty} \le (1-b)\gamma_* $ 

\item
    \bea
    \max_j \max_{t \in \Ic_{j,k}} \|\nu_{t,\new}\|_{\infty} \le (1-b)\gamma_{\new,k},
    \label{ass_ainf}
    \eea
where $\gamma_{\new,k}=\min(v^{k-1} \gamma_\new,\gamma_*)$ with a $v>1$

\item
$$(1-b^2)\lambda^- \le \lambda_{\min}( \Lambda_{\nu,t} ) \le \lambda_{\max}(\Lambda_{\nu,t} ) \le (1-b^2)\lambda^+, $$
and
$${\lambda^-} \le  \lambda_{\min}(\Lambda_{a,0}) \le  \lambda_{\max}(\Lambda_{a,0}) \le (1-b^2){\lambda^+},$$
where $0<\lambda^-<\lambda^+<\infty$

\item
$(1-b^2)\lambda_{\new}^- \le  \lambda_{\min}( (\Lambda_{\nu,t})_\new ) \le \lambda_{\max}( (\Lambda_{\nu,t})_\new ) \le (1-b^2)\lambda_{\new}^+$
and \\ 
${\lambda_{\new}^-} \le  \lambda_{\min}( (\Lambda_{a,t_j})_\new ) \le \lambda_{\max}( \Lambda_{a,t_j})_\new ) \le (1-b^2){\lambda_{\new}^+} $

\een
\end{ass}

With the above assumptions, clearly, the eigenvalues of $\Lambda_{a,t}$ lie between $\lambda^-$ and  $\lambda^+$ and those of $\Lambda_{a,t,\new}$ lie between
$\lambda_\new^-$ and  $\lambda_\new^+$. 
Thus the condition numbers of any $\Lambda_{a,t}$ and of $\Lambda_{a,t,\new}$ are bounded by
$$f: = \frac{\lambda^+}{\lambda^-} \ \ \text{and} \ \ g:= \frac{\lambda_{\new}^+}{\lambda_{\new}^-}$$
respectively.
Define the following quantities for $S_t$.
\begin{definition}
Let $T_t :=\{i: \  (S_t)_i \neq 0 \}$ denote the support of $S_t$. Define
\[
S_{\min}: = \min_{t> t_\train} \min_{i \in T_t} |(S_t)_i |, \ \text{and} \ s: = \max_t |T_t|
\]
\end{definition}

\subsection{Measuring denseness of a matrix and its relation with RIC}
\label{denseness}
Before we can state the denseness assumption, we need to define the denseness coefficient.
\begin{definition}[denseness coefficient]\label{subspace_kappa}
For a matrix or a vector $B$, define
\beq 
\kappa_s(B)=\kappa_s(\Span(B)) : = \max_{|T| \le s} \|{I_T}' \mathrm{basis}(B)\|_2
\eeq
where $\|.\|_2$ is the vector or matrix $\ell_2$-norm.
\end{definition}

Clearly, $\kappa_s(B) \le 1$. First consider an $n$-length vector $B$. Then $\kappa_s$ measures the denseness (non-compressibility) of $B$. A small value indicates that the entries in $B$ are spread out, i.e. it is a dense vector. A large value indicates that it is compressible (approximately or exactly sparse). The worst case (largest possible value) is $\kappa_s(B)=1$ which indicates that $B$ is an $s$-sparse vector. The best case is $\kappa_s(B) = \sqrt{s/n}$ and this will occur if each entry of $B$ has the same magnitude. 
Similarly, for an $n \times r$ matrix $B$, a small $\kappa_s$ means that most (or all) of its columns are dense vectors.
\begin{remark}\label{kapparemark}
The following facts should be noted about $\kappa_s(.)$.
\ben
\item For a matrix $B$ of rank $r$, $\kappa_s(B)$ is an non-decreasing function of $s$ and of $r$

\item A loose bound on $\kappa_s(B)$ obtained using triangle inequality is $\kappa_s(B) \le s \kappa_1(B)$. 

\een
\label{dense_remark}
\end{remark}
%

The lemma below relates the denseness coefficient of a basis matrix $P$ to the RIC of $I-PP'$. The proof is in the Appendix. 
\begin{lem}\label{delta_kappa}
For an $n \times r$ basis matrix $P$  (i.e $P$ satisfying $P'P=I$), 
$$\delta_s(I-PP') = \kappa_s^2 (P).$$
\end{lem}
In other words, if $P$ is dense enough (small $\kappa_s$), then the RIC of $I-PP'$ is small.


In this work, we assume an upper bound on $\kappa_{2s}(P_{j})$ for all $j$, and a tighter upper bound on $\kappa_{2s}(P_{j,\new})$, i.e., there exist $\kappa_{2s,*}^+<1$ and a $\kappa_{2s,\new}^+ < \kappa_{2s,*}^+$ such that
\begin{align}
\max_{j} \kappa_{2s}(P_{j-1}) \leq \kappa_{2s,*}^+ \label{kappa plus}\\
\max_{j} \kappa_{2s}(P_{j,\new}) \leq \kappa_{2s,\new}^+ \label{kappa new plus}
\end{align}
Additionally, we also assume denseness of another matrix, $D_{j,\new,k}$, whose columns span the currently unestimated part of $\Span(P_{j,\new})$ (see Theorem \ref{thm1}).

The denseness coefficient $\kappa_s(B)$ is related to the denseness assumption required by PCP \cite{rpca}. That work uses $\kappa_1(B)$ to quantify denseness.

\section{The ReProCS algorithm and its performance guarantee}

\subsection{Recursive Projected CS}

We summarize the ReProCS algorithm in Algorithm \ref{reprocs} \cite{rrpcp_perf}. This calls the projection PCA algorithm in the subspace update step.
Given a data matrix $\mathcal{D}$, a basis matrix $P$ and an integer $r$, projection-PCA (proj-PCA) applies PCA on $\mathcal{D}_{\text{proj}}:=(I-PP')\mathcal{D}$, i.e., it computes the top $r$ eigenvectors (the eigenvectors with the largest $r$ eigenvalues) of $\frac{1}{\alpha_{\mathcal{D}}} \mathcal{D}_{\text{proj}} {\mathcal{D}_{\text{proj}}}'$. Here $\alpha_{\mathcal{D}}$ is the number of column vectors in $\mathcal{D}$. This is summarized in Algorithm \ref{algo_pPCA}.

If $P =[.]$, then projection-PCA reduces to standard PCA, i.e. it computes the top $r$ eigenvectors of $\frac{1}{\alpha_{\mathcal{D}}} \mathcal{D} {\mathcal{D}}'$.


\begin{algorithm}
\caption{projection-PCA: $Q \leftarrow \text{proj-PCA}(\mathcal{D},P,r)$}\label{algo_pPCA}
\ben
\item Projection: compute $\mathcal{D}_{\text{proj}} \leftarrow (I - P P') \mathcal{D}$
\item PCA: compute $\frac{1}{\alpha_{\mathcal{D}}}  \mathcal{D}_{\text{proj}}{\mathcal{D}_{\text{proj}}}' \overset{EVD}{=}
\left[ \begin{array}{cc}Q & Q_{\perp} \\\end{array}\right]
\left[ \begin{array}{cc} \Lambda & 0 \\0 & \Lambda_{\perp} \\\end{array}\right]
\left[ \begin{array}{c} Q' \\ {Q_{\perp}}'\\\end{array}\right]$
where $Q$ is an $n \times r$ basis matrix and  $\alpha_{\mathcal{D}}$ is the number of columns in $\mathcal{D}$.
\een
\end{algorithm}

\begin{figure*} 
\centerline{
\includegraphics[width=0.65\textwidth]{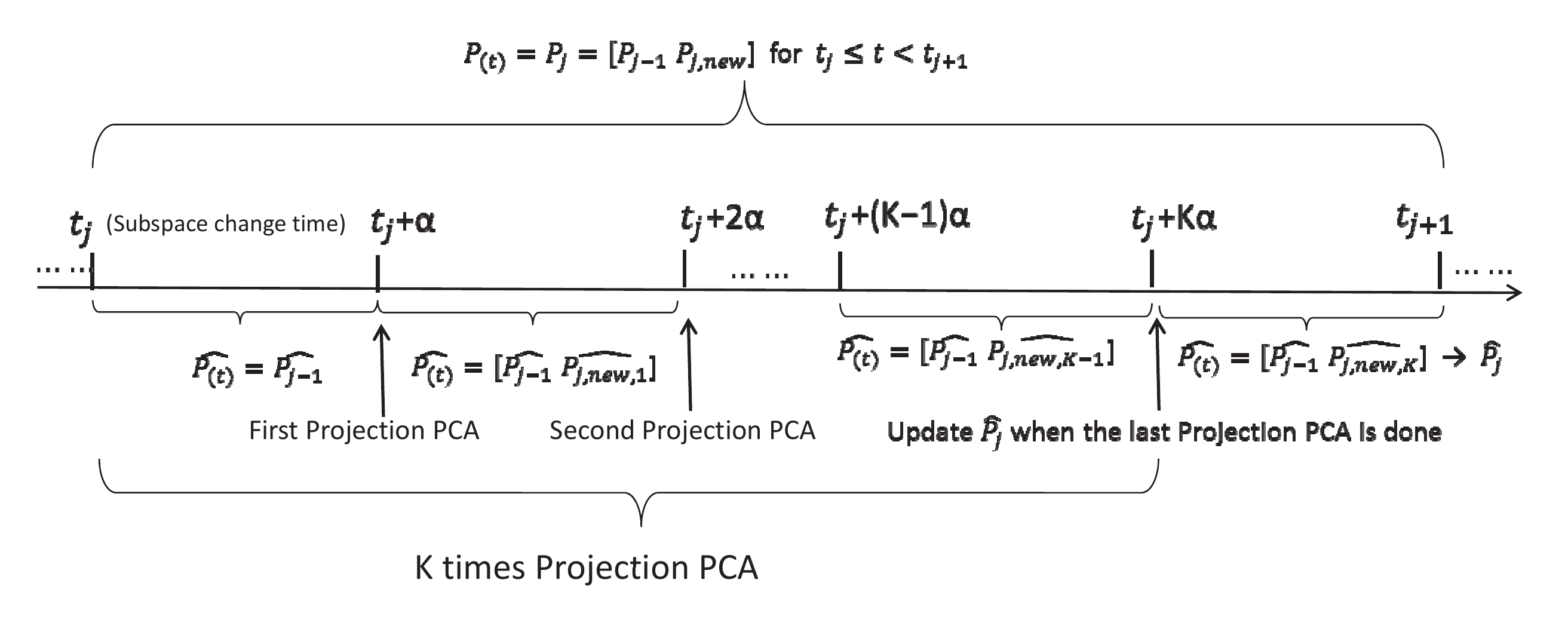}
}
\vspace{-0.2in}
\caption{\small{The K Projection PCA steps.}}
\label{proj_fig}
\vspace{-0.1in}
\end{figure*}



\begin{algorithm}[ht]
\caption{Recursive Projected CS (ReProCS)}\label{reprocs}
{\em Parameters: } algorithm parameters: $\xi$, $\omega$, $\alpha$, $K$, model parameters: $t_j$, $r_0$, $c$
\\ (set as in Theorem \ref{thm1} or as in \cite[Section IX-B]{rrpcp_perf} when the model is not known)
\\
{\em Input: } $M_t$, {\em Output: } $\Shat_t$, $\Lhat_t$, $\Phat_{(t)}$ 
\\
Initialization: Compute $\Phat_0 \leftarrow$ proj-PCA$\left( [L_{1},L_{2},\cdots,L_{t_{\train}}], [.], r_0 \right)$ and set $\Phat_{(t)} \leftarrow \Phat_0$.
\\
Let $j \leftarrow 1$, $k\leftarrow 1$.

For $t > t_{\train}$, do the following:
\ben
\item Estimate $T_t$ and $S_t$ via Projected CS:
\ben
\item \label{othoproj} Nullify most of $L_t$: compute $\Phi_{(t)} \leftarrow I-\Phat_{(t-1)} {\Phat_{(t-1)}}'$, compute $y_t \leftarrow \Phi_{(t)} M_t$
\item \label{Shatcs} Sparse Recovery: compute $\hat{S}_{t,\cs}$ as the solution of $\min_{x} \|x\|_1 \ s.t. \ \|y_t - \Phi_{(t)} x\|_2 \leq \xi$
\item \label{That} Support Estimate: compute $\hat{T}_t = \{i: \ |(\hat{S}_{t,\cs})_i| > \omega\}$
\item \label{LS} LS Estimate of $S_t$: compute $(\hat{S}_t)_{\hat{T}_t}= ((\Phi_t)_{\hat{T}_t})^{\dag} y_t, \ (\hat{S}_t)_{\hat{T}_t^{c}} = 0$
\een
\item Estimate $L_t$: $\hat{L}_t = M_t - \hat{S}_t$.
\item \label{PCA} 
Update $\Phat_{(t)}$: K Projection PCA steps (Figure \ref{proj_fig}).
\ben
\item If $t = t_j + k\alpha-1$,
\ben

\item $\Phat_{j,new,k} \leftarrow$ proj-PCA$\big(\left[\hat{L}_{t_j+(k-1)\alpha}, \dots, \hat{L}_{t_j+k\alpha-1}\right],\Phat_{j-1},c\big)$.

\item set $\Phat_{(t)} \leftarrow [\Phat_{j-1} \ \Phat_{j,\new,k}]$; increment $k \leftarrow k+1$.
\een
Else
\ben
\item set $\Phat_{(t)} \leftarrow \Phat_{(t-1)}$.
\een
\item If $t = t_j + K \alpha - 1$, then set $\Phat_{j} \leftarrow [\Phat_{j-1} \ \Phat_{j,\new,K}]$. Increment $j \leftarrow j + 1$. Reset $k \leftarrow 1$.
\een
\item Increment $t \leftarrow t + 1$ and go to step 1.
\een
\end{algorithm}
The key idea of ReProCS is as follows.
First, consider a time $t$ when the current basis matrix $P_{(t)}=P_{(t-1)}$ and this has been accurately predicted using past estimates of $L_t$, i.e. we have $\Phat_{(t-1)}$ with $\|(I -  \Phat_{(t-1)} \Phat_{(t-1)}') P_{(t)}\|_2$ small. We project the measurement vector, $M_t$, into the space perpendicular to $\Phat_{(t-1)}$ to get the projected measurement vector $y_t:= \Phi_{(t)} M_t$ where $\Phi_{(t)} = I -\Phat_{(t-1)} \Phat_{(t-1)}'$ (step 1a). Since the $n \times n$ projection matrix, $\Phi_{(t)}$ has rank $n- r_*$ where $r_*= \rank(\Phat_{(t-1)})$, therefore $y_t$ has only $n-r_*$ ``effective" measurements\footnote{i.e. some $r_*$ entries of $y_t$ are linear combinations of the other $n-r_*$ entries}, even though its length is $n$. Notice that $y_t$ can be rewritten as $y_t = \Phi_{(t)} S_t + \beta_t$ where $\beta_t: = \Phi_{(t)} L_t$. Since 
$\|(I -  \Phat_{(t-1)} \Phat_{(t-1)}') P_{(t-1)}\|_2$ is small, the projection nullifies most of the contribution of $L_t$ and so the projected noise $\beta_t$ is small.
Recovering the $n$ dimensional sparse vector $S_t$ from $y_t$ now becomes a traditional sparse recovery or CS problem in small noise \cite{feng_bresler,gorod_rao,bpdn,decodinglp,candes,donoho}. We use $\ell_1$ minimization to recover it (step 1b). If the current basis matrix $P_{(t)}$, and hence its estimate, $\Phat_{(t-1)}$, is dense enough, then, by Lemma \ref{delta_kappa}, the RIC of $\Phi_{(t)}$ is small enough. Using Theorem \ref{candes_csbound}, this  ensures that $S_t$ can be accurately recovered from $y_t$.
%
By thresholding on the recovered $S_t$, one gets an estimate of its support (step 1c). By computing a least squares (LS) estimate of $S_t$ on the estimated support and setting it to zero everywhere else (step 1d), we can get a more accurate final estimate, $\Shat_t$, as first suggested in \cite{dantzig}. This $\Shat_t$ is used to estimate $L_t$ as $\Lhat_t = M_t-\Shat_t$.  As we explain in the proof of Lemma \ref{cslem}, if $S_{\min}$ is large enough and the support estimation threshold, $\omega$, is chosen appropriately, we can get exact support recovery, i.e. $\That_t = T_t$. In this case, the error $e_t: = \Shat_t - S_t = L_t - \Lhat_t$ has the following simple expression:
\beq
e_t = I_{T_t} {(\Phi_{(t)})_{T_t}}^{\dag} \beta_t = I_{T_t} [ (\Phi_{(t)})_{T_t}'(\Phi_{(t)})_{T_t}]^{-1}  {I_{T_t}}' \Phi_{(t)} L_t
\label{etdef0}
\eeq
The second equality follows because ${(\Phi_{(t)})_T}' \Phi_{(t)} ={(\Phi_{(t)} I_T)}' \Phi_{(t)} = {I_T}' \Phi_{(t)}$ for any set $T$.

Now consider a time $t$ when $P_{(t)} = P_j = [P_{j-1}, P_{j,\new}]$ and $P_{j-1}$ has been accurately estimated but $P_\new$ has not been estimated, i.e. consider a $t \in \mathcal{I}_{j,1}$.  At this time, $\Phat_{(t-1)} = \Phat_{j-1}$ and so $\Phi_{(t)} = \Phi_{j,0}:=I - \Phat_{j-1} \Phat_{j-1}'$.  Let $r:=r_0 + (j-1)c$.
Assume that the delay between change times is large enough so that by $t=t_j$, $\Phat_{j-1}$ is an accurate enough estimate of $P_{j-1}$, i.e. $\|\Phi_{j,0} P_{j-1}\|_2 \le r \zeta \ll 1$. 
It is easy to see using Lemma \ref{hatswitch} that $\kappa_s(\Phi_{0} P_\new) \le  \kappa_s(P_\new) + r \zeta $, i.e. $\Phi_{0} P_\new$ is dense because $P_\new$ is dense and because $\Phat_{j-1}$ is an accurate estimate of $P_{j-1}$ (which is perpendicular to $P_\new$).
Moreover, using Lemma \ref{delta_kappa}, it can be shown that $\phi_0: = \max_{|T| \le s} \|[ (\Phi_{0})_{T}'(\Phi_{0})_{T}]^{-1}\|_2 \le \frac{1}{1-\delta_s(\Phi_0)} \le  \frac{1}{1- (\kappa_s(P_{j-1}) + r \zeta)^2}$. The error $e_t$ still satisfies (\ref{etdef0}) although its magnitude is not as small.
Using the above facts in (\ref{etdef0}), we get that
$$\|e_t\|_2 \le \frac{1}{1- (\kappa_s(P_{j-1}) + r \zeta)^2} [\kappa_{s}(P_\new) \sqrt{c} \gamma_\new +  r \zeta(\sqrt{r} \gamma_*  + \sqrt{c} \gamma_\new)]$$
If $\sqrt\zeta < 1/\gamma_*$, all terms containing $\zeta$ can be ignored and we get that the above is approximately upper bounded by $\frac{\kappa_{s}(P_\new)}{1- \kappa_s^2(P_{j-1})}  \sqrt{c} \gamma_\new$. Using the denseness assumption, this quantity is a small constant times $\sqrt{c} \gamma_\new$, e.g. with the numbers assumed in Theorem \ref{thm1} we get a bound of $0.18 \sqrt{c} \gamma_\new$.
Since $\gamma_\new \ll S_{\min}$ and $c$ is assumed to be small, thus, $\|e_t\|_2 =  \|S_t - \Shat_t\|_2$ is small compared with $\|S_t\|_2$, i.e. $S_t$ is recovered accurately. With each projection PCA step, as we explain below, the error $e_t$ becomes even smaller.

Since $\Lhat_t = M_t - \Shat_t$ (step 2), $e_t$ also satisfies $e_t = L_t  - \Lhat_t$. Thus, a small $e_t$ means that $L_t$ is also recovered accurately. The estimated $\Lhat_t$'s are used to obtain new estimates of $P_{j,\new}$ every $\alpha$ frames for a total of $K \alpha$ frames via a modification of the standard PCA procedure, which we call projection PCA (step 3). 
In the first projection PCA step, we get the first estimate of $P_{j,\new}$, $\Phat_{j,\new,1}$. For the next $\alpha$ frame interval, $\Phat_{(t-1)} = [\Phat_{j-1}, \Phat_{j,\new,1}]$  and so $\Phi_{(t)} = \Phi_{j,1} = I - \Phat_{j-1} \Phat_{j-1}' - \Phat_{\new,1} \Phat_{\new,1}'$. Using this in the projected CS step reduces the projection noise, $\beta_t$, and hence the reconstruction error, $e_t$, for this interval, as long as $\gamma_{\new,k}$ increases slowly enough.
Smaller $e_t$ makes the perturbation seen by the second projection PCA step even smaller, thus resulting in an improved second estimate $\Phat_{j,\new,2}$. Within $K$ updates ($K$ chosen as given in Theorem \ref{thm1}), it can be shown that both $\|e_t\|_2$ and the subspace error drop down to a constant times $\sqrt{\zeta}$. At this time, we update $\Phat_{j}$ as $\Phat_j = [\Phat_{j-1}, \Phat_{j,\new,K}]$.

\subsection{Main Result}
\label{result}
The following definition is needed for Theorem \ref{thm1}.
\begin{definition}\label{defn_alpha}
\ben
\item Let $r:= r_0 + (J-1)c$.
\item Define $\eta=\max\{\frac{c\gamma_*^2}{\lambda^+},\frac{c\gamma_{\new,k}^2}{\lambda_{\new}^+},k=1,2,\cdots,K\}$.
\item  Define $K(\zeta) := \left\lceil\frac{\log(0.85c\zeta)}{\log {0.6}} \right\rceil$
\item Define $\xi_0(\zeta)  :=  \sqrt{c} \gamma_{\new} + \sqrt{\zeta}(\sqrt{r}  +  \sqrt{c})$
\item With $K = K(\zeta)$, define
\bea
 \alpha_\add  := &&\Big\lceil (\log 61 K J + 11 \log n) \frac{8 \cdot 192^2\min(1.2^{4K} \gamma_{\new}^4, \gamma_*^4)} {\zeta^2 (\lambda^-)^2}\nn
\Big\rceil
\eea
\een
\end{definition}

\begin{theorem} \label{thm1} 
Consider Algorithm \ref{reprocs}. 
Assume that  Assumption \ref{model_lt} holds with $b\leq 0.4$.  Assume also that the initial subspace estimate is accurate enough, i.e. $\|(I - \Phat_0 \Phat_0') P_0\| \le r_0 \zeta$, with 
\[
\zeta  \leq  \min\left(\frac{10^{-4}}{r^2},\frac{1.5 \times 10^{-4}}{r^2 f},\frac{1}{r^{3}\gamma_*^2}\right) \ \text{where} \ f := \frac{\lambda^+}{\lambda^-}
\]
If the following conditions hold:
\ben
\item the algorithm parameters are set as
$\xi = \xi_0(\zeta),   7 \xi \leq \omega \leq S_{\min} - 7 \xi,  \ K = K(\zeta), \ \alpha \ge \alpha_{\text{add}}(\zeta)\geq 100$,

\item slow subspace change holds: $t_{j+1}-t_j\geq K\alpha$; (\ref{ass_ainf}) holds with $v=1.2$; and $14 \xi \le S_{\min}$,

\item denseness holds: $\max_j \kappa_{2s}(P_{j-1})\leq \kappa_{2s,*}^+=0.3$ and $\max_j \kappa_{2s}(P_{j,\new})\leq \kappa_{2s,\new}^+=0.15$
where
$$\kappa_s(B):=\max_{|T| \le s} \|{I_T}'\text{basis}(B)\|_2$$
is the denseness coefficient introduced in \cite{rrpcp_perf},

\item matrices $D_{j,\new,k}:= (I - \Phat_{j-1} \Phat_{j-1}'-\Phat_{j,\new,k} \Phat_{j,\new,k}')P_{j,\new}$ and $Q_{j,\new,k}: = (I-P_{j,\new}{P_{j,\new}}')\Phat_{j,\new,k}$ satisfy
\begin{align*}
\kappa_{s}(D_{j,\new,k}) & \le \kappa_{s}^+ := 0.15, \\
\kappa_{2s}(Q_{j,\new,k}) & \leq \tilde{\kappa}_{2s}^+ := 0.15,
\end{align*}

\item the condition number of the covariance matrix of $a_{t,\new}$ is bounded, i.e., $g \le g^+= \sqrt{2}$,

\item and $\eta \leq 1.7$,

\een

then, with probability at least $(1 -  n^{-10})$, 
\ben
\item at all times, $t$, $$\That_t = T_t \ \ \text{and}$$
\bea
\|e_t\|_2 &&= \|L_t - \hat{L}_t\|_2 = \|\hat{S}_t - S_t\|_2  \le  0.18 \sqrt{c} 0.72^{k-1} \gamma_{\new} + 1.2\sqrt{\zeta}(\sqrt{r} + 0.023 \sqrt{c}) \nn
\eea

\item  the subspace error $\SE_{(t)} := \|(I - \Phat_{(t)} \Phat_{(t)}') P_{(t)} \|_2$ satisfies
\bea
\SE_{(t)}  & \le &   \left\{  \begin{array}{ll}
(r_0 + (j-1)c) \zeta + 0.15 c \zeta  + 0.6^{k-1}   &\text{ if}  \  t \in \mathcal{I}_{j,k}, k\leq K \nn \\  
(r_0 + jc) \zeta  \ &\text{ if} \ \   t \in \mathcal{I}_{j,K+1}  
\end{array} \right. \nn \\
 & \le &  \left\{  \begin{array}{ll}
 10^{-2} \sqrt{\zeta} +  0.6^{k-1} & \ \ \text{if}  \ \ t \in \mathcal{I}_{j,k}, \ k\leq K \nn \\  
10^{-2} \sqrt{\zeta}   & \ \text{if} \ \    t \in \mathcal{I}_{j,K+1}  
\end{array} \right.
\eea

\item the error $e_t = \hat{S}_t - S_t = L_t - \hat{L}_t$ satisfies the following at various times
\bea
\|e_t\|_2  & \le &  \left\{  \begin{array}{ll}
0.18 \sqrt{c}0.72^{k-1}\gamma_{\new} + 1.2 (\sqrt{r} + 0.023 \sqrt{c})  (r_0+(j-1)c)\zeta  \gamma_*     & \ \ \text{if} \ \ t \in \mathcal{I}_{j,k}, \ k=1,2 \dots K  \nn \\ 
1.2(r_0+ j c) \zeta \sqrt{r} \gamma_*  &  \ \ \text{if} \ \ t \in \mathcal{I}_{j,K+1} 
\end{array} \right. \nn \\
 & \le &  \left\{  \begin{array}{ll}
0.18 \sqrt{c}0.72^{k-1}\gamma_{\new} + 1.2(\sqrt{r} + 0.023 \sqrt{c}) \sqrt{\zeta}  & \ \ \text{if} \ \ t \in \mathcal{I}_{j,k}, \ k=1,2 \dots K \nn \\ 
1.2 \sqrt{r} \sqrt{\zeta}  & \ \ \text{if} \ \  t \in \mathcal{I}_{j,K+1} 
\end{array} \right.
\eea


\een
\end{theorem}

{\em Proof: } We give a brief proof outline in Sec \ref{outline}. The full proof is given in Sec \ref{proofthm1}.

\subsection{Discussion}
The above result says the following. Consider Algorithm \ref{reprocs}. Assume that the initial subspace error is small enough. If the algorithm parameters are appropriately set, if slow subspace change holds, if the subspaces are dense, if the condition number of $\text{Cov}[a_{t,\new}]$ is small enough, and if the currently unestimated part of the newly added subspace is dense enough (this is an assumption on the algorithm estimates), then, w.h.p., we will get exact support recovery at all times. Moreover, the sparse recovery error will always be bounded by $0.18\sqrt{c} \gamma_\new$ plus a constant times $\sqrt{\zeta}$. Since $\zeta$ is very small, $\gamma_\new \ll S_{\min}$, and $c$ is also small, the normalized reconstruction error for recovering $S_t$ will be small at all times.
In the second conclusion, we bound the subspace estimation error, $\SE_{(t)}$.
When a subspace change occurs, this error is initially bounded by one. The above result shows that, w.h.p., with each projection PCA step, this error decays exponentially and falls below $0.01\sqrt{\zeta}$ within $K$ projection PCA steps. The third conclusion shows that, with each projection PCA step, w.h.p., the sparse recovery error as well as the error in recovering $L_t$ also decay in a similar fashion.

The above result allows the $a_t$'s, and hence the $L_t$'s, to be correlated over time; it models the correlation using an AR model which is a frequently used practical model. Even with this more general model as long as the AR parameter, $b \le 0.4$, we are able to get almost exactly the same result as that of Qiu et al \cite[Theorem 4.1]{rrpcp_perf}. The $\alpha$  needed is a little larger. Also, the only extra assumption needed is a small enough upper bound on $\eta$ which is the ratio of the maximum magnitude entry of any $\nu_t$ to the maximum variance. This is true for many types of probability distributions. For example if the $i^{th}$ entry of $\nu_t$ is $\pm q_i$ with equal probability independent of all others then $\eta = 1$. If each entry is zero mean uniform distributed (with different spreads) then $\eta = 3$.


Like \cite{rrpcp_perf}, we still need a denseness assumption on $D_{\new,k}$ and $Q_{\new,k}$ both of which are functions of algorithm estimates $\Phat_{j-1}$ and $\Phat_{j,\new,k}$. Because of this, our result cannot be interpreted as a correctness result but only a useful step towards it. As explained in \cite{rrpcp_perf}, from simulation experiments, this assumption does hold whenever there is {\em some} support changes every few frames. In future work, we hope to be able to replace it with an assumption on the support change of $S_t$'s.

Also, like \cite{rrpcp_perf},  the above result analyzes an algorithm that assumes knowledge of the model parameters $c$, $\gamma_\new$ and the subspace change times $t_j$. Requiring kmowledge of $c$ and $\gamma_\new$ is not very restrictive. However it also needs to know the subspace change times $t_j$ and this is somewhat restrictive. One approach to try to remove this requirement is explained in \cite{rrpcp_perf}. 

As explained in \cite{rrpcp_perf}, under slow subspace change, it is quite valid to assume that the condition number of the new directions, $g$, is bounded, in fact if at most one new direction could get added, i.e. if $c=1$, then we would always have $g=1$. On the other hand, notice that we do not need any bound on $f$. This is important and needed to allow $\E[\|L_t\|_2^2]$ to be large ($L_t$ is the large but structured noise in case $S_t$ is the signal of interest) while still allowing slow subspace change (needs small $\gamma_\new$ and hence small $\lambda^- \le \gamma_\new$). Notice that $\E[\|L_t\|_2^2] \le r \lambda^+$.

\subsection{Proof Outline} \label{outline}

The first step in the proof is to analyze the projected sparse recovery step and show exact support recovery conditioned on the fact that the subspace has been accurately recovered in the previous projection-PCA interval. Exact support recovery along with the LS step allow us to get an exact expression for the recovery error in estimating $S_t$ and hence also for that of $L_t$. This exact expression is the key to being able to analyze the subspace recovery.

For subspace recovery, the first step involves bounding the subspace recovery error in terms of sub-matrices of the true matrix, $\sum_t \Phi_{(t)}\Lhat_t \Lhat_t' \Phi_{(t)}'$, and the perturbation in it, $\sum_t \Phi_{(t)}(\Lhat_t \Lhat_t' - L_t L_t') \Phi_{(t)}'$, using the famous $\sin \theta$ theorem \cite{davis_kahan}. This result bounds the error in the eigenvectors of a matrix perturbed by a Hermitian perturbation. The second step involves obtaining high probability bounds on each of the terms in this bound using the matrix Azuma inequality  \cite{tropp2012user}. The third step involves using the assumptions of the theorem to show that this bound decays roughly exponentially with $k$ and finally falls below $c \zeta$ within $K$ proj-PCA steps.%

The most important difference w.r.t. the result of \cite{rrpcp_perf} is the following. Define the random variable $X_{j,k}:=[\nu_0, \nu_1, \cdots, \nu_{t_j+k\alpha-1}]$.
In the second step, we need to bound the minimum or maximum singular values of sub-matrices of terms of the form $\sum_t f_1(X_{j,k-1})a_{t}a_{t}'f_2(X_{j,k-1})$ where $f_1(.),f_2(.)$ are functions of the random variable $X_{j,k-1}$. In Qiu et al \cite{rrpcp_perf}, one could use a simple corollary of the matrix Hoeffding inequality \cite{tropp2012user} to do this because there, the terms of this summation were conditionally independent given $X_{j,k-1}$. However, here they are not. We instead need to first use the AR model to rewrite things in terms of sub-matrices of $\sum_t f_1(X_{j,k-1}) \nu_{t} f_3(\nu_{0}, \nu_{1}, \cdots, \nu_{t-1})' f_2(X_{j,k-1})$. Notice that even now, the terms of this summation are not conditionally independent given $X_{j,k-1}$. However, conditioned on  $X_{j,k-1}$, this term is now in a form for which the matrix Azuma inequality \cite{tropp2012user} can be used.

Notice that the ReProCS algorithm does not need knowledge of $b$. If $b$ were known, one could modify the algorithm to do proj-PCA on $(\Lhat_t - b \Lhat_{t-1})$'s. With this one could use the exact same proof strategy as in \cite{rrpcp_perf}.

\section{Simulation}

In this section, we compare ReProCS with PCP using simulated data that satisfies the assumed signal model. The data was generated as explained in \cite[Section X-C]{rrpcp_perf} except that here we generate correlated $a_t$'s using $a_t=ba_{t-1}+\nu_t$ with $b=0.5$ and with $\nu_{t,*}$ being uniformly distributed between $[-\gamma_*, \gamma_*]$ and $\nu_{t,\new}$ uniformly distributed between $[-\gamma_{\new,k}, \gamma_{\new,k}]$. Also we set $t_{\text{train}}=40, S_{\min}=2, S_{\max}=3, s=7, r_0=12$, $n=200$, $J=2$,  $v=1.1$ and the support $T_t$ was constant for every set of $50$ frames and then changed by 1 index. Other parameters were the same as those in \cite[Section X-C]{rrpcp_perf}.
By running 100 Monte Carlo simulations, we got the result shown in Figure \ref{simu_fig}. 

\begin{figure}
\centering
\includegraphics[width=0.5\textwidth,height=0.3\textwidth]{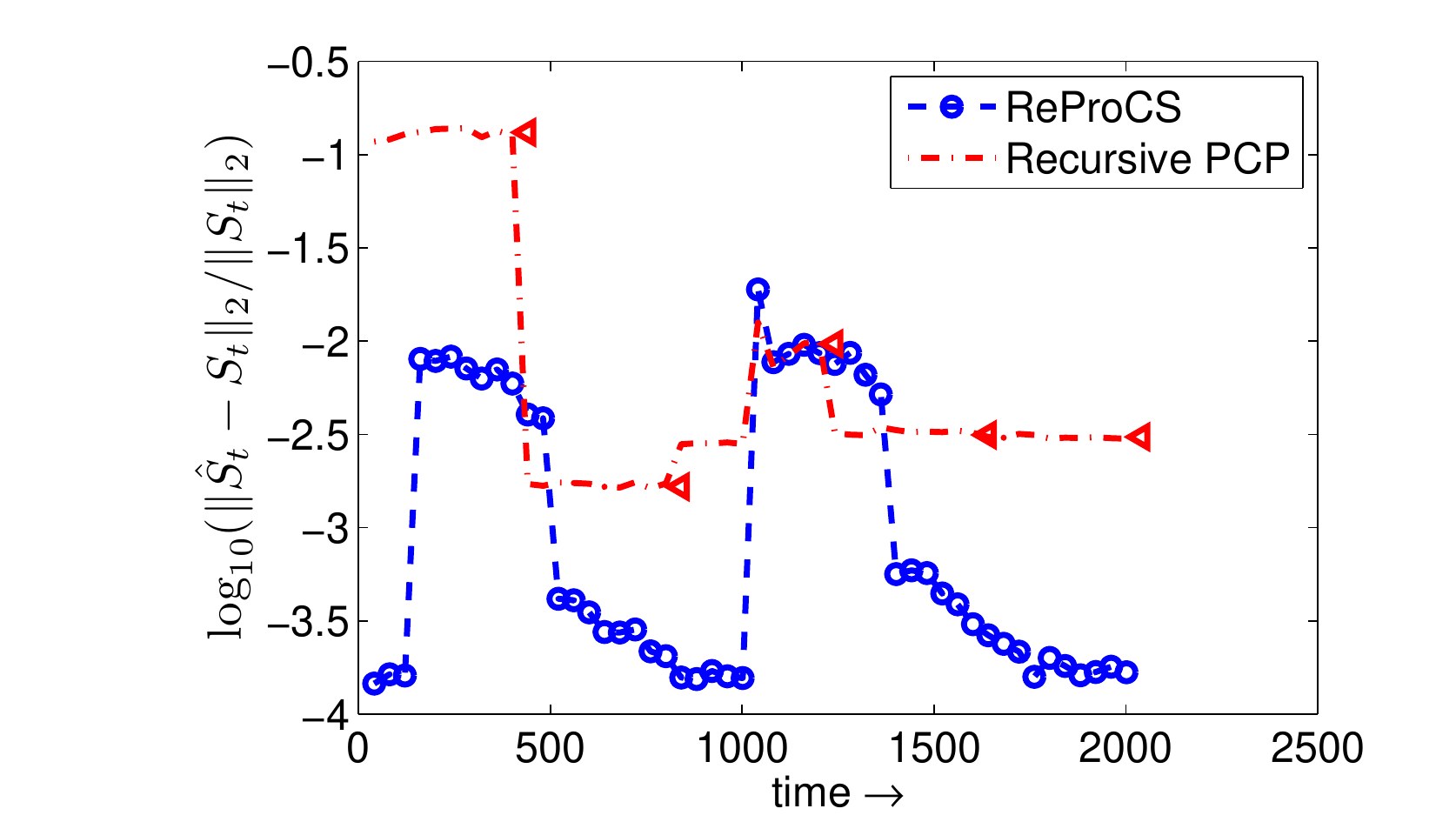}
\vspace{-0.2in}
\caption{Comparing recovery error of PCP implemented at the time instants shown by the triangles and of ReProCS.}
\vspace{-0.2in}
\label{simu_fig}
\end{figure}

\section{Proof of Theorem \ref{thm1}}\label{proofthm1}
\vspace{-0.25cm}
Here we first list some definitions used in the proof.
\begin{definition}
We define the noise seen by the sparse recovery step at time $t$ as
$$\beta_t: = \|(I - \Phat_{(t-1)} \Phat_{(t-1)}') L_t\|_2.$$
\end{definition}

\begin{definition}
\label{def_SEt}
We define the subspace estimation errors as follows. Recall that $\Phat_{j,\new,0}=[.]$ (empty matrix).
\bea  
&& \SE_{(t)} := \|(I - \Phat_{(t)} \Phat_{(t)}') P_{(t)} \|_2, \nn \\ 
&& \zeta_{j,*} := \|(I - \Phat_{j-1} \Phat_{j-1}') P_{j-1}\|_2 \nn \\
&& \zeta_{j,k} := \|(I - \Phat_{j-1} \Phat_{j-1}' - \Phat_{j,\new,k} \Phat_{j,\new,k}') P_{j,\new}\|_2 \nn 
\eea
\end{definition}

\begin{remark}\label{zetastar}
Recall from the model given in Sec \ref{sig_mod} and from Algorithm \ref{reprocs} that
\begin{enumerate}
\item $\Phat_{j,\new,k}$ is orthogonal to $\Phat_{j-1}$, i.e. $\Phat_{j,\new,k}'\Phat_{j-1}=0$
\item $\Phat_{j-1} := [\Phat_{0}, \Phat_{1,\new,K}, \dots \Phat_{j-1,\new,K}]$ and $P_{j-1}: = [P_0, P_{1,\new}, \dots P_{j-1,\new}]$
\item  for $t \in \mathcal{I}_{j,k+1}$, $\Phat_{(t)} = [\Phat_{j-1}, \Phat_{j,\new,k}]$ and $P_{(t)} = P_j = [P_{j-1}, P_{j,\new}]$. 
\item  $\Phi_{(t)} := I - \Phat_{(t-1)} \Phat_{(t-1)}'$
\end{enumerate}

From Definition \ref{def_SEt} and the above, it is easy to see that
\begin{enumerate}
\item $\zeta_{j,*} \le \zeta_{1,*} + \sum_{j'=1}^{j-1} \zeta_{j',K}$
\item $\SE_{(t)}  \le \zeta_{j,*} + \zeta_{j,k} \le \zeta_{1,*} + \sum_{j'=1}^{j-1} \zeta_{j',K} + \zeta_{j,k}$ \; for \ $t \in \mathcal{I}_{j,k+1}$.
\end{enumerate}
\end{remark}

\begin{definition}\label{defn_Phi}
Define the following
\ben
\item $\Phi_{j,k}$, $\Phi_{j,0}$ and $\phi_k$ 
\ben
\item $\Phi_{j,k} := I-\Phat_{j-1} {\Phat_{j-1}}' - \Phat_{j,\new,k} {\Phat_{j,\new,k}}'$ is the CS matrix for $t \in \mathcal{I}_{j,k+1}$, i.e. $\Phi_{(t)} = \Phi_{j,k}$ for this duration.

\item $\Phi_{j,0} := I-\Phat_{j-1} {\Phat_{j-1}}'$ is the CS matrix for $t \in \mathcal{I}_{j,1}$, i.e. $\Phi_{(t)} = \Phi_{j,0}$ for this duration. $\Phi_{j,0}$ is also the matrix used in all of the projection PCA steps for $t \in [t_j, t_{j+1}-1]$.

\item $\phi_k := \max_j \max_{T:|T|\leq s}\|({(\Phi_{j,k})_T}'(\Phi_{j,k})_T)^{-1}\|_2$. It is easy to see that $\phi_k \le \frac{1}{1-\max_j \delta_s(\Phi_{j,k})}$ \cite{decodinglp}.

\een
\item $D_{j,\new,k}$, $D_{j,\new}$ and $D_{j,*}$
\ben
\item $D_{j,\new,k} := \Phi_{j,k} P_{j,\new}$. $\Span(D_{j,\new,k})$ is the unestimated part of the newly added subspace for any $t \in \mathcal{I}_{j,k+1}$. 
\item $D_{j,\new} := D_{j,\new,0} = \Phi_{j,0} P_{j,\new}$. $\Span(D_{j,\new})$ is interpreted similarly for any $t \in \mathcal{I}_{j,1}$.
\item $D_{j,*,k} := \Phi_{j,k} P_{j-1}$. $\Span(D_{j,*,k})$ is the unestimated part of the existing subspace for any $t \in \mathcal{I}_{j,k}$
\item $D_{j,*} := D_{j,*,0} = \Phi_{j,0} P_{j-1}$. $\Span(D_{j,*,k})$ is interpreted similarly for any $t \in \mathcal{I}_{j,1}$
\item Notice that $\zeta_{j,0} = \|D_{j,\new}\|_2$, $\zeta_{j,k} = \|D_{j,\new,k}\|_2$, $\zeta_{j,*} = \|D_{j,*}\|_2$. Also, clearly, $\|D_{j,*,k}\|_2 \le \zeta_{j,*}$. 
\een
\een
\end{definition}

\begin{definition}
\label{defHk}\
\ben
\item Let $D_{j,\new} \overset{QR}{=} E_{j,\new} R_{j,\new}$ denote its QR decomposition. Here $E_{j,\new}$ is a basis matrix and $R_{j,\new}$ is upper triangular.

\item Let $E_{j,\new,\perp}$ be a basis matrix for the orthogonal complement of $\Span(E_{j,\new})=\Span(D_{j,\new})$. To be precise, $E_{j,\new,\perp}$ is a $n \times (n-c_{j,\new})$ basis matrix that satisfies $E_{j,\new,\perp}'E_{j,\new}=0$. 

\item Using $E_{j,\new}$ and $E_{j,\new,\perp}$, define $A_{j,k}$, $A_{j,k,\perp}$, $H_{j,k}$, $H_{j,k,\perp}$ and $B_{j,k}$ as
\bea
A_{j,k} &:=& \frac{1}{\alpha} \sum_{t \in \mathcal{I}_{j,k}} {E_{j,\new}}' \Phi_{j,0} L_t {L_t}' \Phi_{j,0} E_{j,\new} \nn \\
A_{j,k,\perp} &:=& \frac{1}{\alpha} \sum_{t \in \mathcal{I}_{j,k}} {E_{j,\new,\perp}}' \Phi_{j,0} L_t {L_t}' \Phi_{j,0} E_{j,\new,\perp} \nn \\
H_{j,k} &:=& \frac{1}{\alpha}\sum_{t \in \mathcal{I}_{j,k}} {E_{j,\new}}' \Phi_{j,0} (e_t {e_t}' -L_t {e_t}' - e_t {L_t}') \Phi_{j,0} E_{j,\new} \nn\\
H_{j,k,\perp} &:=& \frac{1}{\alpha} \sum_{t \in \mathcal{I}_{j,k}} {E_{j,\new,\perp}}'\Phi_{j,0} (e_t {e_t}' - L_t {e_t}' - e_t {L_t}') \Phi_{j,0} E_{j,\new,\perp} \nn \\
B_{j,k} &:=& \frac{1}{\alpha}\sum_{t \in \mathcal{I}_{j,k}} {E_{j,\new,\perp}}'\Phi_{j,0} \Lhat_t \Lhat_t' \Phi_{j,0} E_{j,\new}= \frac{1}{\alpha}\sum_{t \in \mathcal{I}_{j,k}} {E_{j,\new,\perp}}'\Phi_{j,0} (L_t-e_t)({L_t}'-{e_t}')\Phi_{j,0} E_{j,\new} \nn
\eea

\item Define
\bea
&&\mathcal{A}_{j,k} := \left[ \begin{array}{cc} E_{j,\new} & E_{j,\new,\perp} \\ \end{array} \right]
\left[\begin{array}{cc} A_{j,k} \ & 0 \ \\ 0 \ & A_{j,k,\perp}  \\ \end{array} \right]
\left[ \begin{array}{c} {E_{j,\new}}' \\ {E_{j,\new,\perp}}' \\ \end{array} \right]\nn\\
&&\mathcal{H}_{j,k} := \left[ \begin{array}{cc} E_{j,\new} & E_{j,\new,\perp} \\ \end{array} \right]
\left[\begin{array}{cc} H_{j,k} \ & {B_{j,k}}' \ \\ B_{j,k} \ &  H_{j,k,\perp} \\ \end{array} \right]
\left[ \begin{array}{c} {E_{j,\new}}' \\ {E_{j,\new,\perp}}' \\ \end{array} \right] \nn
\eea

\item From the above, it is easy to see that $$\mathcal{A}_{j,k} + \mathcal{H}_{j,k} =\frac{1}{\alpha} \sum_{t \in \mathcal{I}_{j,k}} \Phi_{j,0} \hat{L}_t {\hat{L}_t}' \Phi_{j,0}.$$

\item Recall from Algorithm \ref{reprocs} that $\mathcal{A}_{j,k} + \mathcal{H}_{j,k} \overset{EVD}{=} \left[ \begin{array}{cc} \Phat_{j,\new,k} & \Phat_{j,\new,k,\perp} \\ \end{array} \right]
\left[\begin{array}{cc} \Lambda_k \ & 0 \ \\ 0 \ & \ \Lambda_{k,\perp} \\ \end{array} \right]
\left[ \begin{array}{c} \Phat_{j,\new,k}' \\ \Phat_{j,\new,k,\perp}' \\ \end{array} \right] $ is the EVD of $\mathcal{A}_{j,k} + \mathcal{H}_{j,k}$. Here $\Phat_{j,\new,k}$ is a $n \times c_{j,\new}$ basis matrix.

\item Using the above, $\mathcal{A}_{j,k} + \mathcal{H}_{j,k}$ can be decomposed in two ways as follows:
\begin{align*}
\mathcal{A}_{j,k} + \mathcal{H}_{j,k}
&= \left[ \begin{array}{cc} \Phat_{j,\new,k} & \Phat_{j,\new,k,\perp} \\ \end{array} \right]
\left[\begin{array}{cc} \Lambda_k \ & 0 \ \\ 0 \ & \ \Lambda_{k,\perp} \\ \end{array} \right] \left[ \begin{array}{c} \Phat_{j,\new,k}' \\ \Phat_{j,\new,k,\perp}' \\ \end{array} \right]
\\
&= \left[ \begin{array}{cc} E_{j,\new} & E_{j,\new,\perp} \\ \end{array} \right]
\left[\begin{array}{cc} A_{j,k} + H_{j,k} \ & B_{j,k}' \ \\ B_{j,k} \ & A_{j,k,\perp} + H_{j,k,\perp}  \\ \end{array} \right]
\left[ \begin{array}{c} {E_{j,\new}}' \\ {E_{j,\new,\perp}}' \\ \end{array} \right]
\end{align*}
\een

\end{definition}

\begin{remark}
Thus, from the above definition,
$\mathcal{H}_{j,k}  = \frac{1}{\alpha} [\Phi_0 \sum_t (-L_t e_t' - e_t L_t' +  e_t e_t') \Phi_0 + F + F']$ where $F:=E_{\new,\perp} E_{\new,\perp}'\Phi_0 \sum_t L_t L_t' \Phi_0 E_\new E_\new' =E_{\new,\perp} E_{\new,\perp}' (D_{*,k-1} a_{t,*})(D_{*,k-1} a_{t,*} + D_{\new,k-1} a_{t,\new})' E_\new E_\new'$. Since $\E[a_{t,*} a_{t,\new}']=0$, $\|\frac{1}{\alpha}F\|_2 \lesssim r^2 \zeta^2 \lambda^+$ w.h.p. Recall $\lesssim$ means (in an informal sense) that the RHS contains the dominant terms in the bound.
\end{remark}

\begin{definition}
\label{kappaplus}
In the sequel, we let
\ben
\item $\kappa_{s,*} := \max_j \kappa_s(P_{j-1})$, $\kappa_{s,\new} := \max_j \kappa_s(P_{j,\new})$, $\kappa_{s,k}:= \max_j \kappa_s (D_{j,\new,k})$, $\tilde{\kappa}_{s,k} := \max_j \kappa_s((I-P_{j,\new}{P_{j,\new}}') \Phat_{j,\new,k})$,
\item $\kappa_{2s,*}^+ := 0.3$, $\kappa_{2s,\new}^+ := 0.15$, ${\kappa}_{s}^+ := 0.15$, $\tilde{\kappa}_{2s}^+ := 0.15$ and $g^+ := \sqrt{2}$ are the upper bounds assumed in Theorem \ref{thm1} on $\max_j \kappa_{2s}(P_j)$, $\max_j \kappa_{2s}(P_{j,\new})$, $\max_j \max_k \kappa_{s}(D_{j,\new,k})$, $\max_j \kappa_{2s}(Q_{j,\new,k})$ and $g$ respectively.
\item $\phi^+ := 1.1735$ We see later that this is an upperbound on $\phi_k$ under the assumptions of Theorem \ref{thm1}.
\item $\gamma_{\new,k} := \min (1.2^{k-1} \gamma_{\new}, \gamma_*)$
\\ (recall that the theorem assumes  $\max_{j}\max_{t\in\mathcal{I}_{j,k}} \|a_{t,\new}\|_{\infty} \leq \gamma_{\new,k}$)
\item $P_{j,*}: = P_{j-1}$ and $\Phat_{j,*}:= \Phat_{j-1}$ (see Remark \ref{remove_j}).
\een
\end{definition}

\begin{remark}
Notice that the subscript $j$ always appears as the first subscript, while $k$ is the last one. At many places in this paper, we remove the subscript $j$ for simplicity. Whenever there is only one subscript, it refers to the value of $k$, e.g., $\Phi_0$ refers to $\Phi_{j,0}$, $\Phat_{\new,k}$ refers to  $\Phat_{j,\new,k}$.  Also, $P_*: = P_{j-1}$ and $\Phat_*: = \Phat_{j-1}$.
\label{remove_j}
\end{remark}

\begin{definition}
\label{zetakplus}
Define the following:
\begin{enumerate}
\item  $\zeta_{*}^+ := r \zeta$ \ (We note that $\zeta_*^+ = (r_0 + (j-1)c)\zeta$ will also work.)

\item Define the sequence $\{{\zeta_{k}}^+\}_{k=0,1,2,\cdots K}$ recursively as follows
\begin{align}
\zeta_0^+ & := 1 \nn  \\
 \zeta_k^+ & :=\frac{b_{\Hc} + 0.125 c \zeta\lambda^-}{b_{A_k} - b_{A_{k,\perp}} - b_{\Hc} -0.25 c \zeta\lambda^-} \; \text{ for} \ k \geq 1,
\end{align}
\end{enumerate}
where $b_{A_k}, b_{A_{k,\perp}}, b_{\Hc} $ are defined in (\ref{A_k_bound}), (\ref{A_k_perp_bound}) and (\ref{H_k_bound}) respectively.
\end{definition}

As we will see, $\zeta_{*}^+$ and $\zeta_k^+$ are the high probability upper bounds on $\zeta_{j,*}$ and $\zeta_{j,k}$ (defined in Definition \ref{def_SEt}) under the assumptions of Theorem \ref{thm1}.

\begin{definition}
Define the random variable $X_{j,k} := \{\nu_1,\nu_2,\cdots,\nu_{t_j+k\alpha-1}\}$.
\end{definition}
Recall that the $\nu_t$'s are mutually independent over $t$.

\begin{definition}

Define the set $\check{\Gamma}_{j,k}$ as follows:
\begin{align*}
\check{\Gamma}_{j,k} &:= \{ X_{j,k} : \zeta_{j,k} \leq \zeta_{k}^+ \text{ and } \hat{T}_t = T_t  \text{ and } e_t \text{ satisfies } (\ref{etdef0}) \text{ for all } t \in \mathcal{I}_{j,k} \} \\
\check{\Gamma}_{j,K+1} &:= \{ X_{j+1,0} :  \hat{T}_t = T_t  \text{ and } e_t \text{ satisfies } (\ref{etdef0}) \text{ for all } t \in \mathcal{I}_{j,K+1} \}
\end{align*}

\end{definition}

\begin{definition}
Recursively define the sets $\Gamma_{j,k}$ as follows:
\begin{align*}
\Gamma_{1,0} &:= \{ X_{1,0} : \zeta_{1,*} \leq r \zeta \text{ and } \hat{T}_t = T_t \text{ and } e_t \text{ satisfies } ( \ref{etdef0} ) \text{ for all } t \in [ t_{\mathrm{train} +1} : t_1 -1 ] \} \\
\Gamma_{j,k} &:= \Gamma_{j,k-1} \cap \check{\Gamma}_{j,k} \quad k = 1, 2, \dots , K, j = 1, 2, \dots, J \\
\Gamma_{j+1, 0} &:= \Gamma_{j,K} \cap \check{\Gamma}_{j,K+1} \quad j = 1, 2, \dots, J
\end{align*}
\end{definition}

\subsection{Main Steps for Theorem \ref{thm1}}

The proof of Theorem \ref{thm1} essentially follows from two main lemmas, \ref{expzeta} and \ref{mainlem}.  Lemma \ref{expzeta} gives an exponentially decaying upper bound on $\zeta_k^+$ defined in Definition \ref{zetakplus}. $\zeta_k^+$ will be shown to be a high probability upper bound for $\zeta_k$ under the assumptions of the Theorem.  Lemma \ref{mainlem} says that conditioned on $X_{j,k-1}\in\Gamma_{j,k-1}$, $X_{j,k}$ will be in $\Gamma_{j,k}$ w.h.p..  In words this says that if, during the time interval $\mathcal{I}_{j,k-1}$, the algorithm has worked well (recovered the support of $S_t$ exactly and recovered the background subspace with subspace recovery error below $\zeta_{k-1}^+ + \zeta_*^+$), then it will also work well in $\mathcal{I}_{j,k}$ w.h.p.. The proof of Lemma \ref{mainlem} requires two lemmas: one for the projected CS step and one for the projection PCA step of the algorithm.  These are lemmas \ref{cslem} and \ref{zetak}.  The proof Lemma \ref{cslem}
follows using Lemmas \ref{expzeta}, \ref{delta_kappa}, \ref{hatswitch}, the CS error bound (Theorem \ref{candes_csbound}), and some straightforward steps.
The proof of Lemma \ref{zetak} is longer and uses a lemma based on the $\sin\theta$ and Weyl theorems (Theorems \ref{sin_theta} and \ref{weyl}) to get a bound on $\zeta_k$.  From here we use the matrix Azuma inequalities (Corollaries \ref{azuma_nonzero} and \ref{azuma_rec}) to bound each of the terms in the bound on $\zeta_k$ to finally show that, conditioned on $\Gamma_{j,k-1}^e$ $\zeta_k \le \zeta_k^+$ w.h.p.. These are Lemmas \ref{zetakbnd} and \ref{termbnds}.

\section{Main Lemmas and Proof of Theorem \ref{thm1}} \label{mainlemmas}

Recall that when there is only one subscript, it refers to the value of $k$ (i.e. $\zeta_k = \zeta_{j,k}$).

\begin{lem}[Exponential decay of $\zeta_{k}^+$]  
\label{expzeta}
Pick $\zeta$ as given in Theorem \ref{thm1}. Assume that the six conditions of Theorem \ref{thm1} hold. Define the series ${\zeta_{k}}^+$ as in Definition \ref{zetakplus}.
Then,
\ben
\item $\zeta_0^+=1, \zeta_1^+=0.5688, \zeta_2^+=0.3568$,  $\zeta_k^+ \le \zeta_{k-1}^+ \le 0.3568$ for all $k \ge 3$. 
\item $\zeta_k^+ \le 0.6^{k} + 0.15c\zeta$ for all $k \ge 0$ 
\item $b_{A_k}  - b_{A_{k,\perp}} - b_{\Hc} - 0.25c\zeta\lambda^-> 0$ for all $k \ge 1$.
\een
\end{lem}

We will prove this lemma in Section \ref{pfoflem1}.

\begin{lem}\label{mainlem}
Assume that all the conditions of Theorem \ref{thm1} hold.  Also assume that $\mathbf{P}(\Gamma^e_{j,k-1})>0.$
Then
\[
\mathbf{P}(\Gamma^e_{j,k}|\Gamma^e_{j,k-1}) \geq p_k(\alpha,\zeta) \geq p_K(\alpha,\zeta) \enspace \text{ for all } k = 1, 2, \ldots, K,
\]
where $p_k(\alpha,\zeta)$ is defined in equation \eqref{pk}.
\end{lem}

\begin{remark}\label{Gamma_rem}
Under the assumptions of Theorem \ref{thm1}, it is easy to see that the following holds.
 For any $k=1,2 \dots K$, $\Gamma_{j,k}^e$ implies that  $\zeta_{j,*} \le \zeta_{*}^+$
\end{remark}

From the definition of $\Gamma_{j,k}^e$, $\zeta_{j',K}\leq \zeta_{K}^+$ for all $j'\leq j-1$. By Lemma \ref{expzeta} and the definition of $K$ in Definition \ref{defn_alpha}, $\zeta_{K}^+\leq 0.6^K + 0.15c\zeta\leq c\zeta$ for all $j'\leq j-1$. Using Remark \ref{zetastar}, $\zeta_{j,*} \leq \zeta_1^* + \sum_{j'=1}^{j-1} \zeta_{j',K} \leq r_0\zeta + (j-1)c\zeta \leq \zeta_*^+$.

\begin{proof}[Proof of Theorem \ref{thm1}]

The theorem is a direct consequence of Lemmas \ref{expzeta}, \ref{mainlem}, and Lemma \ref{subset_lem}.

Notice that $\Gamma_{j,0}^e \supseteq  \Gamma_{j,1}^e \supseteq \dots \supseteq   \Gamma_{j,K,0}^e \supseteq  \Gamma_{j+1,0}^e$. Thus, by Lemma \ref{subset_lem}, $\mathbf{P}(\Gamma_{j+1,0}^e | \Gamma_{j,0}^e) = \mathbf{P}(\Gamma_{j+1,0}^e | \Gamma_{j,K}^e) \prod_{k=1}^{K} \mathbf{P}(\Gamma_{j,k}^e | \Gamma_{j,k-1}^e)$ and $\mathbf{P}(\Gamma_{J+1,0} | \Gamma_{1,0}) = \prod_{j=1}^{J}  \mathbf{P}(\Gamma_{j+1,0}^e | \Gamma_{j,0}^e)$.

Using Lemmas \ref{mainlem}, and the fact that $p_k(\alpha,\zeta) \geq p_K(\alpha,\zeta)$ (see their respective definitions in Lemma \ref{termbnds} and equation \eqref{pk}), we get $\mathbf{P}(\Gamma_{J+1,0}^e| \Gamma_{1,0,0}) \geq {p}_K(\alpha,\zeta)^{KJ}$.
Also, $\mathbf{P}(\Gamma_{1,0}^e)=1$. This follows by the assumption on $\hat{P}_0$ and Lemma \ref{cslem}. Thus, $\mathbf{P}(\Gamma_{J+1,0}^e) \geq {p}_K(\alpha,\zeta)^{KJ}$.

Using the definition of $\alpha_\add$ we get that $\mathbf{P}(\Gamma_{J+1,0}^e) \geq {p}_K(\alpha,\zeta)^{KJ} \geq 1- n^{-10}$ whenever $\alpha \geq \alpha_{\add}$.

The event $\Gamma_{J+1,0}^e$ implies that $\That_t=T_t$ and $e_t$ satisfies (\ref{etdef0}) for all $t < t_{J+1}$. Using Remarks \ref{zetastar} and \ref{Gamma_rem}, $\Gamma_{J+1,0}^e$ implies that all the bounds on the subspace error hold. Using these, $\|a_{t,\new}\|_2 \le \sqrt{c} \gamma_{\new,k}$, and $\|a_t\|_2 \le \sqrt{r} \gamma_*$, $\Gamma_{J+1,0}^e$ implies that all the bounds on $\|e_t\|_2$ hold (the bounds are obtained in Lemma \ref{cslem}).

Thus, all conclusions of the the result hold w.p. at least $1- n^{-10}$.
\end{proof}

\section{Proof of Lemma \ref{expzeta} } \label{pfoflem1}

\begin{proof}
By the high probability bounds on $\lambda_{\min}(A_k)$, (\ref{A_k_bound}), $\|A_{k,\perp}\|$, (\ref{A_k_perp_bound}), $\|H_k\|$, (\ref{H_k_bound}) and Assumption \ref{model_lt}(3g), we have
\bea
b_{A_k}\geq&& \tb_{A_k} = (1-(\zeta_*^+)^2)(1-\frac{b^2-b^{2\alpha+2}}{100(1-b^2)})\lambda_{\new,k}^--2\frac{b^2(1-b^{2\alpha})}{1-b^2}\zeta_*^+\sqrt{\frac{r}{c}}\eta\lambda^+\nn\\
b_{A_{k,\perp}}\leq&& \tb_{A_{k,\perp}} =  (\zeta_*^+)^2\lambda^+ +2\frac{b^2(1-b^{2\alpha})}{1-b^2}(\zeta_*^+)^2\frac{r}{c}\eta\lambda^+\nn\\
b_{\Hc}\leq&& \tb_{\Hc}= (\phi^+)^2(\zeta_*^+)^2\lambda^++(\phi^+)^2(\kappa_s^+\zeta_{k-1}^+)^2\lambda_{\new,k}^+ +2\frac{b^2(1-b^{2\alpha})}{1-b^2}\sqrt{\frac{r}{c}} \eta\lambda^+\zeta_*^+\zeta_{k-1}^+\kappa_s^+(\phi^+)^2 + \frac{b^2(1-b^{2\alpha})}{1-b^2}\eta\frac{r}{c}\lambda^+(\phi^+\zeta_*^+)^2 \nn\\&& + \frac{b^2(1-b^{2\alpha})}{1-b^2}\eta\lambda_{\new,k}^+(\phi^+\kappa_s^+\zeta_{k-1}^+)^2 + 2\phi^+\kappa_s^+\frac{(\zeta_*^+)^2}{\sqrt{1-(\zeta_*^+)^2}}(\lambda^++\frac{b^2(1-b^{2\alpha})}{1-b^2}\eta\frac{r}{c}\lambda^+)+ \nn\\ && 2\phi^+\max\{0.15c\zeta,\zeta_{k-1}^+\}\frac{(\kappa_s^+)^2}{\sqrt{1-(\zeta_*^+)^2}}\bigg(\lambda_{\new,k}^+ +\frac{b^2(1-b^{2\alpha})}{1-b^2}\eta\lambda_{\new,k}^+\bigg) +  4\frac{b^2(1-b^{2\alpha})}{1-b^2}\eta\frac{r}{c}\lambda^+\zeta_*^+\phi^+\frac{\kappa_s^+}{\sqrt{1-(\zeta_*^+)^2}}+ \nn\\ && (\zeta_*^++\phi^+\zeta_*^+)\bigg(\zeta_*^++\zeta_*^+\phi^+\frac{\kappa_s^+}{\sqrt{1-(\zeta_*^+)^2}}\bigg) \bigg(\lambda^++\frac{b^2(1-b^{2\alpha})}{1-b^2}\eta\lambda^+\bigg)+\zeta_{k-1}^+\phi^+\kappa_s^+\bigg(1+\phi^+\zeta_{k-1}^+\frac{(\kappa_s^+)^2}{\sqrt{1-(\zeta_*^+)^2}}\bigg)\nn\\
&&\bigg(\lambda_{\new,k}^+ + \frac{b^2(1-b^{2\alpha})}{1-b^2}\eta\lambda_{\new,k}^+\bigg)+ \frac{b^2(1-b^{2\alpha})}{1-b^2}\sqrt{\frac{r}{c}}\eta\lambda^+\Bigg((\zeta_*^++\phi^+\zeta_*^+)\bigg(1+(\kappa_s^+)^2\zeta_{k-1}^+\frac{\phi^+}{\sqrt{1-(\zeta_*^+)^2}}\bigg) \nn\\&& + \phi^+\kappa_s^+\zeta_{k-1}^+\bigg(\zeta_*^++\zeta_*^+\phi^+\frac{\kappa_s^+}{\sqrt{1-(\zeta_*^+)^2}}\bigg)\Bigg)
\eea
In the first inequality, we use assumption that $\alpha\geq 100$.
When $k=1$, $\tb_{\Hc}$ is a little different because the reason which will be discussed later.
\bea
b_{\Hc_1}\leq&& \tb_{\Hc_1}= (\phi_0)^2(\zeta_*^+)^2\lambda^++(\phi_0)^2(\kappa_s^+\zeta_{k-1}^+)^2\lambda_{\new,k}^+ +2\frac{b^2(1-b^{2\alpha})}{1-b^2}\sqrt{\frac{r}{c}} \eta\lambda^+\zeta_*^+\zeta_{k-1}^+\kappa_s^+(\phi_0)^2 + \frac{b^2(1-b^{2\alpha})}{1-b^2}\eta\frac{r}{c}\lambda^+(\phi_0\zeta_*^+)^2 + \nn\\&& 2\phi_0\kappa_s^+\frac{(\zeta_*^+)^2}{\sqrt{1-(\zeta_*^+)^2}}(\lambda^++\frac{b^2(1-b^{2\alpha})}{1-b^2}\eta\frac{r}{c}\lambda^+)+ 2\phi_0\max\{0.15c\zeta,\zeta_{k-1}^+\}\frac{(\kappa_s^+)^2}{\sqrt{1-(\zeta_*^+)^2}}\bigg(\lambda_{\new,k}^+ +\frac{b^2(1-b^{2\alpha})}{1-b^2}\eta\lambda_{\new,k}^+\bigg)\nn\\ && +  4\frac{b^2(1-b^{2\alpha})}{1-b^2}\eta\frac{r}{c}\lambda^+\zeta_*^+\phi_0\frac{\kappa_s^+}{\sqrt{1-(\zeta_*^+)^2}}+ (\zeta_*^++\phi_0\zeta_*^+)\bigg(\zeta_*^++\zeta_*^+\phi_0\frac{\kappa_s^+}{\sqrt{1-(\zeta_*^+)^2}}\bigg) \bigg(\lambda^++\frac{b^2(1-b^{2\alpha})}{1-b^2}\eta\lambda^+\bigg)+\nn\\ &&\zeta_{k-1}^+\phi_0\kappa_s^+\bigg(1+\phi_0\zeta_{k-1}^+\frac{(\kappa_s^+)^2}{\sqrt{1-(\zeta_*^+)^2}}\bigg)
\lambda_{\new,k}^++ \frac{b^2(1-b^{2\alpha})}{1-b^2}\sqrt{\frac{r}{c}}\eta\lambda^+\Bigg((\zeta_*^++\phi_0\zeta_*^+)\bigg(1+(\kappa_s^+)^2\zeta_{k-1}^+\frac{\phi_0}{\sqrt{1-(\zeta_*^+)^2}}\bigg) \nn\\&& + \phi_0\kappa_s^+\zeta_{k-1}^+\bigg(\zeta_*^++\zeta_*^+\phi_0\frac{\kappa_s^+}{\sqrt{1-(\zeta_*^+)^2}}\bigg)\Bigg)
\eea
Conditions 2, 4 of Theorem \ref{thm1} imply that $\kappa_{2s,*} \leq \kappa_{2s,*}^+ = 0.3$, $\kappa_{2s,\new} \leq \kappa_{2s,\new}^+ = 0.15$, $\tilde{\kappa}_{2s,k} \leq \tilde{\kappa}_{2s}^+ = 0.15$, ${\kappa}_{s,k} \leq {\kappa}_{s}^+ = 0.15$ and $g \le g^+ = \sqrt{2}$. Using Lemma \ref{RIC_bnd}, this implies that $\phi_k \le \phi^+ = 1.1735$ ($\phi_0\leq 1.1111$).

Let $$f_{inc} (\zeta_{k-1}^+;\zeta_*^+,c\zeta, \zeta_*^+rf, \zeta_*^+rf,\kappa_s^+, \phi^+, b,\alpha, \eta)=\frac{\tb_{\Hc}+0.125 c\zeta\lambda^-}{\tb_{A_k}-\tb_{A_{k,\perp}}-\tb_{\Hc}-0.25c\zeta\lambda^-}$$
As $\frac{b^2(1-b^{2\alpha})}{1-b^2}=\sum_{i=0}^{\alpha-1}b^{2(1+i)}$ is an increasing function of $b$ and $\alpha$, $f_{inc} (\zeta_{k-1}^+;\zeta_*^+,f,c\zeta,\zeta_*^+rf,\kappa_s^+, \phi^+, b,\alpha, \eta)$ is an increasing function of all their arguments. Thus, by taking upper bounds on some of the variables, we define
$$\tf_{inc}(\zeta_{k-1}^+;c\zeta,\phi^+,b)= f_{inc} (\zeta_{k-1}^+;10^{-4},c\zeta, 1.5\times 10^{-4},  1.5\times 10^{-4},0.15, \phi^+,b,\infty, 1.7)$$

Using Fact \ref{constants},  $\zeta_*^+ \leq 10^{-4}$; $\zeta_*^+f \leq 1.5\times 10^{-4}$; $\zeta_*^+rf \leq 1.5\times 10^{-4}$ and $c\zeta \le 10^{-4}$. 
\ben
\item By definition, $\zeta_0^+=1$. For $k=1$, $\zeta_k^+ = \tf_{inc} (\zeta_{k-1}^+;c\zeta,\phi^+,b) \le \tf_{inc}(1; 10^{-4},1.1111,0.4)=0.5688$. For $k=2$, $\zeta_k^+ = \tf_{inc} (\zeta_{k-1}^+;c\zeta,\phi^+,b) \le \tf_{inc}(0.5688; 10^{-4},1.1735,0.4)=0.3568$.We prove the first claim by induction.
\bi
\item Base case: For $k=3$, $\zeta_4^+ = \tf_{inc} (\zeta_{k-1}^+;c\zeta,\phi^+,b) \le f_{inc}(0.3568;10^{-4},1.1735,0.4) < 0.3568 $.

\item Induction step: Assume that $\zeta_{k-1}^+ \leq \zeta_{k-2}^+$ for $k >=4$. Since $f_{inc}$ is an increasing function of its arguments,
$\zeta_k^+ = \tf_{inc} (\zeta_{k-1}^+;c\zeta,\phi^+,b) \leq f_{inc} (\zeta_{k-2}^+; c\zeta,\phi^+,b)= \zeta_{k-1}^+$.
\ei

\item For the second claim, when $k=0,1,2$, the result is obvious correct; when $k\geq3$,notice that
\bea
\zeta_{k}^+=&&\tf_{inc} (\zeta_{k-1}^+;c\zeta,\phi^+,b)\nn\\
\leq&& \tf_{inc} (\zeta_{k-1}^+;c\zeta,1.1735,0.4)\nn\\
=&&\frac{0.058\zeta_{k-1}^++0.4284}{0.9981 - 0.058(\zeta_{k-1}^+)^2 - 0.09885\zeta_{k-1}^+ - 0.2505c\zeta }\cdot \zeta_{k-1}^+ + \frac{0.1257c\zeta+ 4.041\times 10^{-5}}{0.9981 - 0.058(\zeta_{k-1}^+)^2 - 0.09885\zeta_{k-1}^+ - 0.2505c\zeta }\nn\\
\leq&& 0.5360\zeta_{k-1}^+ + 0.150c\zeta\nn\\
\leq&& 0.6^{k} + 0.15c\zeta
\eea
Thus the claim is correct for $k\geq0$ by induction.
\item Since $\zeta_k^+ \le \zeta_0^+=1$ and $$g_{dec}(\zeta_{k-1}^+;\zeta_*^+,c\zeta, \zeta_*^+rf, \zeta_*^+rf,\kappa_s^+, \phi^+, b,\alpha, \eta)=\frac{b_{A_k}-b_{A_{k,\perp}}-b_{\Hc}-0.25c\zeta\lambda^-}{\lambda_{\new,k}^-}$$ is a decreasing function of its variables, $g_{dec} \geq  g_{dec} (1;10^{-4},10^{-4}, 1.5\times 10^{-4},  1.5\times 10^{-4},0.15, 1.1735,0.4,\infty, 1.7) > 0$.
\een
\end{proof}

\section{Proof of Lemma \ref{mainlem} }\label{pfoflem}

The proof of Lemma \ref{mainlem} follows from two lemmas. The first is the final conclusion for the projected CS step for $t\in \mathcal{I}_{j,k}$.  The second is the final conclusion for one projection PCA (i.e.) for $t\in \mathcal{I}_{j,k}$.  We will state the two lemmas first and then proceed to prove them in order.

\begin{lem}[Projected Compressed Sensing Lemma]\label{cslem}
Assume that all conditions of Theorem \ref{thm1} hold.
\ben
\item For all $t \in  \mathcal{I}_{j,k}$, for any $k=1,2,\dots K$, if $X_{j,k-1} \in \Gamma_{j,k-1}$,
\ben
\item the projection noise $\beta_t$ satisfies $\|\beta_t\|_2 \leq \zeta_{k-1}^+ \sqrt{c} \gamma_{\new,k} + \zeta_{*}^+ \sqrt{r} \gamma_* \le \sqrt{c} 0.72^{k-1} \gamma_{\new} + \sqrt{\zeta} (\sqrt{r} + 0.15\sqrt{c}) \le \xi_0$.
\item the CS error satisfies $\|\hat{S}_{t,\cs} - S_t\|_2 \le7 \xi_0$.
\item  $\hat{T}_t = T_t$
\item $e_t$ satisfies \eqref{etdef0} and $\|e_t\|_2 \leq \phi^+ [\kappa_s^+ \zeta_{k-1}^+ \sqrt{c} \gamma_{\new,k} + \zeta_{*}^+ \sqrt{r} \gamma_*] \le
 0.18 \sqrt{c}0.72^{k-1}\gamma_{\new} + 1.2 \sqrt{\zeta}(\sqrt{r} + 0.023 \sqrt{c})$.  Recall that \eqref{etdef0} is
\[
 I_{T_t} {(\Phi_{(t)})_{T_t}}^{\dag} \beta_t = I_{T_t} [ (\Phi_{(t)})_{T_t}'(\Phi_{(t)})_{T_t}]^{-1}  {I_{T_t}}' \Phi_{(t)} L_t
\]

\een
\item For all $k=1,2,\dots K$, $\mathbf{P}(\That_t = T_t \ \text{and} \ e_t \ \text{satisfies (\ref{etdef0})}  \text{ for all } t \in \mathcal{{I}}_{j,k}  | \Gamma_{j,k-1}^e)=1$.
\een
\end{lem}

\begin{lem}[Subspace Recovery Lemma] \label{zetak}
Assume that all the conditions of Theorem \ref{thm1} hold.  Let $\zeta_*^+ = r \zeta$.  Then, for all $k=1,2, \dots K$,
\[
\mathbf{P}(\zeta_{k} \le \zeta_k^+|\egam_{j,k-1}) \ge p_k(\alpha,\zeta)
\]
where $\zeta_k^+$ is defined in Definition \ref{zetakplus} and $p_k(\alpha,\zeta)$ is defined in \eqref{pk}.
\end{lem}

\begin{proof}[Proof of Lemma \ref{mainlem}]
Observe that $\mathbf{P}(\Gamma_{j,k}|\Gamma_{j,k-1}) = \mathbf{P}( \check{\Gamma}_{j,k} | \Gamma_{j,k-1})$. The lemma then follows by combining Lemma \ref{zetak} and item 2 of Lemma \ref{cslem}.
\end{proof}

\subsection{Proof of Lemma \ref{cslem}}

In order to prove Lemma \ref{cslem} we first need a bound on the RIC of the compressed sensing matrix $\Phi_k.$

\begin{lem}[Bounding the RIC of $\Phi_k$] \label{RIC_bnd}
Recall that $\zeta_*:= \|(I-\Phat_*{\Phat_*}')P_*\|_2$.  The following hold.
\ben
\item Suppose that a basis matrix $P$ can be split as $P = [P_1, P_2]$ where $P_1$ and $P_2$ are also basis matrices. Then $\kappa_s^2 (P) = \max_{T: |T| \le s} \|I_T'P\|_2^2 \le \kappa_s^2 (P_1) + \kappa_s^2 (P_2)$.
\item $\kappa_s^2(\Phat_*) \leq \kappa_{s,*}^2 + 2\zeta_*$
\item $\kappa_s (\Phat_{\new,k}) \leq \kappa_{s,\new} + \tilde{\kappa}_{s,k} \zeta_k + \zeta_*$
\item $\delta_{s} (\Phi_0) = \kappa_s^2 (\Phat_*) \leq  \kappa_{s,*}^2 + 2 \zeta_*$
\item $\delta_{s}(\Phi_k)  = \kappa_s^2 ([\Phat_* \ \Phat_{\new,k}]) \leq \kappa_s^2 (\Phat_*) + \kappa_s^2 (\Phat_{\new,k})\leq \kappa_{s,*}^2 + 2\zeta_* + (\kappa_{s,\new} + \tilde{\kappa}_{s,k} \zeta_k + \zeta_*)^2$ for $k \ge 1$
\een
\end{lem}

\begin{proof}
\ben
\item Since $P$ is a basis matrix, $\kappa_s^2 (P) = \max_{|T| \leq s} \|{I_T}' P\|_2^2$. Also, $\|{I_T}' P\|_2^2 = \|{I_T}' [P_1, P_2] [P_1, P_2]' I_T \|_2 =  \|{I_T}' (P_1P_1' + P_2 P_2') I_T \|_2 \le \|{I_T}' P_1 P_1' I_T\|_2 + \|{I_T}' P_2 P_2' I_T\|_2$. Thus, the inequality follows.

\item For any set $T$ with $|T| \le s$, $\|{I_T}' \Phat_*\|_2^2  = \|{I_T}' \Phat_* {\Phat_*}'I_T\|_2 =\|{I_T}'( \Phat_* {\Phat_*}' -P_* {P_*}' + P_*{P_*}')I_T\|_2 \leq \|{I_T}'( \Phat_* {\Phat_*}' -P_* {P_*}')I_T\|_2 + \|{I_T}'P_* {P_*}' I_T\|_2 \leq 2\zeta_* + \kappa_{s,*}^2$. The last inequality follows using Lemma \ref{lemma0} with $P=P_*$ and $\hat{P} = \hat{P}_*$.

\item By Lemma \ref{lemma0} with $P = P_*$, $\Phat = \Phat_*$ and $Q = P_\new$, $\|{P_{\new}}' \Phat_*\|_2 \leq \zeta_*$. By Lemma \ref{lemma0} with $P = P_\new$ and $\hat{P} = \hat{P}_{\new,k}$, $\|(I-P_\new P_\new')\Phat_{\new,k}\|_2  = \|(I-\Phat_{\new,k}{\Phat_{\new,k}}')P_{\new}\|_2$.
For any set $T$ with $|T| \leq s$, $\|{I_T}'\Phat_{\new,k}\|_2 \leq \|{I_T}'(I-P_{\new}P_{\new}') \Phat_{\new,k}\|_2 + \|{I_T}'P_{\new}P_{\new}' \Phat_{\new,k}\|_2 \leq \tilde{\kappa}_{s,k} \|(I- P_{\new}{P_{\new}}')\Phat_{\new,k}\|_2 + \|{I_T}'P_{\new}\|_2 = \tilde{\kappa}_{s,k} \|(I-\Phat_{\new,k}{\Phat_{\new,k}}')P_{\new}\|_2 + \|{I_T}'P_{\new}\|_2 \leq \tilde{\kappa}_{s,k} \|D_{\new,k}\|_2 + \tilde{\kappa}_{s,k}\| \Phat_* {\Phat_*}'P_{\new}\|_2 + \|{I_T}'P_{\new}\|_2 \le \tilde{\kappa}_{s,k}\zeta_{k} + \tilde{\kappa}_{s,k} \zeta_* + \kappa_{s,\new} \leq  \tilde{\kappa}_{s,k}\zeta_{k} + \zeta_* + \kappa_{s,\new}$. Taking $\max$ over $|T| \le s$ the claim follows.

\item This follows using Lemma \ref{delta_kappa} and the second claim of this lemma.

\item This follows using Lemma \ref{delta_kappa} and the first three claims of this lemma.
\een

\end{proof}

\begin{corollary}\label{RICnumbnd}
If the conditions of Theorem \ref{thm1} are satisfied, and $X_{j,k-1}\in \Gamma_{j,k-1}$,  then
\ben
\item $\delta_s(\Phi_0) \leq \delta_{2s}(\Phi_0)  \leq {\kappa_{2s,*}^+}^2 + 2\zeta_*^+ <0.1$
\item $\delta_s(\Phi_{k-1}) \leq \delta_{2s}(\Phi_{k-1}) \leq {\kappa_{2s,*}^+}^2 + 2\zeta_*^+ +(\kappa_{2s,\new}^+ + \tilde{\kappa}_{2s,k-1}^+ \zeta_{k-1}^+ + \zeta_*^+)^2 <0.1479$
\item $\phi_{k-1} \le \frac{1}{1-\delta_s(\Phi_{k-1})} < \phi^+$
\een
\end{corollary}

\begin{proof}
This follows using Lemma \ref{RIC_bnd}, the definition of $\Gamma_{j,k-1}$, and the bound on $\zeta_{k-1}^+$ from Lemma \ref{expzeta}.
\end{proof}

The following are striaghtforward bounds that will be useful for the proof of Lemma \ref{cslem} and later.

\begin{fact}\label{constants}
Under the assumptions of Theorem \ref{thm1}:
\ben
\item $\zeta \gamma_* \le \frac{\sqrt{\zeta}}{(r_0 + (J-1) c)^{3/2}} \le  \sqrt{\zeta}$
\item $\zeta_*^+ \leq  \frac{10^{-4}}{(r_0 + (J-1) c)} \le 10^{-4}$
\item $ \zeta_*^+ \gamma_*^2 \leq \frac{1}{(r_0 + (J-1)c)^2} \le 1$
\item $ \zeta_*^+ \gamma_* \le \frac{\sqrt{\zeta}}{\sqrt{r_0 + (J-1)c}} \le \sqrt{\zeta}$
\item $\zeta_*^+ f \leq \frac{1.5 \times 10^{-4}}{r_0 + (J-1)c} \le 1.5 \times 10^{-4}$
\item $\zeta_{k-1}^+ \leq 0.6^{k-1} + 0.15 c\zeta$ (from Lemma \ref{expzeta})
\item $\zeta_{k-1}^+ \gamma_{\new,k} \leq (0.6 \cdot 1.2)^{k-1} \gamma_{\new} + 0.15c\zeta \gamma_*  \le 0.72^{k-1}\gamma_{\new} + \frac{0.15\sqrt{\zeta}}{\sqrt{r_0 + (J-1)c}} \le 0.72^{k-1}\gamma_{\new} + 0.15\sqrt{\zeta}$
\item $\zeta_{k-1}^+ \gamma_{\new,k}^2 \leq (0.6 \cdot 1.2^2)^{k-1} \gamma_{\new}^2 + 0.15 c\zeta \gamma_*^2 \le 0.864^{k-1}\gamma_{\new}^2 + \frac{0.15}{{(r_0 + (J-1)c})^2} \le 0.864^{k-1}\gamma_{\new}^2 + 0.15$
\een
\end{fact}

\begin{proof}[Proof of Lemma \ref{cslem}]
Recall that $X_{j,k-1} \in \Gamma_{j,k-1}$ implies that $\zeta_* \leq \zeta_*^+$ and $\zeta_{k-1}\leq \zeta_{k-1}^+$.
\ben
\item
\begin{enumerate}
\item For $t \in \mathcal{I}_{j,k}$, $\beta_t := (I-\Phat_{(t-1)} {\Phat_{(t-1)} }') L_t = D_{*,k-1} a_{t,*} + D_{\new,k-1} a_{t,\new} $. Thus, using Fact \ref{constants}
\begin{align*}
\|\beta_t\|_2 & \leq \zeta_* \sqrt{r} \gamma_* + \zeta_{k-1} \sqrt{c} \gamma_{\new,k} \\
& \leq \sqrt{\zeta}\sqrt{r} + (0.72^{k-1}\gamma_{\new} + 0.15\sqrt{\zeta})\sqrt{c} \\
& = \sqrt{c} 0.72^{k-1} \gamma_{\new} + \sqrt{\zeta} (\sqrt{r} + 0.15\sqrt{c}) \leq \xi_0.
\end{align*}

\item By Corollary \ref{RICnumbnd}, $\delta_{2s} (\Phi_{k-1}) < 0.15< \sqrt{2}-1$. Given $|T_t| \leq s$, $\|\beta_t\|_2 \leq \xi_0 = \xi$, by Theorem \ref{candes_csbound}, the CS error satisfies
\[
\|\hat{S}_{t,\cs} - S_t\|_2 \leq  \frac{4\sqrt{1+\delta_{2s}(\Phi_{k-1})}}{1-(\sqrt{2}+1)\delta_{2s}(\Phi_{k-1})} \xi_0 < 7 \xi_0.
\]

\item Using the above and the definition of $\rho$, $\|\hat{S}_{t,\cs} - S_t\|_{\infty} \leq 7 \rho \xi_0$. Since $\min_{i\in T_t} |(S_t)_{i}| \geq S_{\min}$ and $(S_t)_{T_t^c} = 0$, $\min_{i\in T_t} |(\hat{S}_{t,cs})_i| \geq S_{\min} - 7 \rho \xi_0$ and $\min_{i\in T_t^c} |(\hat{S}_{t,\cs})_i| \leq 7 \rho \xi_0$. If $\omega < S_{\min} - 7 \rho \xi_0$, then $\hat{T}_t \supseteq T_t$. On the other hand, if $\omega > 7 \rho \xi_0$, then $\hat{T}_t \subseteq T_t$. Since $S_{\min} > 14 \rho \xi_0$ (condition 2 of the theorem) and $\omega$ satisfies $7 \rho \xi_0 \leq \omega \leq S_{\min} -7\rho \xi_0$ (condition 1 of the theorem), then the support of $S_t$ is exactly recovered, i.e. $\hat{T}_t = T_t$.

\item Given $\hat{T}_t = T_t$, the LS estimate of $S_t$ satisfies $(\hat{S}_t)_{T_t} = [(\Phi_{k-1})_{T_t}]^{\dag} y_t =[(\Phi_{k-1})_{T_t}]^{\dag} (\Phi_{k-1} S_t + \Phi_{k-1}L_t)$ and $(\hat{S}_t)_{T_t^c} = 0$ for $t \in \mathcal{I}_{j,k}$. Also,  ${(\Phi_{k-1})_{T_t}}' \Phi_{k-1} = {I_{T_t}}' \Phi_{k-1}$ (this follows since $(\Phi_{k-1})_{T_t} = \Phi_{k-1} I_{T_t}$ and $\Phi_{k-1}'\Phi_{k-1} = \Phi_{k-1}$).  Using this, the LS error $e_t := \hat{S}_t - S_t$ satisfies (\ref{etdef0}).
Thus, using Fact \ref{constants} and condition 2 of the theorem,
\begin{align*}
\|e_t\|_2 & \le \phi^+ (\zeta_*^+ \sqrt{r}\gamma_* + \kappa_{s,k-1} \zeta_{k-1}^+ \sqrt{c}\gamma_{\new,k}) \\
& \le 1.2 \left(\sqrt{r}\sqrt{\zeta}+ \sqrt{c} 0.15 (0.72)^{k-1} + \sqrt{c} 0.023\sqrt{\zeta}\right) \\
& = 0.18 \sqrt{c}0.72^{k-1}\gamma_{\new} + 1.2 \sqrt{\zeta}(\sqrt{r} + 0.023 \sqrt{c}).
\end{align*}
\end{enumerate}

\item The second claim is just a restatement of the first.

\een

\end{proof}

\section{Proof of Lemma \ref{zetak}}

The proof of Lemma \ref{zetak} will use the next two lemmas (\ref{zetakbnd}, and \ref{termbnds}).

\begin{lem}\label{zetakbnd}
If $\lambda_{\min}(A_k) - \|A_{k,\perp}\|_2 - \|\mathcal{H}_k\|_2 >0$, then
\beq \label{zetakbound}
\zeta_k \leq  \frac{\|\mathcal{R}_k\|_2}{\lambda_{\min} (A_k) - \|A_{k,\perp}\|_2 - \|\mathcal{H}_k\|_2} \leq  \frac{\|\mathcal{H}_k\|_2}{\lambda_{\min} (A_k) - \|A_{k,\perp}\|_2 - \|\mathcal{H}_k\|_2}
\eeq
where $\mathcal{R}_k := \mathcal{H}_k E_\new$ and $A_k$, $A_{k,\perp}$, $\mathcal{H}_k$ are defined in Definition \ref{defHk}.
\end{lem}

\begin{proof}
Since $\lambda_{\min}(A_k) - \|A_{k,\perp}\|_2 - \|\mathcal{H}_k\|_2 >0$, so $\lambda_{\min}(A_k) > \|A_{k,\perp}\|_2$. Since $A_k$ is of size $c_\new \times c_\new$ and $\lambda_{\min}(A_k) > \|A_{k,\perp}\|_2$, $\lambda_{c_\new+1} (\mathcal{A}_k) = \|A_{k,\perp}\|_2$. By definition of EVD, and since $\Lambda_k$ is a $c_\new \times c_\new$ matrix, $\lambda_{\max}(\Lambda_{k,\perp}) = \lambda_{c_\new+1}(\mathcal{A}_k + \mathcal{H}_k)$. By Weyl's theorem (Theorem \ref{weyl}), $\lambda_{c_\new+1}(\mathcal{A}_k + \mathcal{H}_k) \leq \lambda_{c_\new+1} (\mathcal{A}_k) + \|\mathcal{H}_k\|_2 = \|A_{k,\perp}\|_2 + \|\mathcal{H}_k\|_2$. Therefore, $\lambda_{\max}(\Lambda_{k,\perp})\leq \|A_{k,\perp}\|_2 + \|\mathcal{H}_k\|_2$ and hence $\lambda_{\min}(A_k) - \lambda_{\max}(\Lambda_{k,\perp})\geq \lambda_{\min}(A_k) - \|A_{k,\perp}\|_2 - \|\mathcal{H}_k\|_2 > 0$. Apply the $\sin \theta$ theorem (Theorem \ref{sin_theta}) with $\lambda_{\min}(A_k) - \lambda_{\max}(\Lambda_{k,\perp})> 0$, we get
$$\|(I- \Phat_{\new,k} {\Phat_{\new,k}}') E_{\new} \|_2 \leq \frac{\|\mathcal{R}_k\|_2}{ \lambda_{\min}(A_k) - \lambda_{\max}(\Lambda_{k,\perp})}
\leq \frac{\|\mathcal{H}_k\|_2}{\lambda_{\min}(A_k) - \|A_{k,\perp}\|_2 - \|\mathcal{H}_k\|_2}$$
Since $\zeta_k  = \|(I- \Phat_{\new,k} {\Phat_{\new,k}}') D_{\new}\|_2 = \|(I- \Phat_{\new,k} {\Phat_{\new,k}}') E_{\new} R_{\new} \|_2 \leq \|(I- \Phat_{\new,k} {\Phat_{\new,k}}') E_{\new}\|_2$, the result follows.
The last inequality follows because  $\|R_\new\|_2 = \|E_\new' D_\new\|_2 \le 1$.
\end{proof}

\begin{lem}[High probability bounds for each of the terms in the $\zeta_k$ bound (\ref{zetakbound})]\label{termbnds}
Assume the conditions of Theorem \ref{thm1} hold.  Also assume that $\mathbf{P}(\Gamma_{j,k-1}^e)>0$ for all $1\leq k \leq K+1$. Then, for all $1 \le k \le K$
\begin{enumerate}

\item \bea
\Pb\left(\lambda_{\min}(A_k)\geq b_{A_k} - \frac{c\zeta\lambda^-}{12}\right)\geq 1-p_a(\alpha,\zeta)
\label{A_k_bound}
\eea
where $b_{A_k}=(1-(\zeta_*^+)^2)(1-\frac{b^2-b^{2\alpha+2}}{\alpha(1-b^2)})\lambda_{\new,k}^--2\frac{b^2\left(1-b^{2\alpha}\right)}{1-b^2}\sqrt{cr}\gamma_*\gamma_{\new,k}\zeta_*^+$ and
\bea
p_a(\alpha,\zeta)= &&c\digamma\left(\alpha,\frac{c\zeta\lambda^-}{96},\frac{(1-b^{2\alpha})(1-b)}{1+b}c\gamma_{\new,k}^2\right) + 2c\digamma\left(\alpha,\frac{c\zeta\lambda^-}{96}, 2\frac{b(1-b)(1-b^{2(\alpha-1)})}{1+b}c\gamma_{\new,k}^2\right) \nn\\&&+ 2c\digamma\left(\alpha,\frac{c\zeta\lambda^-}{96},2\frac{(1-b^{\alpha})^2}{1+b}c\gamma_{\new,k}^2\right) +2c\digamma(\alpha,\frac{c\zeta\lambda^-}{48},4\frac{1+b-2\sqrt{b^{\alpha+1}(1+b-b^{\alpha+1})}}{1+b}\sqrt{cr}\gamma_*\gamma_{\new,k}\zeta_*^+)
\nn\\&&+2c\digamma\left(\alpha,\frac{c\zeta\lambda^-}{96},4\frac{b(1-b^{2\alpha})}{1+b}c\gamma_{\new,k}\gamma_{\new,k-1}\right) +2c\digamma\left(\alpha,\frac{c\zeta\lambda^-}{48}, 8\frac{b(1-b^{2\alpha})}{1+b}\sqrt{cr}\gamma_*\gamma_{\new,k}\zeta_*^+\right)
\eea

\item \bea
&&\Pb\left(\left\|A_{k,\perp}\right\|\leq b_{A_{k,\perp}}+\frac{c\zeta\lambda^-}{24}\Big|X_{j,k-1}\right)\geq1-p_b(\alpha,\zeta)
\label{A_k_perp_bound}
\eea
where $b_{A_{k,\perp}}=\lambda^+(\zeta_*^+)^2+\frac{b^2(1-b^{2\alpha})}{1-b^2}r\gamma_*^2(\zeta_*^+)^2$ and
\bea
p_b(\alpha,\zeta)=&&(n-c)\digamma\left(\alpha,\frac{c\zeta\lambda^-}{96},\frac{(1-b^{2\alpha})(1-b)}{1+b}r\gamma_*^2(\zeta_*^+)^2\right) + 2(n-c)\digamma\left(\alpha,\frac{c\zeta\lambda^-}{96}, 2\frac{b(1-b)(1-b^{2(\alpha-1)})}{(1+b)}r\gamma_*^2(\zeta_*^+)^2\right) \nn\\&&+ 2(n-c)\digamma\left(\alpha,\frac{c\zeta\lambda^-}{96},2\frac{(1-b^{\alpha})^2}{1+b}r\gamma_*^2(\zeta_*^+)^2\right) +2(n-c)\digamma\left(\alpha,\frac{c\zeta\lambda^-}{96},4\frac{b(1-b^{2\alpha})}{1+b}r\gamma_*^2(\zeta_*^+)^2\right)
\eea

\item \bea
\mathbf{P} \left(\|\mathcal{H}_k\|_2 \leq b_{\Hc} + \frac{c\zeta\lambda^-}{8}|X_{j,k-1}\right) \geq 1- p_c(\alpha,\zeta) - p_f(\alpha,\zeta)- \max\{p_d(\alpha,\zeta), p_e(\alpha,\zeta)\}
\label{H_k_bound}
\eea
where

\bea
b_{\Hc} :=&&
(\phi^+)^2((\zeta_*^+)^2+(\kappa_s^+\zeta_{k-1}^+)^2)\lambda^+ +2\frac{b^2(1-b^{2\alpha})}{1-b^2}\sqrt{cr}\gamma_*\gamma_{\new,k}\zeta_*^+\zeta_{k-1}^+\kappa_s^+(\phi^+)^2 + \frac{b^2(1-b^{2\alpha})}{1-b^2}r\gamma_*^2(\phi^+\zeta_*^+)^2 + \nn\\&&\frac{b^2(1-b^{2\alpha})}{1-b^2}c\gamma_{\new,k-1}^2(\phi^+\kappa_s^+\zeta_{k-1}^+)^2 + 2\phi^+\kappa_s^+\frac{(\zeta_*^+)^2}{\sqrt{1-(\zeta_*^+)^2}}\left(\lambda^++\frac{b^2(1-b^{2\alpha})}{1-b^2}r\gamma_{*}^2\right)+ \nn\\ && 2\phi^+\max\{0.15c\zeta,\zeta_{k-1}^+\}\frac{(\kappa_s^+)^2}{\sqrt{1-(\zeta_*^+)^2}}\left(\lambda^+ +\frac{b^2(1-b^{2\alpha})}{1-b^2}c\gamma_{\new,k-1}^2\right) +  4\frac{b^2(1-b^{2\alpha})}{1-b^2}r\gamma_*^2\zeta_*^+\phi^+\frac{\kappa_s^+}{\sqrt{1-(\zeta_*^+)^2}}+ \nn\\ && (\zeta_*^++\phi^+\zeta_*^+)(\zeta_*^++\zeta_*^+\phi^+\frac{\kappa_s^+}{\sqrt{1-(\zeta_*^+)^2}}) (\lambda^++\frac{b^2(1-b^{2\alpha})}{1-b^2}r\gamma_{*}^2)+\zeta_{k-1}^+\phi^+\kappa_s^+(1+\phi^+\zeta_{k-1}^+\frac{(\kappa_s^+)^2}{\sqrt{1-(\zeta_*^+)^2}})\nn\\
&&\left(\lambda^+ + \frac{b^2(1-b^{2\alpha})}{1-b^2}c\gamma_{\new,k-1}^2\right)+ \frac{b^2(1-b^{2\alpha})}{1-b^2}\sqrt{cr}\gamma_*\gamma_{\new,k}\Bigg((\zeta_*^++\phi^+\zeta_*^+)\bigg(1+(\kappa_s^+)^2\zeta_{k-1}^+\frac{\phi^+}{\sqrt{1-(\zeta_*^+)^2}}\bigg) \nn\\&& + \phi^+\kappa_s^+\zeta_{k-1}^+\Big(\zeta_*^++\zeta_*^+\phi^+\frac{\kappa_s^+}{\sqrt{1-(\zeta_*^+)^2}}\Big)\Bigg)
\eea

\end{enumerate}

\end{lem}

\begin{proof}[Proof of Lemma \ref{zetak}]
Lemma \ref{zetak} now follows by combining Lemmas \ref{zetakbnd} and \ref{termbnds} and defining
\begin{equation}
p_k(\alpha,\zeta) := 1 - p_{a}(\alpha,\zeta) - p_b(\alpha,\zeta) - p_c(\alpha,\zeta) - p_f(\alpha,\zeta)- \max\{p_d(\alpha,\zeta), p_e(\alpha,\zeta)\}. \label{pk}
\end{equation}
\end{proof}

As above, we will start with some simple facts that will be used to prove Lemma \ref{termbnds}.

For convenience, we will use $\frac{1}{\alpha}\sum_t$ to denote $\frac{1}{\alpha} \sum_{t \in \mathcal{I}_{j,k}}$

\begin{fact}\label{keyfacts}
Under the assumptions of Theorem \ref{thm1} the following are true.

\begin{enumerate}

\item The matrices $D_{\new}$, $R_{\new}$, $E_{\new}$, $D_{*}, D_{\new,k-1}$, $\Phi_{k-1}$ are functions of the r.v. $X_{j,k-1}$.  All terms that we bound for the first two claims of the lemma are of the form $\frac{1}{\alpha} \sum_{t \in \mathcal{I}_{j,k}} Z_t$ where $Z_t= f_1(X_{j,k-1}) Y_t f_2(X_{j,k-1})$, $Y_t$ is a sub-matrix of $a_t a_t'$ and $f_1(.)$ and $f_2(.)$ are functions of $X_{j,k-1}$.

\item  $X_{j,k-1}$ is independent of any $a_{t}$ for $t \in  \mathcal{I}_{j,k}$ , and hence the same is true for the matrices  $D_{\new}$, $R_{\new}$, $E_{\new}$, $D_{*}, D_{\new,k-1}$, $\Phi_{k-1}$. Also, $a_{t}$'s for different $t \in \mathcal{I}_{j,k}$ are mutually independent. Thus, conditioned on  $X_{j,k-1}$, the $Z_t$'s defined above are mutually independent.
\label{X_at_indep} \label{zt_indep}

\item All the terms that we bound for the third claim contain $e_t$. Using the second claim of Lemma \ref{cslem}, conditioned on $X_{j,k-1}$, $e_t$ satisfies (\ref{etdef0}) w.p. one whenever $X_{j,k-1} \in \Gamma_{j,k-1}$. Conditioned on $X_{j,k-1}$, all these terms are also of the form $\frac{1}{\alpha} \sum_{t \in \mathcal{I}_{j,k}} Z_t$ with $Z_t$ as defined above, whenever $X_{j,k-1} \in \Gamma_{j,k-1}$. Thus, conditioned on  $X_{j,k-1}$, the $Z_t$'s for these terms are mutually independent, whenever $X_{j,k-1} \in \Gamma_{j,k-1}$.

\item It is easy to see that $\|\Phi_{k-1} P_*\|_2 \le \zeta_*$, $\zeta_0 = \|D_\new\|_2 \le 1$,  $\Phi_0 D_\new = \Phi_0'D_\new = D_\new$,  $\|R_\new\| \le  1$, $\|(R_\new)^{-1}\| \le 1/\sqrt{1 - \zeta_*^2}$, ${E_{\new,\perp}}' D_\new = 0$, and $\|{E_{\new}}' \Phi_0 e_t\| =\|(R_\new')^{-1} D_\new' \Phi_0 e_t\| = \|(R_\new)^{-1} D_\new' e_t\| \le \|(R_\new')^{-1} D_\new' I_{T_t}\| \|e_t\| \le \frac{\kappa_s^+}{\sqrt{1 - \zeta_*^2}}\|e_t\|$. The bounds on $\|R_\new\|$ and $\|(R_\new)^{-1}\|$ follow using Lemma \ref{lemma0} and the fact that $\sigma_{i}(R_\new) = \sigma_{i}(D_\new)$.
 \label{rnew}

\item $X_{j,k-1} \in \Gamma_{j,k-1}$ implies that
\begin{enumerate}
\item $\zeta_{*} \le \zeta_*^+$ (see Remark \ref{Gamma_rem})
\item $\zeta_{k-1}\leq \zeta_{k-1}^+ \leq 0.6^{k-1} + 0.15c\zeta$  (This follows by the definition of $\Gamma_{j,k-1}$ and Lemma \ref{expzeta}.)
\end{enumerate}
\label{leqzeta}

\item Item \ref{leqzeta} implies that
\ben
\item $\lambda_{\min}(R_{\new}{R_{\new}}') \geq 1-(\zeta_*^+)^2$. This follows from Lemma \ref{lemma0} and the fact that $\sigma_{\min}(R_\new) = \sigma_{\min}(D_\new)$.
\item $\|{I_{T_t}}' \Phi_{k-1} P_* \|_2 \leq \|\Phi_{k-1} P_* \|_2 \leq \zeta_* \leq \zeta_*^+$, $\|{I_{T_t}}'D_{\new,k-1}\|_2 \leq \kappa_{s,k-1} \zeta_{k-1} \leq \kappa_s^+ \zeta_{k-1}^+$.
\een
\label{X_in_Gamma}

\item By Weyl's theorem (Theorem \ref{weyl}),  for a sequence of matrices $B_t$, $\lambda_{\min}(\sum_t B_t) \ge \sum_t \lambda_{\min}(B_t)$ and $\lambda_{\max}(\sum_t B_t) \le \sum_t \lambda_{\max}(B_t)$.

\end{enumerate}

\end{fact}

\begin{proof}[Proof of Lemma \ref{termbnds}]
In this proof, we frequently refer to items from Sec. \ref{keyfacts} and the following bounds (\ref{err_bound}).
\bea
&&\|E_{\new}'\Phi_0\|_2 \leq 1\nn\\
&&\|E_{\new}'\Phi_0I_{T_t}\|_2 = \|\left(R_{\new}'\right)^{-1}D_{\new}'\Phi_0I_{T_t}\|_2 = \|\left(R_{\new}'\right)^{-1}D_{\new}'I_{T_t}\|_2 \leq \frac{\zeta_0\kappa_s^+}{\sqrt{1-\zeta_*^2}} \leq \frac{\kappa_s^+}{\sqrt{1-\left(\zeta_*^+\right)^2}}\nn\\
&&\|\Phi_0P_*\|_2=\zeta_*\leq \zeta_*^+\nn\\
&&\|\Phi_{k-1}P_*\|_2\leq \|\Phi_{0}P_*\|_2 \leq \zeta_*^+\nn\\
&&\|I_{T_t}'D_{\new,k-1}\|\leq\kappa_s^+\zeta_{k-1}^+
\label{err_bound}
\eea

\subsection{Bounds for $\sum_{t\in \Ic_{j,k}}a_{t,\new}a_{t,\new}', \sum_{t\in \Ic_{j,k}}a_{t,\new}a_{t,*}', \sum_{t\in \Ic_{j,k}}a_{t,*}a_{t,*}'$}

For calculation simplicity, let $\{d_t\}_{0\leq t\leq \alpha-1}$ denote $\{a_{t,\new}\}_{t\in \Ic_{j,k}}$, $\{\mu_t\}_{0\leq t\leq \alpha-1}$ denote $\{\nu_{t,\new}\}_{t\in \Ic_{j,k}}$, $\{m_t\}_{0\leq t\leq \alpha-1}$ denote $\{a_{t,*}\}_{t\in \Ic_{j,k}}$, and $\{\omega_t\}_{0\leq t\leq \alpha-1}$ denote $\{\nu_{t,*}\}_{t\in \Ic_{j,k}}$.
\bean
&&\sum_{t=0}^{\alpha-1}d_{t}d_{t}'\nn\\
=&&\sum_{t=0}^{\alpha-1}(bd_{t-1}+\mu_{t})(bd_{t-1}+\mu_{t})'\nn\\
=&&\sum_{t=0}^{\alpha-1}(b^{t+1}d_{-1}+ \sum_{i=0}^{t}b^{t-i}\mu_{i})(b^{t+1}d_{-1}+ \sum_{i=0}^{t}b^{t-i}\mu_{i})'\nn\\
=&&\sum_{t=0}^{\alpha-1}(b^{2t}\sum_{i=0}^{t}b^{-2i}\mu_{i}\mu_{i}')+\sum_{t=0}^{\alpha-1}(b^{2t}\sum_{0\le i_1, i_2 \le t, i_1\neq i_2}b^{-i_1-i_2}\mu_{i_1}\mu_{i_2}')+\sum_{t=0}^{\alpha-1}\sum_{i=0}^{t}b^{2t-i+1}d_{-1}\mu_{i}' +\nn\\
 && \sum_{t=0}^{\alpha-1}\sum_{i=0}^{t}b^{2t-i+1}\mu_{i}d_{-1}'+ \sum_{t=0}^{\alpha-1}b^{2(t+1)}d_{-1}d_{-1}'\nn\\
=&&\sum_{i=0}^{\alpha-1}\sum_{t=i }^{\alpha-1 }(b^{2t}b^{-2i}\mu_{i}\mu_{i}') + \sum_{t=0}^{\alpha-1}b^{2t} (\sum_{i_1=0}^{t}\sum_{i_2=0}^{i_1-1} b^{-i_1-i_2}\mu_{i_1}\mu_{i_2}' +\sum_{i_1=0}^{ t}\sum_{i_2= i_1+1}^{t} b^{-i_1-i_2}\mu_{i_1}\mu_{i_2}')+ \nn\\
&&\sum_{i=0}^{\alpha-1}\sum_{t=i}^{\alpha-1}b^{2t-i+1}d_{-1}\mu_{i}' + \sum_{i=0}^{\alpha-1}\sum_{t=i}^{\alpha-1}b^{2t-i+1}\mu_{i}d_{-1}'+ \frac{b^2(1-b^{2\alpha})}{1-b^2}d_{-1}d_{-1}'\nn\\
=&&\sum_{i=0}^{\alpha-1}\frac{(1-b^{2(\alpha-i)})}{1-b^2}\mu_{i}\mu_{i}'+\sum_{i_1=0}^{\alpha-1}\sum_{t=i_1 }^{\alpha -1 }b^{2t}(\sum_{i_2=t_j+(k-1)}^{i_1-1} b^{-i_1-i_2}\mu_{i_1}\mu_{i_2}'+\sum_{i_2= i_1+1}^{t} b^{-i_1-i_2}\mu_{i_1}\mu_{i_2}')+\nn\\
&&\sum_{i=0}^{\alpha-1}\frac{b^{i+1}(1-b^{2(\alpha-i)})}{1-b^2}(\mu_{i}d_{-1}' +d_{-1}\mu_{i}')+ \frac{b^2(1-b^{2\alpha})}{1-b^2}d_{-1}d_{-1}'\nn\\
=&&\sum_{i=0}^{\alpha-1}\frac{(1-b^{2(\alpha-i)})}{1-b^2}\mu_{i}\mu_{i}'+\sum_{i_1=0}^{\alpha-1}\frac{(1-b^{2(\alpha-i_1)})}{1-b^2} \sum_{i_2=0}^{i_1-1} b^{i_1-i_2}\mu_{i_1}\mu_{i_2}' + \sum_{i_1=-1}^{\alpha-1}\sum_{i_2=i_1+1}^{\alpha-1 }\sum_{t=i_2}^{\alpha-1} b^{2t}b^{-i_1-i_2}\mu_{i_1}\mu_{i_2}'\nn\\ &&+\sum_{i=0}^{\alpha-1}\frac{b^{i+1}(1-b^{2(\alpha-i)})}{1-b^2}(\mu_{i}d_{-1}' + d_{-1}\mu_{i}')+ \frac{b^2(1-b^{2\alpha})}{1-b^2}d_{-1}d_{-1}'\nn\\
=&&\sum_{i=0}^{\alpha-1}\frac{(1-b^{2(\alpha-i)})}{1-b^2}\mu_{i}\mu_{i}' + \sum_{i_1=0}^{\alpha-1}\frac{(1-b^{2(\alpha-i_1)})}{1-b^2}\sum_{i_2=0}^{i_1-1} b^{i_1-i_2}\mu_{i_1}\mu_{i_2}'+ \sum_{i_1=0}^{\alpha-1}\sum_{i_2=i_1+1}^{\alpha-1 }\frac{(1-b^{2(\alpha-i_2)})}{1-b^2}b^{-i_1+i_2}\mu_{i_1}\mu_{i_2}'\nn\\ && +\sum_{i=0}^{\alpha-1}\frac{b^{i+1}(1-b^{2(\alpha-i)})}{1-b^2}(\mu_{i}d_{-1}' + d_{-1}\mu_{i}')+ \frac{b^2(1-b^{2\alpha})}{1-b^2}d_{-1}d_{-1}'\nn\\
=&&\sum_{i=0}^{\alpha-1}\frac{(1-b^{2(\alpha-i)})}{1-b^2}\mu_{i}\mu_{i}' + \sum_{i_1=0}^{\alpha-1}\frac{(1-b^{2(\alpha-i_1)})}{1-b^2}\sum_{i_2=0}^{i_1-1} b^{i_1-i_2}\mu_{i_1}\mu_{i_2}'+ \sum_{i_1=0}^{\alpha-1}\sum_{i_2=\alpha-i_1}^{\alpha-1 }\frac{(1-b^{2(\alpha-i_2)})}{1-b^2}b^{i_1+i_2-\alpha+1}\mu_{\alpha-i_1-1}\mu_{i_2}'\nn\\ && +\sum_{i=0}^{\alpha-1}\frac{b^{i+1}(1-b^{2(\alpha-i)})}{1-b^2}(\mu_{i}d_{-1}' + d_{-1}\mu_{i}')+ \frac{b^2(1-b^{2\alpha})}{1-b^2}d_{-1}d_{-1}'
\label{d_t_new}
\eean

Let
\bea
&&Z_{1,i}=\frac{(1-b^{2(\alpha-i)})}{1-b^2}\mu_{i}\mu_{i}', \nn\\
&&Z_{2,i}=\sum_{i_2=0}^{i-1}\frac{(1-b^{2(\alpha-i)})}{1-b^2}b^{i-i_2}\mu_{i}\mu_{i_2}', i\geq 1  \nn\\
&&Z_{3,i}=\sum_{i_2=\alpha-i}^{\alpha-1 }\frac{(1-b^{2(\alpha-i_2)})}{1-b^2}b^{i+i_2-\alpha+1}\mu_{\alpha-i-1}\mu_{i_2}',\nn\\
&&Z_{4,i}=\frac{b^{i+1}(1-b^{2(\alpha-i)})}{1-b^2}\mu_{i}d_{-1}'
\eea
$Z_{3,i}=\sum_{i_2=\alpha-i}^{\alpha-1 }\frac{(1-b^{2(\alpha-i_2)})}{1-b^2}b^{i+i_2-\alpha+1}\mu_{\alpha-i-1}\mu_{i_2}'=\mu_{\alpha-i-1}h_i(\mu_{\alpha-i}, \mu_{\alpha-i+1}, \cdots, \mu_{\alpha-1})$, thus
\bea
\E_{i-1}(Z_{3,i})&&=\E(\mu_{\alpha-i-1}h_i(\mu_{\alpha-i}, \mu_{\alpha-i+1}, \cdots, \mu_{\alpha-1})|g_i(\mu_{\alpha-i}, \mu_{\alpha-i+1}, \cdots, \mu_{\alpha-1}))\nn\\
&&=\E(\mu_{\alpha-i-1})\E(h_i(\mu_{\alpha-i}, \mu_{\alpha-i+1}, \cdots, \mu_{\alpha-1})|g_i(\mu_{\alpha-i}, \mu_{\alpha-i+1}, \cdots, \mu_{\alpha-1}))\nn\\
&&=0
\eea

Then we can split $\{d_{t}d_{t}'\}_{0\leq t\leq \alpha-1}$ into four adapted sequences, i.e., $\{Z_{1,i}\}_{0\leq i\leq \alpha-1}, \{Z_{2,i}\}_{0\leq i\leq \alpha-1}, \{Z_{3,i}\}_{0\leq i\leq \alpha-1}$ and $\{Z_{4,i}\}_{0\leq i\leq \alpha-1}.$ \\
According to Assumption \ref{model_lt}(3),
\bea
&&\lambda_{\min}\left(\E_{i-1}\left(\frac{1}{\alpha}\sum_{i=0}^{\alpha-1}Z_{1,i}|X_{j,k-1}\right)\right)\nn\\
\geq&&\frac{1}{\alpha}\sum_{i=0}^{\alpha-1}\frac{(1-b^{2(\alpha-i)})}{1-b^2}\lambda_{\min}(\E(\mu_{i}\mu_{i}'))\nn\\
\geq&&\frac{\alpha(1-b^2)-b^2+b^{2\alpha+2}}{\alpha(1-b^2)}\lambda_{\new,k}^-\nn\\
=&&(1-\frac{b^2-b^{2\alpha+2}}{\alpha(1-b^2)})\lambda_{\new,k}^-\nn
\eea
and
\bea
&&\lambda_{\max}\left(\E_{i-1}\left(\frac{1}{\alpha}\sum_{i=0}^{\alpha-1}Z_{1,i}|X_{j,k-1}\right)\right)\nn\\
\leq&&\frac{1}{\alpha}\sum_{i=0}^{\alpha-1}\frac{(1-b^{2(\alpha-i)})}{1-b^2}\lambda_{\max}(\E(\mu_{i}\mu_{i}'))\nn\\
\leq&&\left(\frac{\alpha(1-b^2)-b^2+b^{2\alpha+2}}{\alpha(1-b^2)}\lambda_{\new,k}^+\right)\nn\\
\leq&&\lambda_{\new,k}^+\nn
\eea
and by Lemma \ref{ind_expec},
\bea
\E_{i-1} (Z_{2,i}|X_{j,k-1})=&&\E_{i-1}\left(\sum_{i_2=0}^{i-1}\frac{(1-b^{2(\alpha-i)})}{1-b^2}b^{i-i_2}\mu_{i}\mu_{i_2}'\right)\nn\\
=&&\E\left(\mu_{i}\sum_{i_2=0}^{i-1}\frac{(1-b^{2(\alpha-i)})}{1-b^2}b^{i-i_2}\mu_{i_2}'|Z_{2,0},\cdots,Z_{2,i-1}\right)\nn\\
=&&\E\left(\mu_{i}h(\mu_0,\mu_1,\cdots,\mu_{i-1})|g(\mu_0,\mu_1,\cdots,\mu_{i-1})\right)\nn\\
=&&\E(\mu_{i})\E\left(h(\mu_0,\mu_1,\cdots,\mu_{i-1})|g(\mu_0,\mu_1,\cdots,\mu_{i-1})\right)\nn\\
=&&0,
\label{exp_0}
\eea
Similarly, we have $\E_{i-1} (Z_{3,i}|X_{j,k-1})=0, \E_{i-1}(Z_{4,i}|X_{j,k-1})=0.$

$$0\preceq Z_{1,i}\preceq\frac{(1-b^{2(\alpha-i)})}{1-b^2} c(1-b)^2\gamma_{\new,k}^2I\preceq \frac{(1-b^{2\alpha})(1-b)}{1+b}c\gamma_{\new,k}^2I,$$
\bea
\|Z_{2,i_1}\|&&\leq\sum_{i_2=0}^{i_1-1}\frac{(1-b^{2(\alpha-i_1)})}{1-b^2}b^{i_1-i_2}c(1-b)^2\gamma_{\new,k}^2\nn\\
&&\leq \frac{b(1-b^{2(\alpha-i_1)})(1-b^{i_1})}{1+b}c\gamma_{\new,k}^2\nn\\
&&\leq \frac{b(1-b)(1-b^{2(\alpha-1)})}{1+b}c\gamma_{\new,k}^2
\eea

\bea
\|Z_{3,i}\|&&\leq\sum_{i_2=\alpha-i}^{\alpha-1 }\frac{(1-b^{2(\alpha-i_2)})}{1-b^2}b^{i+i_2-\alpha+1}c(1-b)^2\gamma_{\new,k}^2\nn\\
&&\leq \frac{(1-b^{i+1})^2}{1+b}c\gamma_{\new,k}^2\nn\\
&&\leq \frac{(1-b^{\alpha})^2}{1+b}c\gamma_{\new,k}^2
\eea
\bea
\|Z_{4,i}\|&&\leq\frac{b^{i+1}(1-b^{2(\alpha-i)})}{1-b^2}c(1-b)\gamma_{\new,k}\gamma_{\new,k-1}\nn\\
&&\leq \frac{b(1-b^{2\alpha})}{1+b}c\gamma_{\new,k}\gamma_{\new,k-1}
\eea

\bea
0\preceq\frac{b^2(1-b^{2\alpha})}{1-b^2}d_{-1}d_{-1}'\preceq \frac{b^2(1-b^{2\alpha})}{1-b^2}c\gamma_{\new,k-1}^2I
\eea

Given $X_{j,k-1}$, $\{Z_{1,i}\}_{i\in \Ic_{j,k}}$ is an adapted sequence of Hermitian matrices, thus by Corollary \ref{azuma_nonzero}, for all $X_{j,k-1}\in \Gamma_{j,k-1}$, we have

\begin{align}
\Pb\left(\lambda_{\min}\left(\frac{1}{\alpha}\sum_{i=0}^{\alpha-1}Z_{1,i}\right)\geq
\left(1-\frac{b^2-b^{2\alpha+2}}{\alpha(1-b^2)}\right)\lambda_{\new,k}^--\epsilon\Big|X_{j,k-1}\right) \geq
 1-c\digamma\left(\alpha,\epsilon,\frac{(1-b^{2\alpha})(1-b)}{1+b}c\gamma_{\new,k}^2\right)
\label{z_1_bound}
\end{align}
Similarly, given $X_{j,k-1}$, $\{Z_{2,i_1}\}_{i_1\in \Ic_{j,k}}, \{Z_{3,i_1}\}_{i_1\in \Ic_{j,k}}, \{Z_{4,i}\}_{i\in \Ic_{j,k}}$ are adapted sequences of matrices, by Corollary \ref{azuma_rec}, $\text{ for all } X_{j,k-1}\in \Gamma_{j,k-1}$ we have
\bea
\Pb\left(\left\|\frac{1}{\alpha}\sum_{i_1=0}^{\alpha-1}Z_{2,i_1}\right\|\leq\epsilon\Big|X_{j,k-1}\right)\geq 1-2c\digamma\left(\alpha,\epsilon, 2\frac{b(1-b)(1-b^{2(\alpha-1)})}{1+b}c\gamma_{\new,k}^2\right)
\label{z_2_bound}
\eea
\bea
\Pb\left(\left\|\frac{1}{\alpha}\sum_{i_1=0}^{\alpha-1}Z_{3,i_1}\right\|\leq\epsilon\Big|X_{j,k-1}\right)\geq 1-2c\digamma\left(\alpha,\epsilon,2\frac{(1-b^{\alpha})^2}{1+b}c\gamma_{\new,k}^2\right)
\label{z_3_bound}
\eea
\bea
\Pb\left(\left\|\frac{1}{\alpha}\sum_{i=0}^{\alpha-1}(Z_{4,i}+Z_{4,i}')\right\|\leq\epsilon\Big|X_{j,k-1}\right)\geq 1-2c\digamma\left(\alpha,\epsilon,4\frac{b(1-b^{2\alpha})}{1+b}c\gamma_{\new,k}\gamma_{\new,k-1}\right)
\label{z_4_bound}
\eea

Thus, given $X_{j,k-1}$, we have
\bea
&&\Pb\left(\left\|\frac{1}{\alpha}\sum_{t\in \Ic_{j,k}}a_{t,\new}a_{t,\new}'\right\|\leq\lambda_{\new,k}^++\frac{b^2(1-b^{2\alpha})}{1-b^2}c\gamma_{\new,k-1}^2+\epsilon\Big|X_{j,k-1}\right)\nn\\
=&& 1- \Pb\left(\left\|\frac{1}{\alpha}\sum_{t\in \Ic_{j,k}}a_{t,\new}a_{t,\new}'\right\|>\lambda_{\new,k}^++\frac{b^2(1-b^{2\alpha})}{1-b^2}c\gamma_{\new,k-1}^2+\epsilon\Big|X_{j,k-1}\right)\nn\\
\geq&& 1- \Pb\left(\left\|\frac{1}{\alpha}\sum_{i=0}^{\alpha-1}Z_{1,i}\right\|>\lambda_{\new,k}^++\frac{\epsilon}{4}\Big|X_{j,k-1}\right) -\Pb\left(\left\|\frac{1}{\alpha}\sum_{i_1=0}^{\alpha-1}Z_{2,i_1}\right\|>\frac{\epsilon}{4}\bigg|X_{j,k-1}\right)
-\Pb\left(\left\|\frac{1}{\alpha}\sum_{i_1=0}^{\alpha-1}Z_{3,i}\right\|>\frac{\epsilon}{4}\bigg|X_{j,k-1}\right)\nn\\
&&-\Pb\left(\left\|\frac{1}{\alpha}\sum_{i_1=0}^{\alpha-1}(Z_{4,i}+Z_{4,i}')\right\|>\frac{\epsilon}{4}\bigg|X_{j,k-1}\right)- \Pb\left(\left\|\frac{b^2(1-b^{2\alpha})}{1-b^2}d_{-1}d_{-1}'\right\|>\frac{b^2(1-b^{2\alpha})}{1-b^2}c\gamma_{\new,k-1}^2\right)\nn\\
\geq&&1-c\digamma\left(\alpha,\frac{\epsilon}{4},\frac{(1-b^{2\alpha})(1-b)}{1+b}c\gamma_{\new,k}^2\right) - 2c\digamma\left(\alpha,\frac{\epsilon}{4}, 2\frac{b(1-b)(1-b^{2(\alpha-1)})}{1+b}c\gamma_{\new,k}^2\right) - 2c\digamma\left(\alpha,\frac{\epsilon}{4},2\frac{(1-b^{\alpha})^2}{1+b}c\gamma_{\new,k}^2\right) \nn\\&& -2c\digamma\left(\alpha,\frac{\epsilon}{4},4\frac{b(1-b^{2\alpha})}{1+b}c\gamma_{\new,k}\gamma_{\new,k-1}\right)
\label{a_t_new_upper_bound}
\eea

\bea
&&\Pb\left(\lambda_{\min}\left(\frac{1}{\alpha}\sum_{t\in \Ic_{j,k}}a_{t,\new}a_{t,\new}'\right)\geq(1-\frac{b^2-b^{2\alpha+2}}{\alpha(1-b^2)})\lambda_{\new,k}^--\epsilon\Big|X_{j,k-1}\right)\nn\\
=&& 1- \Pb\left(\lambda_{\min}\left(\frac{1}{\alpha}\sum_{t\in \Ic_{j,k}}a_{t,\new}a_{t,\new}'\right)<(1-\frac{b^2-b^{2\alpha+2}}{\alpha(1-b^2)})\lambda_{\new,k}^--\epsilon\Big|X_{j,k-1}\right)\nn\\
\geq&& 1- \Pb\left(\lambda_{\min}\left(\frac{1}{\alpha}\sum_{i=0}^{\alpha-1}Z_{1,i}\right)<(1-\frac{b^2-b^{2\alpha+2}}{\alpha(1-b^2)})\lambda_{\new,k}^--\frac{\epsilon}{4}\Big|X_{j,k-1}\right) -\Pb\left(\lambda_{\min}\left(\frac{1}{\alpha}\sum_{i_1=0}^{\alpha-1}Z_{2,i_1}\right)<-\frac{\epsilon}{4}\bigg|X_{j,k-1}\right)\nn\\
&&-\Pb\left(\lambda_{\min}\left(\frac{1}{\alpha}\sum_{i_1=0}^{\alpha-1}Z_{3,i}\right)<-\frac{\epsilon}{4}\bigg|X_{j,k-1}\right)
-\Pb\left(\lambda_{\min}\left(\frac{1}{\alpha}\sum_{i_1=0}^{\alpha-1}(Z_{4,i}+Z_{4,i}')\right)<-\frac{\epsilon}{4}\bigg|X_{j,k-1}\right)\nn\\
&&-\Pb\left(\lambda_{\min}\left(\frac{b^2(1-b^{2\alpha})}{1-b^2}d_{-1}d_{-1}'\right)<0\right)\nn\\
\geq&&1-c\digamma\left(\alpha,\frac{\epsilon}{4},\frac{(1-b^{2\alpha})(1-b)}{1+b}c\gamma_{\new,k}^2\right) - 2c\digamma\left(\alpha,\frac{\epsilon}{4}, 2\frac{b(1-b)(1-b^{2(\alpha-1)})}{1+b}c\gamma_{\new,k}^2\right) - 2c\digamma\left(\alpha,\frac{\epsilon}{4},2\frac{(1-b^{\alpha})^2}{1+b}c\gamma_{\new,k}^2\right) \nn\\&& -2c\digamma\left(\alpha,\frac{\epsilon}{4},4\frac{b(1-b^{2\alpha})}{1+b}c\gamma_{\new,k}\gamma_{\new,k-1}\right)
\label{a_t_new_lower_bound}
\eea

\begin{remark}
$d_{-1}$ corresponds to $a_{t_j+(k-1)\alpha-1}$. When $k=1$, there is no $d_{-1}$ in previous equations, in which case there is no $\frac{b^2(1-b^{2\alpha})}{1-b^2}c\gamma_{\new,k-1}^2$ in the upper bound (\ref{a_t_new_upper_bound}).
\label{a_t_new_k_1}
\end{remark}

\bea
&&\sum_{t=0}^{\alpha-1}d_{t}m_{t}'\nn\\
=&&\sum_{t=0}^{\alpha-1}(bd_{t-1}+\mu_{t})(bm_{t-1}+\omega_{t}')\nn\\
=&&\sum_{t=0}^{\alpha-1}(b^{t+1}d_{-1}+ \sum_{i=0}^{t}b^{t-i}\mu_{i})(b^{t+1} m_{-1}+\sum_{i=0}^{t}b^{t-i}\omega_{i}')\nn\\
=&&\sum_{t=0}^{\alpha-1}\sum_{i_1=0}^{t}\sum_{i_2=0}^{t}b^{2t-i_1-i_2}\mu_{i_1}\omega_{i_2}' + \sum_{t=0}^{\alpha-1}\sum_{i=0}^{t}b^{2t-i+1}\mu_{i}m_{-1}'+\sum_{t=0}^{\alpha-1}b^{2(t+1)}d_{-1}m_{-1}'+ \nn\\
&&\sum_{t=0}^{\alpha-1} \sum_{i=0}^{t} b^{2t-i+1} d_{-1}\omega_{i}'\nn\\
=&&\sum_{i_1=0}^{ \alpha-1}\sum_{t=i_1}^{ \alpha-1}\sum_{i_2=0}^{t}b^{2t-i_1-i_2}\mu_{i_1}\omega_{i_2}' + \sum_{i=0}^{ \alpha-1}\sum_{t=i}^{ \alpha-1}b^{2t-i+1}\mu_{i}m_{-1}'+\frac{b^2(1-b^{2\alpha})}{1-b^2}d_{-1}m_{-1}'+ \nn\\
&&\sum_{i=0}^{\alpha-1} \sum_{t=i}^{\alpha-1} b^{2t-i+1} d_{-1}\omega_{i}'\nn\\
=&&\sum_{i_1=0}^{ \alpha-1}\sum_{i_2=0}^{\alpha-1}\sum_{t=\max\{i_1,i_2\}}^{\alpha-1}b^{2t-i_1-i_2}\mu_{i_1}\omega_{i_2}' + \sum_{i=0}^{ \alpha-1}\frac{(1-b^{2(\alpha-i)})}{1-b^2}b^{i+1}\mu_{i}m_{-1}'+\frac{b^2(1-b^{2\alpha})}{1-b^2}d_{-1}m_{-1}' \nn\\
&&+\sum_{i=0}^{ \alpha-1}\frac{(1-b^{2(\alpha-i)})} {1-b^2}b^{i+1} d_{-1}\omega_{i}'\nn\\
=&&\sum_{i_1=0}^{ \alpha-1} \sum_{i_2=0}^{\alpha-1} \frac{(1-b^{2(\alpha-\max\{i_1,i_2\})})}{1-b^2}b^{2\max\{i_1,i_2\}-i_1-i_2}\mu_{i_1}\omega_{i_2}'+ \sum_{i=0}^{ \alpha-1}\frac{(1-b^{2(\alpha-i)})}{1-b^2}b^{i+1}(\mu_{i}m_{-1}'+ d_{-1}\omega_{i}')\nn\\
&&+ \frac{b^2(1-b^{2\alpha})}{1-b^2}d_{-1}m_{-1}'
\label{d_t_new_star}
\eea

Let
\bea
&&Y_{1,i}=\sum_{i_2=0}^{\alpha-1} \frac{(1-b^{2(\alpha-\max\{i,i_2\})})}{1-b^2}b^{2\max\{i,i_2\}-i-i_2}\mu_{i}\omega_{i_2}' \nn\\
&&Y_{2,i}=\frac{(1-b^{2(\alpha-i)})}{1-b^2}b^{i+1}(\mu_{i}m_{-1}'+d_{-1}\omega_{i}')\nn
\eea
then by Lemma \ref{ind_expec} and similar procedure to (\ref{exp_0}),
\bea
\E_{i-1}(Y_{1,i}|X_{j,k-1})=0, \E_{i-1}(Y_{2,i}|X_{j,k-1})=0
\eea
and
\bea
\|Y_{1,i_1}\| \leq \sum_{i_2=0}^{\alpha-1} \frac{(1-b^{2(\alpha-\max\{i_1,i_2\})})}{1-b^2}b^{2\max\{i_1,i_2\}-i_1-i_2}\sqrt{cr}(1-b)^2\gamma_*\gamma_{\new,k}
\eea
As
\bea
&&\sum_{i_2=0}^{\alpha-1} \frac{(1-b^{2(\alpha-\max\{i_1,i_2\})})}{1-b^2}b^{2\max\{i_1,i_2\}-i_1-i_2}\nn\\
&&=\sum_{i_2=0}^{i_1} \frac{(1-b^{2(\alpha-i_1)})}{1-b^2}b^{2i_1-i_1-i_2}+\sum_{i_2=i_1+1}^{\alpha-1} \frac{(1-b^{2(\alpha-i_2)})}{1-b^2}b^{2i_2-i_1-i_2}\nn\\
&&=\frac{(1-b^{2(\alpha-i_1)})b^{i_1}(1-b^{-(i_1+1)})}{(1-b^2)(1-b^{-1})}+\frac{b(1-b^{\alpha-1-i_1})+b^{2(\alpha)-2i_1}(1-b^{-(\alpha-i_1-1)})}{(1-b^2)(1-b)}\nn\\
&&=\frac{-b^{i_1+1}+1+b^{\alpha-i_1+\alpha+1}+b-b^{\alpha-i_1}-b^{\alpha-i_1+1}}{(1-b^2)(1-b)}\nn\\
&&\leq\frac{1+b-2\sqrt{b^{\alpha+1}(1+b-b^{\alpha+1})}}{(1-b^2)(1-b)},\nn
\eea
we have
$$\|Y_{1,i_1}\|_2\leq\frac{1+b-2\sqrt{b^{\alpha+1}(1+b-b^{\alpha+1})}}{1+b}\sqrt{cr}\gamma_*\gamma_{\new,k}.$$
Also,
\bea
\|Y_{2,i}\|_2&&\leq 2\frac{(1-b^{2(\alpha-i)})}{1-b^2}b^{+i+1}\sqrt{cr}(1-b)\gamma_*\gamma_{\new,k}\nn\\
&&\leq2\frac{b(1-b^{2\alpha})}{1+b}\sqrt{cr}\gamma_*\gamma_{\new,k}
\eea
\bea
\|\frac{b^2\left(1-b^{2\alpha}\right)}{1-b^2}d_{-1}m_{-1}'\|\leq\frac{b^2\left(1-b^{2\alpha}\right)}{1-b^2}\sqrt{cr}\gamma_*\gamma_{\new,k}
\eea

Thus, by Corollary \ref{azuma_rec}, we have
\bea
\Pb\left(\left\|\frac{1}{\alpha}\sum_{i_1\in \Ic_{j,k}}Y_{1,i_1}\right\|\leq\epsilon\Big|X_{j,k-1}\right)\geq 1-(r_j+c)\digamma(\alpha,\epsilon, 2\frac{1+b-2\sqrt{b^{\alpha+1}(1+b-b^{\alpha+1})}}{1+b}\sqrt{cr}\gamma_*\gamma_{\new,k}),
\label{y_1_bound}
\eea
for all $X_{j,k-1}\in \Gamma_{j,k-1}$; and
\bea
\Pb\left(\left\|\frac{1}{\alpha}\sum_{i_1\in \Ic_{j,k}}Y_{2,i_1}\right\|\leq\epsilon\Big|X_{j,k-1}\right)\geq 1-(r_j+c)\digamma\left(\alpha,\epsilon, 4\frac{b(1-b^{2\alpha})}{1+b}\sqrt{cr}\gamma_*\gamma_{\new,k}\right),
\label{y_2_bound}
\eea
for all $X_{j,k-1}\in \Gamma_{j,k-1}$.\\
Thus, given $X_{j,k-1}$, we have
\bea
&&\Pb\left(\lambda_{\max}\left(\frac{1}{\alpha} \sum_{t\in \Ic_{j,k}} a_{t,\new}{a_{t,*}}'\right)\leq\frac{b^2\left(1-b^{2\alpha}\right)}{1-b^2}\sqrt{cr}\gamma_*\gamma_{\new,k} +\epsilon\right)\nn\\
&&= 1- \Pb\left(\lambda_{\max}\left(\frac{1}{\alpha} \sum_{t\in \Ic_{j,k}} a_{t,\new}{a_{t,*}}'\right)>\frac{b^2\left(1-b^{2\alpha}\right)}{1-b^2}\sqrt{cr}\gamma_*\gamma_{\new,k} + \epsilon\right)\nn\\
&&\geq 1- \Pb\left(\lambda_{\max}\left(\frac{1}{\alpha}\sum_{i_1\in \Ic_{j,k}}Y_{1,i_1}\right)>\frac{\epsilon}{2}\Big|X_{j,k-1}\right) - \Pb\left(\lambda_{\max}\left(\frac{1}{\alpha}\sum_{i_1\in \Ic_{j,k}}Y_{2,i_1}\right)>\frac{\epsilon}{2}\Big|X_{j,k-1}\right) -\nn\\
&&\Pb\left(\lambda_{\max}\left(\frac{1}{\alpha}\sum_{i_1\in \Ic_{j,k}}\frac{b^2\left(1-b^{2\alpha}\right)}{1-b^2}d_{-1}m_{-1}'\right)>\frac{b^2\left(1-b^{2\alpha}\right)}{1-b^2}\sqrt{cr}\gamma_*\gamma_{\new,k}\Big|X_{j,k-1}\right)\nn\\
&&\geq 1- (r_j+c)\digamma(\alpha,\frac{\epsilon}{2}, 2\frac{1+b-2\sqrt{b^{\alpha+1}(1+b-b^{\alpha+1})}}{1+b}\sqrt{cr}\gamma_*\gamma_{\new,k}) -(r_j+c)\digamma\left(\alpha,\frac{\epsilon}{2}, 4\frac{b(1-b^{2\alpha})}{1+b}\sqrt{cr}\gamma_*\gamma_{\new,k}\right)
\label{a_t_new_star_upper_bound}
\eea

\bea
&&\Pb\left(\lambda_{\min}\left(\frac{1}{\alpha} \sum_{t\in \Ic_{j,k}} a_{t,\new}{a_{t,*}}'\right)\geq-\frac{b^2\left(1-b^{2\alpha}\right)}{1-b^2}\sqrt{cr}\gamma_*\gamma_{\new,k} - \epsilon\right)\nn\\
&&= 1- \Pb\left(\lambda_{\min}\left(\frac{1}{\alpha} \sum_{t\in \Ic_{j,k}} a_{t,\new}{a_{t,*}}'\right)<-\frac{b^2\left(1-b^{2\alpha}\right)}{1-b^2}\sqrt{cr}\gamma_*\gamma_{\new,k} - \epsilon\right)\nn\\
&&\geq 1- \Pb\left(\lambda_{\min}\left(\frac{1}{\alpha}\sum_{i_1\in \Ic_{j,k}}Y_{1,i_1}\right)<-\frac{\epsilon}{2}\Big|X_{j,k-1}\right) - \Pb\left(\lambda_{\min}\left(\frac{1}{\alpha}\sum_{i_1\in \Ic_{j,k}}Y_{2,i_1}\right)<-\frac{\epsilon}{2}\Big|X_{j,k-1}\right) -\nn\\
&&\Pb\left(\lambda_{\min}\left(\frac{1}{\alpha}\sum_{i_1\in \Ic_{j,k}}\frac{b^2\left(1-b^{2\alpha}\right)}{1-b^2}d_{-1}m_{-1}'\right)<-\frac{b^2\left(1-b^{2\alpha}\right)}{1-b^2}\sqrt{cr}\gamma_*\gamma_{\new,k}\Big|X_{j,k-1}\right)\geq 1- \nn\\
&& (r_j+c)\digamma\left(\alpha,\frac{\epsilon}{2}, 2\frac{1+b-2\sqrt{b^{\alpha+1}(1+b-b^{\alpha+1})}}{1+b}\sqrt{cr}\gamma_*\gamma_{\new,k}\right) -(r_j+c)\digamma\left(\alpha,\frac{\epsilon}{2}, 4\frac{b(1-b^{2\alpha})}{1+b}\sqrt{cr}\gamma_*\gamma_{\new,k}\right)
\label{a_t_new_star_lower_bound}
\eea

\bea
&&\sum_{t=0}^{\alpha-1}m_{t}m_{t}'\nn\\
=&&\sum_{t=0}^{\alpha-1}(bm_{t-1}+\omega_{t})(bm_{t-1}+\omega_{t}')\nn\\
=&&\sum_{t=0}^{\alpha-1}(b^{t+1}m_{-1}+ \sum_{i=0}^{t}b^{t-i}\omega_{i})(b^{t+1}m_{-1}+ \sum_{i=0}^{t}b^{t-i}\omega_{i}')\nn\\
=&&\sum_{t=0}^{\alpha-1}(b^{2t}\sum_{i=0}^{t}b^{-2i}\omega_{i}\omega_{i}')+\sum_{t=0}^{\alpha-1}(b^{2t}\sum_{0\le i_1, i_2 \le t, i_1\neq i_2}b^{-i_1-i_2}\omega_{i_1}\omega_{i_2}')+\sum_{t=0}^{\alpha-1}\sum_{i=0}^{t}b^{2t-i+1}m_{-1}\omega_{i}' \nn\\
 &&+ \sum_{t=0}^{\alpha-1}\sum_{i=0}^{t}b^{2t-i+1}\omega_{i}m_{-1}'+ \sum_{t=0}^{\alpha-1}b^{2(t+1)}m_{-1}m_{-1}'\nn\\
=&&\sum_{i=0}^{\alpha-1}\sum_{t=i }^{\alpha-1 }(b^{2t}b^{-2i}\omega_{i}\omega_{i}') + \sum_{t=0}^{\alpha-1}b^{2t} (\sum_{i_1=0}^{t}\sum_{i_2=0}^{i_1-1} b^{-i_1-i_2}\omega_{i_1}\omega_{i_2}' +\sum_{i_1=0}^{ t}\sum_{i_2= i_1+1}^{t} b^{-i_1-i_2}\omega_{i_1}\omega_{i_2}')+ \nn\\
&&\sum_{i=0}^{\alpha-1}\sum_{t=i}^{\alpha-1}b^{2t-i+1}m_{-1}\omega_{i}' + \sum_{i=0}^{\alpha-1}\sum_{t=i}^{\alpha-1}b^{2t-i+1}\omega_{i}m_{-1}' + \frac{b^2(1-b^{2\alpha})}{1-b^2}m_{-1}m_{-1}'\nn\\
=&&\sum_{i=0}^{\alpha-1}\frac{(1-b^{2(\alpha-i)})}{1-b^2}\omega_{i}\omega_{i}'+\sum_{i_1=0}^{\alpha-1}\sum_{t=i_1 }^{\alpha -1 }b^{2t}(\sum_{i_2=t_j+(k-1)}^{i_1-1} b^{-i_1-i_2}\omega_{i_1}\omega_{i_2}'+\sum_{i_2= i_1+1}^{t} b^{-i_1-i_2}\omega_{i_1}\omega_{i_2}')\nn\\
&&\sum_{i=0}^{\alpha-1}\frac{b^{i+1}(1-b^{2(\alpha-i)})}{1-b^2}(\omega_{i}m_{-1}' + m_{-1}\omega_{i}')+ \frac{b^2(1-b^{2\alpha})}{1-b^2}m_{-1}m_{-1}'\nn\\
=&&\sum_{i=0}^{\alpha-1}\frac{(1-b^{2(\alpha-i)})}{1-b^2}\omega_{i}\omega_{i}'+\sum_{i_1=0}^{\alpha-1}\frac{(1-b^{2(\alpha-i_1)})}{1-b^2} \sum_{i_2=0}^{i_1-1} b^{i_1-i_2}\omega_{i_1}\omega_{i_2}' + \sum_{i_1=-1}^{\alpha-1}\sum_{i_2=i_1+1}^{\alpha-1 }\sum_{t=i_2}^{\alpha-1} b^{2t}b^{-i_1-i_2}\omega_{i_1}\omega_{i_2}'+\nn\\ &&\sum_{i=0}^{\alpha-1}\frac{b^{i+1}(1-b^{2(\alpha-i)})}{1-b^2}(\omega_{i}m_{-1}' + m_{-1}\omega_{i}')+ \frac{b^2(1-b^{2\alpha})}{1-b^2}m_{-1}m_{-1}'\nn\\
=&&\sum_{i=0}^{\alpha-1}\frac{(1-b^{2(\alpha-i)})}{1-b^2}\omega_{i}\omega_{i}' + \sum_{i_1=0}^{\alpha-1}\frac{(1-b^{2(\alpha-i_1)})}{1-b^2}\sum_{i_2=0}^{i_1-1} b^{i_1-i_2}\omega_{i_1}\omega_{i_2}' + \sum_{i_1=0}^{\alpha-1}\sum_{i_2=i_1+1}^{\alpha-1 }\frac{(1-b^{2(\alpha-i_2)})}{1-b^2}b^{-i_1+i_2}\omega_{i_1}\omega_{i_2}'\nn\\ &&+\sum_{i=0}^{\alpha-1}\frac{b^{i+1}(1-b^{2(\alpha-i)})}{1-b^2}(\omega_{i}m_{-1}' + m_{-1}\omega_{i}')+ \frac{b^2(1-b^{2\alpha})}{1-b^2}m_{-1}m_{-1}'\nn\\
=&&\sum_{i=0}^{\alpha-1}\frac{(1-b^{2(\alpha-i)})}{1-b^2}\omega_{i}\omega_{i}' + \sum_{i_1=0}^{\alpha-1}\frac{(1-b^{2(\alpha-i_1)})}{1-b^2}\sum_{i_2=0}^{i_1-1} b^{i_1-i_2}\omega_{i_1}\omega_{i_2}' + \sum_{i_1=0}^{\alpha-1}\sum_{i_2=\alpha-i_1}^{\alpha-1 }\frac{(1-b^{2(\alpha-i_2)})}{1-b^2}b^{i_1+i_2+1-\alpha}\omega_{\alpha-1-i_1}\omega_{i_2}'\nn\\ &&+\sum_{i=0}^{\alpha-1}\frac{b^{i+1}(1-b^{2(\alpha-i)})}{1-b^2}(\omega_{i}m_{-1}' + m_{-1}\omega_{i}')+ \frac{b^2(1-b^{2\alpha})}{1-b^2}m_{-1}m_{-1}'
\label{d_t_star}
\eea

Let
\bea
&&Z_{1,i}=\frac{(1-b^{2(\alpha-i)})}{1-b^2}\omega_{i}\omega_{i}', \nn\\
&&Z_{2,i}=\sum_{i_2=0}^{i-1}\frac{(1-b^{2(\alpha-i)})}{1-b^2}b^{i-i_2}\omega_{i}\omega_{i_2}', i\geq 1  \nn\\
&&Z_{3,i}=\sum_{i_2=\alpha-i}^{\alpha-1 }\frac{(1-b^{2(\alpha-i_2)})}{1-b^2}b^{i+i_2+1-\alpha}\omega_{\alpha-1-i}\omega_{i_2}',\nn\\
&&Z_{4,i}=\frac{b^{i+1}(1-b^{2(\alpha-i)})}{1-b^2}\omega_{i}m_{-1}'
\eea
then,
\bea
&&\lambda_{\max}\left(\E_{i-1}(\frac{1}{\alpha}\sum_{i=0}^{\alpha-1}Z_{1,i}\Big|X_{j,k-1})\right)\nn\\
\leq&&\frac{1}{\alpha}\sum_{i=0}^{\alpha-1}\frac{(1-b^{2(t_j+k\alpha-i)})}{1-b^2}\lambda_{\max}(\E(\omega_{i}\omega_{i}'))\nn\\
\leq&&(\frac{\alpha(1-b^2)-b^2+b^{2\alpha+2}}{\alpha(1-b^2)}\lambda^+)\nn\\
\leq&&\lambda^+\nn
\eea
\bea
&&\lambda_{\min}\left(\E_{i-1}(\frac{1}{\alpha}\sum_{i=0}^{\alpha-1}Z_{1,i}\Big|X_{j,k-1})\right)\nn\\
\geq&&\frac{1}{\alpha}\sum_{i=0}^{\alpha-1}\frac{(1-b^{2(t_j+k\alpha-i)})}{1-b^2}\lambda_{\min}(\E(\omega_{i}\omega_{i}'))\nn\\
\geq&&\frac{\alpha(1-b^2)-b^2+b^{2\alpha+2}}{\alpha(1-b^2)}\lambda_{\new,k}^-
\eea
and by Lemma \ref{ind_expec} and similar procedure to (\ref{exp_0}),
\bea
\E_{i-1} (Z_{2,i}|X_{j,k-1})=0, \E_{i-1} (Z_{3,i}|X_{j,k-1})=0, \E_{i-1}(Z_{4,i}|X_{j,k-1})=0.
\eea

$$0\preceq  Z_{1,i_1}\preceq\frac{(1-b^{2(\alpha-i)})}{1-b^2} r(1-b)^2\gamma_*^2I\preceq \frac{(1-b^{2\alpha})(1-b)}{1+b}r\gamma_*^2 I,$$
\bea
 \|Z_{2,i}\|&&\leq\sum_{i_2=0}^{i-1}\frac{(1-b^{2(\alpha-i)})}{1-b^2}b^{i-i_2}r(1-b)^2\gamma_*^2\nn\\
&&\leq \frac{b(1-b^{2(\alpha-i)})(1-b^{i})}{1+b}r\gamma_*^2\nn\\
&&\leq \frac{b(1-b)(1-b^{2(\alpha-1)})}{(1+b)}r\gamma_*^2\nn
\eea

\bea
\|Z_{3,i}\|&&\leq\sum_{i_2=\alpha-i}^{\alpha-1 }\frac{(1-b^{2(\alpha-i_2)})}{1-b^2}b^{i+i_2+1-\alpha}r(1-b)^2\gamma_*^2\nn\\
&&\leq \frac{(1-b^{i+1})^2}{1+b}r\gamma_*^2\nn\\
&&\leq \frac{(1-b^{\alpha})^2}{1+b}r\gamma_*^2\nn
\eea

\bea
 \|Z_{4,i_1}\|&&\leq\frac{b^{i+1}(1-b^{2(\alpha-i)})}{1-b^2}c(1-b)\gamma_*^2\nn\\
&&\leq \frac{b(1-b^{2\alpha})}{1+b}r\gamma_*^2\nn
\eea

\bea
0 \preceq \frac{b^2(1-b^{2\alpha})}{1-b^2}m_{-1}m_{-1}' \preceq \frac{b^2(1-b^{2\alpha})}{1-b^2}r\gamma_*^2I
\eea

Given $X_{j,k-1}$, $\{Z_{1,i}\}_{i\in \Ic_{j,k}}$ is an adapted sequence of Hermitian matrices, thus by Corollary \ref{azuma_nonzero}, for all $X_{j,k-1}\in \Gamma_{j,k-1}$, we have

\bea
\Pb&&\left(\lambda_{\max}\left(\frac{1}{\alpha}\sum_{i=0}^{\alpha-1} Z_{1,i}\right)\leq\lambda^+-\epsilon\Big|X_{j,k-1}\right)\geq  1-r_j\digamma(\alpha,\epsilon,\frac{(1-b^{2\alpha})(1-b)}{1+b}r\gamma_*^2 )
\label{z_1_bound_2}
\eea
Similarly, given $X_{j,k-1}$, $\{Z_{2,i_1}\}_{i_1\in \Ic_{j,k}}, \{Z_{3,i_1}\}_{i_1\in \Ic_{j,k}}, \{Z_{4,i}\}_{i\in \Ic_{j,k}}$ are adapted sequences of matrices, by Corollary \ref{azuma_rec}, $\text{ for all } X_{j,k-1}\in \Gamma_{j,k-1}$ we have
\bea
\Pb&&\left(\left\|\frac{1}{\alpha}\sum_{i=0}^{\alpha-1} Z_{2,i}\right\|\leq\epsilon\Big|X_{j,k-1}\right)\geq 1-2r_j\digamma(\alpha,\epsilon,2\frac{b(1-b)(1-b^{2(\alpha-1)})}{(1+b)}r\gamma_*^2)
\label{z_2_bound_2}
\eea
\bea
\Pb\left(\left\|\frac{1}{\alpha}\sum_{i=0}^{\alpha-1} Z_{3,i}\right\|\leq\epsilon\bigg|X_{j,k-1}\right)\geq 1-2r_j\digamma(\alpha,\epsilon,2\frac{(1-b^{\alpha})^2}{1+b}r\gamma_*^2)
\label{z_3_bound_2}
\eea
\bea
\Pb&&\left(\left\|\frac{1}{\alpha}\sum_{i=0}^{\alpha-1} (Z_{4,i}+Z_{4,i}')\right\|\leq\epsilon\bigg|X_{j,k-1}\right)\geq 1-2r_j\digamma(\alpha,\epsilon,4\frac{b(1-b^{2\alpha})}{1+b}r\gamma_*^2)
\label{z_4_bound_2}
\eea

Thus, given $X_{j,k-1}$, we have
\bea
&&\Pb\left(\left\|\frac{1}{\alpha}\sum_{t\in \Ic_{j,k}}a_{t,*}a_{t,*}'\right\|\leq\lambda^++\frac{b^2(1-b^{2\alpha})}{1-b^2}r\gamma_*^2+\epsilon\Big|X_{j,k-1}\right)\nn\\
\geq&& 1- \Pb\left(\left\|\frac{1}{\alpha}\sum_{t\in \Ic_{j,k}}a_{t,*}a_{t,*}'\right\|>\lambda^++\frac{b^2(1-b^{2\alpha})}{1-b^2}c\gamma_{*,k-1}^2+\epsilon\Big|X_{j,k-1}\right)\nn\\
\geq&& 1- \Pb\left(\left\|\frac{1}{\alpha}\sum_{i=0}^{\alpha-1}Z_{1,i}\right\|>\lambda^++\frac{\epsilon}{4}\Big|X_{j,k-1}\right) -\Pb\left(\left\|\frac{1}{\alpha}\sum_{i_1=0}^{\alpha-1}Z_{2,i_1}\right\|>\frac{\epsilon}{4}\bigg|X_{j,k-1}\right)
-\Pb\left(\left\|\frac{1}{\alpha}\sum_{i_1=0}^{\alpha-1}Z_{3,i}\right\|>\frac{\epsilon}{4}\bigg|X_{j,k-1}\right)\nn\\
&&-\Pb\left(\left\|\frac{1}{\alpha}\sum_{i_1=0}^{\alpha-1}(Z_{4,i}+Z_{4,i}')\right\|>\frac{\epsilon}{4}\bigg|X_{j,k-1}\right)- \Pb\left(\left\|\frac{b^2(1-b^{2\alpha})}{1-b^2}m_{-1}m_{-1}'\right\|>\frac{b^2(1-b^{2\alpha})}{1-b^2}r\gamma_*^2\right)\nn\\
\geq&&1-r_j\digamma\left(\alpha,\frac{\epsilon}{4},\frac{(1-b^{2\alpha})(1-b)}{1+b}r\gamma_*^2\right) - 2r_j\digamma\left(\alpha,\frac{\epsilon}{4}, 2\frac{b(1-b)(1-b^{2(\alpha-1)})}{(1+b)}r\gamma_*^2\right) \nn\\&&- 2r_j\digamma\left(\alpha,\frac{\epsilon}{4},2\frac{(1-b^{\alpha})^2}{1+b}r\gamma_*^2\right) -2r_j\digamma\left(\alpha,\frac{\epsilon}{4},4\frac{b(1-b^{2\alpha})}{1+b}r\gamma_*^2\right)
\label{a_t_star_upper_bound}
\eea

\bea
&&\Pb\left(\lambda_{\min}\left(\frac{1}{\alpha}\sum_{t\in \Ic_{j,k}}a_{t,*}a_{t,*}'\right)\geq\frac{\alpha(1-b^2)-b^2+b^{2\alpha+2}}{\alpha(1-b^2)}\lambda_{\new,k}^--\epsilon\Big|X_{j,k-1}\right)\nn\\
\geq&& 1- \Pb\left(\lambda_{\min}\left(\frac{1}{\alpha}\sum_{t\in \Ic_{j,k}}a_{t,*}a_{t,*}'\right)<\frac{\alpha(1-b^2)-b^2+b^{2\alpha+2}}{\alpha(1-b^2)}\lambda_{\new,k}^--\epsilon\Big|X_{j,k-1}\right)\nn\\
\geq&& 1- \Pb\left(\lambda_{\min}\left(\frac{1}{\alpha}\sum_{i=0}^{\alpha-1}Z_{1,i}\right)<\frac{\alpha(1-b^2)-b^2+b^{2\alpha+2}}{\alpha(1-b^2)}\lambda_{\new,k}^- - \frac{\epsilon}{4}\Big|X_{j,k-1}\right) -\Pb\left(\lambda_{\min}\left(\frac{1}{\alpha}\sum_{i_1=0}^{\alpha-1}Z_{2,i_1}\right)<-\frac{\epsilon}{4}\bigg|X_{j,k-1}\right)
\nn\\&&-\Pb\left(\lambda_{\min}\left(\frac{1}{\alpha}\sum_{i_1=0}^{\alpha-1}Z_{3,i}\right)<-\frac{\epsilon}{4}\bigg|X_{j,k-1}\right)-
\Pb\left(\lambda_{\min}\left(\frac{1}{\alpha}\sum_{i_1=0}^{\alpha-1}(Z_{4,i}+Z_{4,i}')\right)<-\frac{\epsilon}{4}\bigg|X_{j,k-1}\right)\nn\\&&- \Pb\left(\lambda_{\min}\left(\frac{b^2(1-b^{2\alpha})}{1-b^2}m_{-1}m_{-1}'\right)<0\right)\nn\\
\geq&&1-r_j\digamma\left(\alpha,\frac{\epsilon}{4},\frac{(1-b^{2\alpha})(1-b)}{1+b}r\gamma_*^2\right) - 2r_j\digamma\left(\alpha,\frac{\epsilon}{4}, 2\frac{b(1-b)(1-b^{2(\alpha-1)})}{(1+b)}r\gamma_*^2\right) - 2r_j\digamma\left(\alpha,\frac{\epsilon}{4},2\frac{(1-b^{\alpha})^2}{1+b}r\gamma_*^2\right) \nn\\&& -2r_j\digamma\left(\alpha,\frac{\epsilon}{4},4\frac{b(1-b^{2\alpha})}{1+b}r\gamma_*^2\right)
\label{a_t_star_lower_bound}
\eea

\subsection{$\lambda_{\min}\left(A_k\right)$}
\label{Akbound}
Consider $A_k := \frac{1}{\alpha} \sum_{t\in \Ic_{j,k}} {E_{\new}}' \Phi_{0} L_t {L_t}' \Phi_{0} E_{\new}$. Notice that ${E_{\new}}' \Phi_{0} L_t = R_{\new} a_{t,\new} + {E_{\new}}' D_* a_{t,*}$. We have
\beq
A_k  \succeq \frac{1}{\alpha} \sum_{t\in \Ic_{j,k}} R_{\new} a_{t,\new} {a_{t,\new}}' {R_{\new}}' + \frac{1}{\alpha} \sum_{t\in \Ic_{j,k}} \left(R_{\new} a_{t,\new}{a_{t,*}}' {D_*}' {E_{\new}}' +  {E_{\new}}' D_* a_{t,*}{a_{t,\new}}' {R_{\new}}'\right) \label{lemmabound_1}
\eeq
By $\lambda_{\min}(R_{\new}R_{\new}') \geq 1-(\zeta_*^+)^2, \|R_{\new}\|\leq 1$ and similar procedure to get (\ref{a_t_new_lower_bound}), we have

\bea
&&\Pb\left(\lambda_{\min}\left(\frac{1}{\alpha}\sum_{t\in \Ic_{j,k}} R_{\new}a_{t,\new}a_{t,\new}'R_{\new}'\right)\geq(1-(\zeta_*^+)^2)(1-\frac{b^2-b^{2\alpha+2}}{\alpha(1-b^2)})\lambda_{\new,k}^--\frac{c\zeta\lambda^-}{24}\Big|X_{j,k-1}\right)\nn\\
\geq&&1-c\digamma\left(\alpha,\frac{c\zeta\lambda^-}{96},\frac{(1-b^{2\alpha})(1-b)}{1+b}c\gamma_{\new,k}^2\right) - 2c\digamma\left(\alpha,\frac{c\zeta\lambda^-}{96}, 2\frac{b(1-b)(1-b^{2(\alpha-1)})}{1+b}c\gamma_{\new,k}^2\right) \nn\\&&- 2c\digamma\left(\alpha,\frac{c\zeta\lambda^-}{96},2\frac{(1-b^{\alpha})^2}{1+b}c\gamma_{\new,k}^2\right) -2c\digamma\left(\alpha,\frac{c\zeta\lambda^-}{96},4\frac{b(1-b^{2\alpha})}{1+b}c\gamma_{\new,k}\gamma_{\new,k-1}\right)
\label{A_k_1}
\eea

By $\|R_{\new}\| \leq 1, \|E_{\new}\|\leq 1, \|D_*\|\leq\zeta_*^+$ and similar procedure to get (\ref{a_t_new_star_lower_bound}), we have
\bea
&&\Pb\left(\lambda_{\min}\left(\frac{1}{\alpha} \sum_{t\in \Ic_{j,k}} \left(R_{\new} a_{t,\new}{a_{t,*}}' {D_*}' {E_{\new}}' +  {E_{\new}}' D_* a_{t,*}{a_{t,\new}}' {R_{\new}}'\right)\right)\geq -2\frac{b^2\left(1-b^{2\alpha}\right)}{1-b^2}\sqrt{cr}\gamma_*\gamma_{\new,k}\zeta_*^+  -\frac{c\zeta\lambda^-}{24}\right)\geq \nn\\
&& 1 - 2c\digamma\left(\alpha,\frac{c\zeta\lambda^-}{48}, 4\frac{1+b-2\sqrt{b^{\alpha+1}(1+b-b^{\alpha+1})}}{1+b}\sqrt{cr}\gamma_*\gamma_{\new,k}\zeta_*^+\right) -2c\digamma\left(\alpha,\frac{c\zeta\lambda^-}{48}, 8\frac{b(1-b^{2\alpha})}{1+b}\sqrt{cr}\gamma_*\gamma_{\new,k}\zeta_*^+\right)
\label{A_k_2}
\eea

Combining the previous two inequalities, we have,
\bea
\Pb\left(\lambda_{\min}(A_k)\geq b_{A_k} - \frac{c\zeta\lambda^-}{12}\right)\geq 1-p_a(\alpha,\zeta)
\label{A_k_bound}
\eea
where $b_{A_k}=(1-(\zeta_*^+)^2)(1-\frac{b^2-b^{2\alpha+2}}{\alpha(1-b^2)})\lambda_{\new,k}^--2\frac{b^2\left(1-b^{2\alpha}\right)}{1-b^2}\sqrt{cr}\gamma_*\gamma_{\new,k}\zeta_*^+$ and
\bea
p_a(\alpha,\zeta)= &&c\digamma\left(\alpha,\frac{c\zeta\lambda^-}{96},\frac{(1-b^{2\alpha})(1-b)}{1+b}c\gamma_{\new,k}^2\right) + 2c\digamma\left(\alpha,\frac{c\zeta\lambda^-}{96}, 2\frac{b(1-b)(1-b^{2(\alpha-1)})}{1+b}c\gamma_{\new,k}^2\right) \nn\\&&+ 2c\digamma\left(\alpha,\frac{c\zeta\lambda^-}{96},2\frac{(1-b^{\alpha})^2}{1+b}c\gamma_{\new,k}^2\right) +2c\digamma\left(\alpha,\frac{c\zeta\lambda^-}{96},4\frac{b(1-b^{2\alpha})}{1+b}c\gamma_{\new,k}\gamma_{\new,k-1}\right) +2c\digamma\bigg(\alpha,\frac{c\zeta\lambda^-}{48},\nn\\&& 4\frac{1+b-2\sqrt{b^{\alpha+1}(1+b-b^{\alpha+1})}}{1+b}\sqrt{cr}\gamma_*\gamma_{\new,k}\zeta_*^+\bigg) +2c\digamma\left(\alpha,\frac{c\zeta\lambda^-}{48}, 8\frac{b(1-b^{2\alpha})}{1+b}\sqrt{cr}\gamma_*\gamma_{\new,k}\zeta_*^+\right)
\eea

\subsection{$\lambda_{\max}\left(A_{k,\perp}\right)$}
\label{Akperpbound}
\bea
A_{k,\perp} :&&= \frac{1}{\alpha} \sum_{t\in t_{j,k}} {E_{\new,\perp}}' \Phi_{0} L_t {L_t}' \Phi_{0} E_{\new,\perp} \nn\\
&&= \frac{1}{\alpha} \sum_{t\in t_{j,k}} {E_{\new,\perp}}' D_*a_{t,*}a_{t,*}'D_*' E_{\new,\perp}\nn
\eea
By $\|E_{\new,\perp}\|\leq 1, \|D_*\|\leq \zeta_*^+$ and similar procedure to get (\ref{a_t_star_upper_bound}), we have
\bea
&&\Pb\left(\left\|\frac{1}{\alpha}\sum_{t\in \Ic_{j,k}}{E_{\new,\perp}}' D_*a_{t,*}a_{t,*}'D_*' E_{\new,\perp}\right\|\leq b_{A_{k,\perp}}+\frac{c\zeta\lambda^-}{24}\Big|X_{j,k-1}\right)\geq1-p_b(\alpha,\zeta)
\label{A_k_perp_bound}
\eea
where $b_{A_{k,\perp}}=\lambda^+(\zeta_*^+)^2+\frac{b^2(1-b^{2\alpha})}{1-b^2}r\gamma_*^2(\zeta_*^+)^2$ and
\bea
p_b(\alpha,\zeta)=&&(n-c)\digamma\left(\alpha,\frac{c\zeta\lambda^-}{96},\frac{(1-b^{2\alpha})(1-b)}{1+b}r\gamma_*^2(\zeta_*^+)^2\right) + 2(n-c)\digamma\left(\alpha,\frac{c\zeta\lambda^-}{96}, 2\frac{b(1-b)(1-b^{2(\alpha-1)})}{(1+b)}r\gamma_*^2(\zeta_*^+)^2\right) \nn\\&&+ 2(n-c)\digamma\left(\alpha,\frac{c\zeta\lambda^-}{96},2\frac{(1-b^{\alpha})^2}{1+b}r\gamma_*^2(\zeta_*^+)^2\right) +2(n-c)\digamma\left(\alpha,\frac{c\zeta\lambda^-}{96},4\frac{b(1-b^{2\alpha})}{1+b}r\gamma_*^2(\zeta_*^+)^2\right)
\eea

\subsection{$\|\Hc_k\|_2$}
\label{Hkbound}
In this proof, we frequently refer to items from Sec. \ref{keyfacts} and the bounds (\ref{err_bound}).

For the second claim, using the expression for $\mathcal{H}_k$ given in Definition \ref{defHk}, it is easy to see that
\bea
\|\mathcal{H}_k \|_2 &\leq&  \max\{ \|H_k\|_2, \|H_{k,\perp}\|_2 \} + \|B_k\|_2 \leq \|\frac{1}{\alpha} \sum_t e_t {e_t}'\|_2 +  \max\left(\|T2\|_2, \|T4\|_2\right) + \|B_k\|_2
\label{add_calH1}
\eea
where $T2:= \frac{1}{\alpha} \sum_t  {E_{\new}}' \Phi_{0}\left( L_t {e_t}' + e_t {L_t}'\right)\Phi_{0} E_{\new}$ and $T4 :=\frac{1}{\alpha} \sum_t {E_{\new,\perp}}'\Phi_{0} \left(L_t {e_t}' + {e_t}'L_t\right)\Phi_{0} E_{\new,\perp}$. The second inequality follows by using the facts that (i) $H_k = T1 - T2$ where $T1 := \frac{1}{\alpha} \sum_t {E_{\new}}' \Phi_{0} e_t {e_t}'\Phi_{0} E_{\new}$, (ii) $H_{k,\perp} = T3 - T4$ where $T3 := \frac{1}{\alpha} \sum_t {E_{\new,\perp}}'\Phi_0 e_t {e_t}'\Phi_0  E_{\new,\perp}$, and (iii) $\max\left(\|T1\|_2, \|T3\|_2\right) \le \|\frac{1}{\alpha} \sum_t e_t {e_t}'\|_2$.
Next, we obtain high probability bounds on each of the terms on the RHS of (\ref{add_calH1}) using the Azuma corollaries.

\subsubsection{$\|\frac{1}{\alpha} \sum_t e_t {e_t}'\|_2$}
Consider $\|\frac{1}{\alpha} \sum_t e_t {e_t}'\|_2$. Let $Z_t = e_t {e_t}'$.
Then
\bea
Z_t=&&{I_{T_t}} [{\left(\Phi_{k-1}\right)_{T_t}}'\left(\Phi_{k-1}\right)_{T_t}]^{-1} {I_{T_t}}'[ \left(\Phi_{k-1} P_{*}\right) a_{t,*} + D_{\new,k-1} a_{t,\new}][ \left(\Phi_{k-1} P_{*}\right) a_{t,*} + D_{\new,k-1} a_{t,\new}]'{I_{T_t}} [{\left(\Phi_{k-1}\right)_{T_t}}'\left(\Phi_{k-1}\right)_{T_t}]^{-1} {I_{T_t}}'\nn\\
=&&{I_{T_t}} [{\left(\Phi_{k-1}\right)_{T_t}}'\left(\Phi_{k-1}\right)_{T_t}]^{-1} {I_{T_t}}'[ \left(\Phi_{k-1} P_{*}\right) a_{t,*}a_{t,*}' P_*' \Phi_{k-1} + D_{\new,k-1} a_{t,\new} a_{t,\new}' D_{\new,k-1}' + \left(\Phi_{k-1} P_{*}\right) a_{t,*}a_{t,\new}' D_{\new,k-1}'\nn\\&& + D_{\new,k-1} a_{t,\new}a_{t,*}' P_*' \Phi_{k-1}]{I_{T_t}} [{\left(\Phi_{k-1}\right)_{T_t}}'\left(\Phi_{k-1}\right)_{T_t}]^{-1} {I_{T_t}}'
\eea
(a) By (\ref{err_bound}) and similar procedure to get (\ref{a_t_star_upper_bound}), we have
$$\|{I_{T_t}} [{\left(\Phi_{k-1}\right)_{T_t}}'\left(\Phi_{k-1}\right)_{T_t}]^{-1} {I_{T_t}}' \left(\Phi_{k-1} P_{*}\right)\|\leq \phi^+\zeta_*^+$$
and
\bea
\Pb&&\Bigg(\left\|\frac{1}{\alpha}\sum_{t\in \Ic_{j,k}}{I_{T_t}} {\left(\Phi_{k-1}\right)_{T_t}}'\left(\Phi_{k-1}\right)_{T_t}]^{-1} {I_{T_t}}'[ \left(\Phi_{k-1} P_{*}\right) a_{t,*}a_{t,*}' P_*' \Phi_{k-1}{I_{T_t}} [{\left(\Phi_{k-1}\right)_{T_t}}'\left(\Phi_{k-1}\right)_{T_t}]^{-1} {I_{T_t}}'\right\|\nn\\
&&\leq\lambda^+(\phi^+\zeta_*^+)^2+\frac{b^2(1-b^{2\alpha})}{1-b^2}r\gamma_*^2(\phi^+\zeta_*^+)^2+\frac{c\zeta\lambda^-}{72}\Big|X_{j,k-1}\Bigg)\nn\\
\geq&& 1- p_{c_1}(\alpha,\zeta)
\label{e_t_1}
\eea
where
\bea
p_{c_1}(\alpha,\zeta)=&&n\digamma\left(\alpha,\frac{c\zeta\lambda^-}{288},\frac{(1-b^{2\alpha})(1-b)}{1+b}r\gamma_*^2(\phi^+\zeta_*^+)^2\right) + 2n\digamma\left(\alpha,\frac{c\zeta\lambda^-}{288}, 2\frac{b(1-b)(1-b^{2(\alpha-1)})}{(1+b)}r\gamma_*^2(\phi^+\zeta_*^+)^2\right) \nn\\&&+ 2n\digamma\left(\alpha,\frac{c\zeta\lambda^-}{288},2\frac{(1-b^{\alpha})^2}{1+b}r\gamma_*^2(\phi^+\zeta_*^+)^2\right) +2n\digamma\left(\alpha,\frac{c\zeta\lambda^-}{288},4\frac{b(1-b^{2\alpha})}{1+b}r\gamma_*^2(\phi^+\zeta_*^+)^2\right)
\eea
(b) By (\ref{err_bound}) and similar procedure to get (\ref{a_t_new_upper_bound}), we have
$$\|{I_{T_t}} [{\left(\Phi_{k-1}\right)_{T_t}}'\left(\Phi_{k-1}\right)_{T_t}]^{-1} {I_{T_t}}' D_{\new,k-1}\|\leq \phi^+\kappa_s^+\zeta_{k-1}^+$$
and
\bea
\Pb&&\Bigg(\left\|\frac{1}{\alpha}\sum_{t\in \Ic_{j,k}}{I_{T_t}} [{\left(\Phi_{k-1}\right)_{T_t}}'\left(\Phi_{k-1}\right)_{T_t}]^{-1} {I_{T_t}}' D_{\new,k-1} a_{t,\new} a_{t,\new}' D_{\new,k-1}'{I_{T_t}} [{\left(\Phi_{k-1}\right)_{T_t}}'\left(\Phi_{k-1}\right)_{T_t}]^{-1} {I_{T_t}}'\right\|\nn\\&&\leq\lambda_{\new,k}^+(\phi^+\kappa_s^+\zeta_{k-1}^+)^2+\frac{b^2(1-b^{2\alpha})}{1-b^2}c\gamma_{\new,k-1}^2(\phi^+\kappa_s^+\zeta_{k-1}^+)^2+\frac{c\zeta\lambda^-}{72}\Big|X_{j,k-1}\Bigg)\nn\\
\geq&&1-p_{c_2}(\alpha,\zeta)
\label{e_t_2}
\eea
where
\bea
&&p_{c_2}(\alpha,\zeta)=n\digamma\left(\alpha,\frac{c\zeta\lambda^-}{288},\frac{(1-b^{2\alpha})(1-b)}{1+b}c\gamma_{\new,k}^2(\phi^+\kappa_s^+\zeta_{k-1}^+)^2\right) + 2n\digamma\Bigg(\alpha,\frac{c\zeta\lambda^-}{288}, 2\frac{b(1-b)(1-b^{2(\alpha-1)})}{1+b}c\gamma_{\new,k}^2\left(\phi^+\kappa_s^+\right.\nn\\&&\left.\zeta_{k-1}^+\right)^2\Bigg)+ 2n\digamma\left(\alpha,\frac{c\zeta\lambda^-}{288},2\frac{(1-b^{\alpha})^2}{1+b}c\gamma_{\new,k}^2(\phi^+\kappa_s^+\zeta_{k-1}^+)^2\right) +2n\digamma\left(\alpha,\frac{c\zeta\lambda^-}{288},4\frac{b(1-b^{2\alpha})}{1+b}c\gamma_{\new,k}\gamma_{\new,k-1}(\phi^+\kappa_s^+\zeta_{k-1}^+)^2\right)\nn
\eea
(c) By (\ref{err_bound}) and similar procedure to get (\ref{a_t_new_star_upper_bound}), we have
$$\|{I_{T_t}} [{\left(\Phi_{k-1}\right)_{T_t}}'\left(\Phi_{k-1}\right)_{T_t}]^{-1}\|\leq \phi^+, \|{I_{T_t}}' D_{\new,k-1}\|\leq \kappa_s^+\zeta_{k-1}^+, \|\Phi_{k-1} P_{*}\|\leq \zeta_*^+$$
and
\bea
\Pb&&\Bigg(\bigg\|\frac{1}{\alpha} \sum_{t\in \Ic_{j,k}} {I_{T_t}} [{\left(\Phi_{k-1}\right)_{T_t}}'\left(\Phi_{k-1}\right)_{T_t}]^{-1} {I_{T_t}}'[\left(\Phi_{k-1} P_{*}\right) a_{t,*}a_{t,\new}' D_{\new,k-1}' + D_{\new,k-1} a_{t,\new}a_{t,*}' P_*' \Phi_{k-1}]\nn\\&&{I_{T_t}} [{\left(\Phi_{k-1}\right)_{T_t}}'\left(\Phi_{k-1}\right)_{T_t}]^{-1} {I_{T_t}}'\bigg\|\leq2\frac{b^2\left(1-b^{2\alpha}\right)}{1-b^2}\sqrt{cr}\gamma_*\gamma_{\new,k}(\phi^+)^2\kappa_s^+\zeta_{k-1}^+\zeta_*^+ +\frac{c\zeta\lambda^-}{72}\Big|X_{j,k-1}\Bigg)\nn\\
\geq&& 1- p_{c_3}(\alpha,\zeta)
\label{e_t_3}
\eea
where
\bea
p_{c_3}(\alpha,\zeta)=&&2n\digamma(\alpha,\frac{c\zeta\lambda^-}{144}, 4\frac{1+b-2\sqrt{b^{\alpha+1}(1+b-b^{\alpha+1})}}{1+b}\sqrt{cr}\gamma_*\gamma_{\new,k}(\phi^+)^2\kappa_s^+\zeta_{k-1}^+\zeta_*^+) +\nn\\ &&2n\digamma\left(\alpha,\frac{c\zeta\lambda^-}{144}, 8\frac{b(1-b^{2\alpha})}{1+b}\sqrt{cr}\gamma_*\gamma_{\new,k}(\phi^+)^2\kappa_s^+\zeta_{k-1}^+\zeta_*^+\right)
\eea
Combining (\ref{e_t_1}), (\ref{e_t_2}) and (\ref{e_t_3}), we have
\bea
\Pb&&\Bigg(\left\|\frac{1}{\alpha} \sum_{t\in \Ic_{j,k}} e_t {e_t}'\right\|\leq\lambda^+(\phi^+\zeta_*^+)^2+\frac{b^2(1-b^{2\alpha})}{1-b^2}r\gamma_*^2(\phi^+\zeta_*^+)^2+\lambda_{\new,k}^+(\phi^+\kappa_s^+\zeta_{k-1}^+)^2+\frac{b^2(1-b^{2\alpha})}{1-b^2}c\gamma_{\new,k-1}^2(\phi^+\kappa_s^+\zeta_{k-1}^+)^2 +\nn\\&&2\frac{b^2\left(1-b^{2\alpha}\right)}{1-b^2}\sqrt{cr}\gamma_*\gamma_{\new,k}(\phi^+)^2\kappa_s^+\zeta_{k-1}^+\zeta_*^+ +\frac{c\zeta\lambda^-}{24}\Big|X_{j,k-1}\Bigg)
\geq 1- p_{c}(\alpha,\zeta)
\label{e_t_bound}
\eea
where $p_c(\alpha,\zeta)=p_{c_1}(\alpha,\zeta)+p_{c_2}(\alpha,\zeta)+p_{c_3}(\alpha,\zeta)$.

\subsubsection{$\|T_2\|$}
Consider $T2$. Let
\bea
Z_t:=&&{E_{\new}}' \Phi_{0} L_t {e_t}'\Phi_{0} E_{\new}\nn\\
=&& {E_{\new}}' \Phi_{0} \left(P_* a_{t,*}+P_{\new} a_{t,\new}\right) [ \left(\Phi_{k-1} P_{*}\right) a_{t,*} + D_{\new,k-1} a_{t,\new}]'{I_{T_t}} [{\left(\Phi_{k-1}\right)_{T_t}}'\left(\Phi_{k-1}\right)_{T_t}]^{-1} {I_{T_t}}'\Phi_{0} E_{\new}\nn\\
=&&{E_{\new}}' \Phi_{0} \left(P_* a_{t,*}a_{t,*}'P_*'\Phi_{k-1}+P_{\new} a_{t,\new}a_{t,*}'P_*'\Phi_{k-1} + P_* a_{t,*}a_{t,\new}'D_{\new,k-1}' + P_{\new} a_{t,\new}a_{t,\new}'D_{\new,k-1}'\right) {I_{T_t}} \nn\\
&&[{\left(\Phi_{k-1}\right)_{T_t}}'\left(\Phi_{k-1}\right)_{T_t}]^{-1} {I_{T_t}}'\Phi_{0} E_{\new}
\eea
which is of size $c \times c$. Then $T2 = \frac{1}{\alpha} \sum_t \left(Z_t + Z_t'\right)$.\\
(a) By (\ref{err_bound}) and similar procedure to get (\ref{a_t_star_upper_bound}), we have
$$\|{E_{\new}}' \Phi_{0}P_*\|\leq \zeta_*^+, \|P_*'\Phi_{k-1}{I_{T_t}}[{\left(\Phi_{k-1}\right)_{T_t}}'\left(\Phi_{k-1}\right)_{T_t}]^{-1} {I_{T_t}}'\Phi_{0} E_{\new}\|\leq \zeta_*^+\phi^+\frac{\kappa_s^+}{\sqrt{1-\left(\zeta_*^+\right)^2}}$$
and
\bea
\Pb&&\Bigg(\left\|\frac{1}{\alpha}\sum_{t\in \Ic_{j,k}}{E_{\new}}' \Phi_{0}P_* a_{t,*}a_{t,*}' P_*'\Phi_{k-1}{I_{T_t}}[{\left(\Phi_{k-1}\right)_{T_t}}'\left(\Phi_{k-1}\right)_{T_t}]^{-1} {I_{T_t}}'\Phi_{0} E_{\new}\right\|\nn\\
&&\leq\lambda^+\phi^+(\zeta_*^+)^2\frac{\kappa_s^+}{\sqrt{1-\left(\zeta_*^+\right)^2}}+\frac{b^2(1-b^{2\alpha})}{1-b^2}r\gamma_*^2\phi^+(\zeta_*^+)^2\frac{\kappa_s^+}{\sqrt{1-\left(\zeta_*^+\right)^2}} +\frac{c\zeta\lambda^-}{144}\Big|X_{j,k-1}\Bigg)\nn\\
\geq&& 1- p_{d_1}(\alpha,\zeta)
\label{T_2_1}
\eea
where
\bea
&&p_{d_1}(\alpha,\zeta)=\nn\\&&c\digamma\left(\alpha,\frac{c\zeta\lambda^-}{576},\frac{(1-b^{2\alpha})(1-b)}{1+b}r\gamma_*^2\phi^+(\zeta_*^+)^2\frac{\kappa_s^+}{\sqrt{1-\left(\zeta_*^+\right)^2}}\right) + 2c\digamma\left(\alpha,\frac{c\zeta\lambda^-}{576}, 2\frac{b(1-b)(1-b^{2(\alpha-1)})}{(1+b)}r\gamma_*^2\phi^+(\zeta_*^+)^2\frac{\kappa_s^+}{\sqrt{1-\left(\zeta_*^+\right)^2}}\right) \nn\\&&+ 2c\digamma\left(\alpha,\frac{c\zeta\lambda^-}{576},2\frac{(1-b^{\alpha})^2}{1+b}r\gamma_*^2\phi^+(\zeta_*^+)^2\frac{\kappa_s^+}{\sqrt{1-\left(\zeta_*^+\right)^2}}\right) +2c\digamma\left(\alpha,\frac{c\zeta\lambda^-}{576},4\frac{b(1-b^{2\alpha})}{1+b}r\gamma_*^2\phi^+(\zeta_*^+)^2\frac{\kappa_s^+}{\sqrt{1-\left(\zeta_*^+\right)^2}}\right)\nn\\
\eea
(b) By (\ref{err_bound}) and similar procedure to get (\ref{a_t_new_star_upper_bound}), we have
$$\|{E_{\new}}' \Phi_{0}P_*\|\leq \zeta_*^+\leq 1, \|{E_{\new}}' \Phi_{0}P_{\new}\|\leq 1, \|P_*'\Phi_{k-1}{I_{T_t}}[{\left(\Phi_{k-1}\right)_{T_t}}'\left(\Phi_{k-1}\right)_{T_t}]^{-1} {I_{T_t}}'\Phi_{0} E_{\new}\|\leq \zeta_*^+\phi^+\frac{\kappa_s^+}{\sqrt{1-\left(\zeta_*^+\right)^2}}$$
and
\bea
\Pb&&\Bigg(\bigg\|\frac{1}{\alpha} \sum_{t\in \Ic_{j,k}} {E_{\new}}' \Phi_{0}(P_{\new} a_{t,\new}a_{t,*}'P_*'\Phi_{k-1} + P_* a_{t,*}a_{t,\new}'D_{\new,k-1}'){I_{T_t}}[{\left(\Phi_{k-1}\right)_{T_t}}'\left(\Phi_{k-1}\right)_{T_t}]^{-1} {I_{T_t}}'\Phi_{0} E_{\new}\bigg\|\nn\\&&\leq2\frac{b^2\left(1-b^{2\alpha}\right)}{1-b^2}\sqrt{cr}\gamma_*\gamma_{\new,k}\phi^+\zeta_*^+\frac{\kappa_s^+}{\sqrt{1-\left(\zeta_*^+\right)^2}} +\frac{c\zeta\lambda^-}{144}\Bigg)
\geq 1- p_{d_3}(\alpha,\zeta)
\label{T_2_2}
\eea
where
\bea
p_{d_3}(\alpha,\zeta)=&&2c\digamma(\alpha,\frac{c\zeta\lambda^-}{288}, 4\frac{1+b-2\sqrt{b^{\alpha+1}(1+b-b^{\alpha+1})}}{1+b}\sqrt{cr}\gamma_*\gamma_{\new,k}\phi^+\zeta_*^+\frac{\kappa_s^+}{\sqrt{1-\left(\zeta_*^+\right)^2}}) \nn\\&&+2c\digamma\left(\alpha,\frac{c\zeta\lambda^-}{288}, 8\frac{b(1-b^{2\alpha})}{1+b}\sqrt{cr}\gamma_*\gamma_{\new,k}\phi^+\zeta_*^+\frac{\kappa_s^+}{\sqrt{1-\left(\zeta_*^+\right)^2}}\right)
\eea

(c) By (\ref{err_bound}) and similar procedure to get (\ref{a_t_new_upper_bound}), we have
$$\|{E_{\new}}' \Phi_{0}P_{\new}\|\leq 1, \|D_{\new,k-1}'{I_{T_t}}[{\left(\Phi_{k-1}\right)_{T_t}}'\left(\Phi_{k-1}\right)_{T_t}]^{-1} {I_{T_t}}'\Phi_{0} E_{\new}\|\leq \zeta_{k-1}^+\phi^+\frac{(\kappa_s^+)^2}{\sqrt{1-\left(\zeta_*^+\right)^2}}$$
and
\bea
\Pb&&\Bigg(\left\|\frac{1}{\alpha}\sum_{t\in \Ic_{j,k}}{I_{T_t}} [{\left(\Phi_{k-1}\right)_{T_t}}'\left(\Phi_{k-1}\right)_{T_t}]^{-1} {I_{T_t}}' D_{\new,k-1} a_{t,\new} a_{t,\new}' D_{\new,k-1}'{I_{T_t}} [{\left(\Phi_{k-1}\right)_{T_t}}'\left(\Phi_{k-1}\right)_{T_t}]^{-1} {I_{T_t}}'\right\|\nn\\&&\leq\lambda_{\new,k}^+\zeta_{k-1}^+\phi^+\frac{(\kappa_s^+)^2}{\sqrt{1-\left(\zeta_*^+\right)^2}}+\frac{b^2(1-b^{2\alpha})}{1-b^2}c\gamma_{\new,k-1}^2\zeta_{k-1}^+\phi^+\frac{(\kappa_s^+)^2}{\sqrt{1-\left(\zeta_*^+\right)^2}} +\frac{c\zeta\lambda^-}{144}\Big|X_{j,k-1}\Bigg)\nn\\
\geq&&1-p_{d_3}(\alpha,\zeta)
\label{T_2_3}
\eea
where
\bea
p_{d_3}(\alpha,\zeta)=&&c\digamma\left(\alpha,\frac{c\zeta\lambda^-}{576},\frac{(1-b^{2\alpha})(1-b)}{1+b}c\gamma_{\new,k}^2(\phi^+\kappa_s^+\zeta_{k-1}^+)^2\frac{(\kappa_s^+)^2}{\sqrt{1-\left(\zeta_*^+\right)^2}} \right) \nn\\&&+ 2c\digamma\left(\alpha,\frac{c\zeta\lambda^-}{576}, 2\frac{b(1-b)(1-b^{2(\alpha-1)})}{1+b}c\gamma_{\new,k}^2(\phi^+\kappa_s^+\zeta_{k-1}^+)^2\frac{(\kappa_s^+)^2}{\sqrt{1-\left(\zeta_*^+\right)^2}}\right) \nn\\&&+ 2c\digamma\left(\alpha,\frac{c\zeta\lambda^-}{576},2\frac{(1-b^{\alpha})^2}{1+b}c\gamma_{\new,k}^2(\phi^+\kappa_s^+\zeta_{k-1}^+)^2\frac{(\kappa_s^+)^2}{\sqrt{1-\left(\zeta_*^+\right)^2}}\right) \nn\\&&+2c\digamma\left(\alpha,\frac{c\zeta\lambda^-}{576},4\frac{b(1-b^{2\alpha})}{1+b}c\gamma_{\new,k}\gamma_{\new,k-1}(\phi^+\kappa_s^+\zeta_{k-1}^+)^2\frac{(\kappa_s^+)^2}{\sqrt{1-\left(\zeta_*^+\right)^2}}\right)
\eea

Thus, combining (\ref{T_2_1}), (\ref{T_2_2}) and (\ref{T_2_3}), we have
\bea
\Pb&&\Bigg(\|T2\| = \left\|\frac{1}{\alpha} \sum_t \left(Z_t + Z_t'\right)\right\|\leq2\phi^+\kappa_s^+\frac{\left(\zeta_*^+\right)^2}{\sqrt{1-\left(\zeta_*^+\right)^2}}\lambda^++2\phi^+\zeta_{k-1}^+\frac{\left(\kappa_s^+\right)^2}{\sqrt{1-\left(\zeta_*^+\right)^2}}\lambda_{\new,k}^+ + \nn\\ &&
2\frac{b^2\left(1-b^{2\alpha}\right)}{1-b^2}r\gamma_*^2\phi^+\kappa_s^+\frac{\left(\zeta_*^+\right)^2}{\sqrt{1-\left(\zeta_*^+\right)^2}} + 2\frac{b^2\left(1-b^{2\alpha}\right)}{1-b^2}c\gamma_{\new,k-1}^2\phi^+\zeta_{k-1}^+\frac{\left(\kappa_s^+\right)^2}{\sqrt{1-\left(\zeta_*^+\right)^2}}+\nn\\&&
4\frac{b^2\left(1-b^{2\alpha}\right)}{1-b^2}\sqrt{cr}\gamma_*\gamma_{\new,k}\zeta_*^+\phi^+\frac{\kappa_s^+}{\sqrt{1-\left(\zeta_*^+\right)^2}}  +\frac{c\zeta\lambda^-}{24}\Big|X_{j,k-1}\Bigg)\nn\\
\geq&&1-p_{d}(\alpha,\zeta)
\label{T_2_bound}
\eea
where $p_d(\alpha,\zeta)=2p_{d_1}(\alpha,\zeta)+2p_{d_2}(\alpha,\zeta)+2p_{d_3}(\alpha,\zeta)$.

\subsubsection{$\|T_4\|$}
Consider $T4$. Let
\bea
Z_t: = {E_{\new,\perp}}'\Phi_{0} L_t {e_t}' \Phi_{0} E_{\new,\perp}
\eea
which is of size $\left(n-c\right)\times \left(n-c\right)$. Then $T4 = \frac{1}{\alpha} \sum_t \left(Z_t+Z_t'\right)$.
${E_{\new,\perp}}'\Phi_{0} L_t = {E_{\new,\perp}}'\Phi_{0}P_*a_{t,*} + {E_{\new,\perp}}'E_{\new}R_{\new}a_{t,\new} ={E_{\new,\perp}}' D_* a_{t,*}$.
Thus,
\bea
Z_t=&& {E_{\new,\perp}}' D_* a_{t,*} [ \left(\Phi_{k-1} P_{*}\right) a_{t,*} + D_{\new,k-1} a_{t,\new}]'{I_{T_t}} [{\left(\Phi_{k-1}\right)_{T_t}}'\left(\Phi_{k-1}\right)_{T_t}]^{-1} {I_{T_t}}' \Phi_{0} E_{\new,\perp}\nn\\
=&& {E_{\new,\perp}}' D_* a_{t, *} \left(a_{t,*}'P_*'\Phi_{k-1} + a_{t,\new}'D_{\new,k-1}'\right) {I_{T_t}} [{\left(\Phi_{k-1}\right)_{T_t}}'\left(\Phi_{k-1}\right)_{T_t}]^{-1} {I_{T_t}}' \Phi_{0} E_{\new,\perp}\nn
\eea
(a) Using (\ref{err_bound}) and similar procedure to get (\ref{a_t_star_upper_bound}), we have
\bea
\Pb&&\left(\|\frac{1}{\alpha}\sum_{t\in t_{j,k}}{E_{\new,\perp}}' D_* a_{t, *} a_{t,*}'P_*'\Phi_{k-1} {I_{T_t}} [{\left(\Phi_{k-1}\right)_{T_t}}'\left(\Phi_{k-1}\right)_{T_t}]^{-1} {I_{T_t}}' \Phi_{0} E_{\new,\perp}\|_2 \leq\phi^+{\left(\zeta_*^+\right)^2}\lambda^++ \right.\nn\\ && \left. \frac{b^2\left(1-b^{2\alpha}\right)}{1-b^2}r\gamma_{*}^2\phi^+{\left(\zeta_*^+\right)^2}
+ \frac{c\zeta\lambda^-}{96}\Big|X_{j,k-1}\right) \geq 1 - p_{e_1}\left(\alpha,\zeta\right)
\label{T_4_1}
\eea
where
\bea
p_{e_1}&&\left(\alpha,\zeta\right)= 2(n-c)\digamma\left(\alpha, \frac{c\zeta\lambda^-}{384}, \frac{\left(1-b^{2\alpha}\right)\left(1-b\right)}{1+b}r\gamma_*^2\phi^+{\left(\zeta_*^+\right)^2}\right) +2(n-c)\digamma\Bigg(\alpha,\frac{c\zeta\lambda^-}{384}, 2\frac{b(1-b)\left(1-b^{2\left(\alpha-1\right)}\right)}{1+b}r\gamma_*^2\phi^+\nn\\
&&\left(\zeta_*^+\right)^2\Bigg) + 2(n-c)\digamma\left(\alpha, \frac{c\zeta\lambda^-}{384}, 2\frac{\left(1-b^{\alpha}\right)^2}{1+b}r\gamma_*^2\phi^+{\left(\zeta_*^+\right)^2}\right)  + 2(n-c)\digamma\left(\alpha, \frac{c\zeta\lambda^-}{384}, 2\frac{b\left(1-b^{2\alpha}\right)}{1+b}r\gamma_*^2\phi^+{\left(\zeta_*^+\right)^2}\right)\nn
\eea

(b) Using (\ref{err_bound}) and similar procedure to get (\ref{a_t_new_star_upper_bound}), we have, conditioned on $X_{j,k-1}$,
\bea
\Pb&&\left(\|\frac{1}{\alpha}\sum_{t\in t_{j,k}}{E_{\new,\perp}}' D_* a_{t, *} a_{t,\new}'D_{\new,k-1}' {I_{T_t}} [{\left(\Phi_{k-1}\right)_{T_t}}'\left(\Phi_{k-1}\right)_{T_t}]^{-1} {I_{T_t}}' \Phi_{0} E_{\new,\perp}\|_2 \leq \right.\nn\\ &&\left.
\frac{b^2\left(1-b^{2\alpha}\right)}{1-b^2}\sqrt{cr}\gamma_*\gamma_{\new,k}\zeta_*^+\zeta_{k-1}^+\phi^+{\kappa_s^+}+ \frac{c\zeta\lambda^-}{96}\Big|X_{j,k-1}\right) \geq 1 - p_{e_2}\left(\alpha,\zeta\right)
\eea
where
\bea
p_{e_2}&&\left(\alpha,\zeta\right)=2(n-c)\digamma\left(\alpha, \frac{c\zeta\lambda^-}{192}, 2\frac{1+b-2\sqrt{b^{\alpha+1}\left(1+b-b^{\alpha+1}\right)}}{1+b}\sqrt{cr}\gamma_*\gamma_{\new,k}\zeta_*^+\zeta_{k-1}^+\phi^+{\kappa_s^+}\right) \nn\\&&+ 2(n-c)\digamma\left(\alpha, \frac{c\zeta\lambda^-}{192}, 4\frac{b\left(1-b^{2\alpha}\right)}{1+b}\sqrt{cr}\gamma_*\gamma_{\new,k}\zeta_*^+\zeta_{k-1}^+\phi^+{\kappa_s^+}\right)\nn
\eea

Thus, combining last two inequalities, we have
\bea
\Pb&&\left(\|T_4\|_2 = \|\frac{1}{\alpha}\sum_{t\in t_{j,k}}Z_t + Z_t'\|_2 \leq
2\phi^+{\left(\zeta_*^+\right)^2}\lambda^+ + 2\frac{b^2\left(1-b^{2\alpha}\right)}{1-b^2}r\gamma_{*}^2\phi^+{\left(\zeta_*^+\right)^2}+ 2\frac{b^2\left(1-b^{2\alpha}\right)}{1-b^2}\sqrt{cr}\gamma_*\gamma_{\new,k}\zeta_*^+\zeta_{k-1}^+\phi^+{\kappa_s^+}+\right.\nn\\ &&\left.\frac{c\zeta\lambda^-}{24}\Big|X_{j,k-1}\right)\geq 1 - p_{e}\left(\alpha,\zeta\right)
\label{T_4_bound}
\eea
where
\bea
p_e\left(\alpha,\zeta\right)=2p_{e_1}\left(\alpha,\zeta\right)+ 2p_{e_2}\left(\alpha,\zeta\right)
\eea

By condition in Theorem 18, $\zeta_*^+ = r_0\zeta+(j-1)c\zeta \leq r\zeta$, $\kappa_s^+=0.15, \zeta_*^+rf<1.5\times10^{-4}$, we have
\bea
r\gamma_*^2(\zeta_*^+)^2\leq (\zeta_*^+)^2\lambda^+ \eta/c= (\zeta_*^+)^2f\lambda_{\new,k}^-\eta/c = \zeta_*^+fr\zeta\lambda_{\new,k}^-\eta/c\leq 1.5\times10^{-4}\zeta\lambda_{\new,k}^-\eta \nn\\< 0.15c\zeta \frac{\kappa_s^+}{\sqrt{1-(\zeta_*^+)^2}}\eta\lambda_{\new,k}^- \leq \max\{0.15c\zeta,\zeta_{k-1}^+\}\frac{\kappa_s^+}{\sqrt{1-(\zeta_*^+)^2}}\eta\lambda_{\new,k}^-
\eea

Thus,
\bea
\Pb&&\Bigg(\max\{\|T_2\|_2, \|T_4\|_2\} \leq
2\phi^+\kappa_s^+\frac{(\zeta_*^+)^2}{\sqrt{1-(\zeta_*^+)^2}}\left(\lambda^++\frac{b^2(1-b^{2\alpha})}{1-b^2}r\gamma_{*}^2\right) +2\phi^+\max\{0.15c\zeta,\zeta_{k-1}^+\}\frac{(\kappa_s^+)^2}{\sqrt{1-(\zeta_*^+)^2}}\Big(\lambda_{\new,k}^+ \nn\\ &&+\frac{b^2(1-b^{2\alpha})}{1-b^2}\eta\lambda^+_{\new,k}\Big) +  4\frac{b^2(1-b^{2\alpha})}{1-b^2}r\gamma_*^2\zeta_*^+\phi^+\frac{\kappa_s^+}{\sqrt{1-(\zeta_*^+)^2}}+\frac{c\zeta\lambda^-}{24}\bigg|X_{j,k-1}\Bigg)\geq 1 - \max\{p_{d}(\alpha,\zeta), p_{e}(\alpha,\zeta)\}
\label{T_2_4_bound}
\eea

\subsubsection{$\|B_k\|$}
Consider $\|B_k\|_2$. Let $Z_t := {E_{\new,\perp}}'\Phi_{0} (L_t-e_t)({L_t}'-{e_t}')\Phi_{0} E_{\new}$ which is of size $(n-c)\times c$. Then $B_k = \frac{1}{\alpha} \sum_t Z_t$.
As ${E_{\new,\perp}}'\Phi_{0} P_{\new}a_{t,\new}={E_{\new,\perp}}'E_{\new} R_{\new} a_{t,\new}=0$, ${E_{\new,\perp}}'\Phi_{0} (L_t-e_t) = {E_{\new,\perp}}'( D_{*} a_{t,*} - \Phi_{0} e_t)$,  ${E_{\new}}' \Phi_{0} (L_t - e_t) = R_{\new} a_{t,\new}+ {E_{\new}}' D_* a_{t,*} - (R_\new')^{-1} D_\new' e_t$.
Thus,
\bea
Z_t&&= {E_{\new,\perp}}'( D_{*} a_{t,*} - \Phi_{0} e_t)(a_{t,\new}'R_{\new}'+ a_{t,*}'D_*'{E_{\new}} - e_t 'D_\new (R_\new)^{-1})\nn\\
=&&{E_{\new,\perp}}'( D_{*} a_{t,*} - \Phi_{0} {I_{T_t}}[{(\Phi_{k-1})_{T_t}}'(\Phi_{k-1})_{T_t}]^{-1}I_{T_t}'[ (\Phi_{k-1} P_{*}) a_{t,*} + D_{\new,k-1} a_{t,\new}])(a_{t,\new}'R_{\new}'+ a_{t,*}'D_*'{E_{\new}} - [ (\Phi_{k-1} P_{*}) a_{t,*}\nn\\
&& + D_{\new,k-1} a_{t,\new}]'{I_{T_t}} [{(\Phi_{k-1})_{T_t}}'(\Phi_{k-1})_{T_t}]^{-1} {I_{T_t}}'D_\new (R_\new)^{-1})\nn\\
=&& {E_{\new,\perp}}'( D_{*}  - \Phi_{0} {I_{T_t}}[{(\Phi_{k-1})_{T_t}}'(\Phi_{k-1})_{T_t}]^{-1}I_{T_t}' (\Phi_{k-1} P_{*}) ) a_{t,*}a_{t,*}'(D_*'{E_{\new}} - P_*' \Phi_{k-1}{I_{T_t}} [{(\Phi_{k-1})_{T_t}}'(\Phi_{k-1})_{T_t}]^{-1} {I_{T_t}}'D_\new (R_\new)^{-1})\nn\\
&&+ {E_{\new,\perp}}'( D_{*}  - \Phi_{0} {I_{T_t}}[{(\Phi_{k-1})_{T_t}}'(\Phi_{k-1})_{T_t}]^{-1}I_{T_t}' (\Phi_{k-1} P_{*}) ) a_{t,*}a_{t,\new}(R_{\new}'-D_{\new,k-1}'{I_{T_t}} [{(\Phi_{k-1})_{T_t}}'(\Phi_{k-1})_{T_t}]^{-1} {I_{T_t}}'D_\new (R_\new)^{-1})\nn\\
&&+ {E_{\new,\perp}}'(- \Phi_{0} {I_{T_t}}[{(\Phi_{k-1})_{T_t}}'(\Phi_{k-1})_{T_t}]^{-1}I_{T_t}' D_{\new,k-1} ) a_{t,\new}a_{t,*}'(D_*'{E_{\new}} - P_*' \Phi_{k-1}{I_{T_t}} [{(\Phi_{k-1})_{T_t}}'(\Phi_{k-1})_{T_t}]^{-1} {I_{T_t}}'D_\new (R_\new)^{-1})\nn\\
&&+ {E_{\new,\perp}}'(- \Phi_{0} {I_{T_t}}[{(\Phi_{k-1})_{T_t}}'(\Phi_{k-1})_{T_t}]^{-1}I_{T_t}' D_{\new,k-1} ) a_{t,\new}a_{t,\new}(R_{\new}'-D_{\new,k-1}'{I_{T_t}} [{(\Phi_{k-1})_{T_t}}'(\Phi_{k-1})_{T_t}]^{-1} {I_{T_t}}'D_\new (R_\new)^{-1})
\eea

(a) Using (\ref{err_bound}) and similar procedure to get (\ref{a_t_star_upper_bound}), we have
\bea
\Pb&&\Bigg(\bigg\|\frac{1}{\alpha}\sum_{t\in t_{j,k}}{E_{\new,\perp}}'( D_{*}  - \Phi_{0} {I_{T_t}}[{(\Phi_{k-1})_{T_t}}'(\Phi_{k-1})_{T_t}]^{-1}I_{T_t}' (\Phi_{k-1} P_{*}) ) a_{t,*}a_{t,*}'(D_*'{E_{\new}} - P_*' \Phi_{k-1}{I_{T_t}} [{(\Phi_{k-1})_{T_t}}'(\Phi_{k-1})_{T_t}]^{-1} {I_{T_t}}'\nn\\ &&D_\new (R_\new)^{-1})\bigg\|_2 \leq(\zeta_*^++\phi^+\zeta_*^+)(\zeta_*^++\zeta_*^+\phi^+\frac{\kappa_s^+}{\sqrt{1-(\zeta_*^+)^2}})(\lambda^++ \frac{b^2(1-b^{2\alpha})}{1-b^2}r\gamma_{*}^2) +
\frac{c\zeta\lambda^-}{96}\bigg|X_{j,k-1}\Bigg) \geq 1 - p_{f_1}(\alpha,\zeta)\nn\\
\label{B_k_1}
\eea
where
\bea
&&p_{f_1}(\alpha,\zeta)= n\digamma\left(\alpha, \frac{c\zeta\lambda^-}{384}, \frac{(1-b^{2\alpha})(1-b)}{1+b}r\gamma_*^2 (\zeta_*^+)^2\bigg(1+\phi^+)\bigg(1+\phi^+\frac{\kappa_s^+}{\sqrt{1-(\zeta_*^+)^2}}\bigg)\right) +n\digamma\Bigg(\alpha,\frac{c\zeta\lambda^-}{384}, \nn\\&&2\frac{b(1-b)(1-b^{2(\alpha-1)})}{1+b}r
\gamma_*^2(\zeta_*^+)^2(1+\phi^+)\bigg(1+\phi^+\frac{\kappa_s^+}{\sqrt{1-(\zeta_*^+)^2}}\bigg)\Bigg) + n\digamma\Bigg(\alpha, \frac{c\zeta\lambda^-}{384}, 2\frac{(1-b^{\alpha})^2}{1+b}r\gamma_*^2(\zeta_*^+)^2\nn\\
&&(1+\phi^+)(1+\phi^+\frac{\kappa_s^+}{\sqrt{1-(\zeta_*^+)^2}})\Bigg)+n\digamma\Bigg(\alpha, \frac{c\zeta\lambda^-}{384}, 2\frac{b(1-b^{2\alpha})}{1+b}r\gamma_*^2(\zeta_*^+)^2(1+\phi^+)\bigg(1+\phi^+\frac{\kappa_s^+}{\sqrt{1-(\zeta_*^+)^2}}\bigg)\Bigg)
\eea

(b) Using (\ref{err_bound}) and similar procedure to get (\ref{a_t_new_star_upper_bound}), we have
\bea
\Pb&&\Bigg(\bigg\| {E_{\new,\perp}}'( D_{*}  - \Phi_{0} {I_{T_t}}[{(\Phi_{k-1})_{T_t}}'(\Phi_{k-1})_{T_t}]^{-1}I_{T_t}' (\Phi_{k-1} P_{*}) ) a_{t,*}a_{t,\new}(R_{\new}'-D_{\new,k-1}'{I_{T_t}} [{(\Phi_{k-1})_{T_t}}'(\Phi_{k-1})_{T_t}]^{-1} {I_{T_t}}'D_\new \nn\\ &&(R_\new)^{-1})\bigg\|_2 \leq
\frac{b^2(1-b^{2\alpha})}{1-b^2}\sqrt{cr}\gamma_*\gamma_{\new,k}(\zeta_*^++\phi^+\zeta_*^+)\bigg(1+(\kappa_s^+)^2\zeta_{k-1}^+\frac{\phi^+}{\sqrt{1-(\zeta_*^+)^2}}\bigg) + \frac{c\zeta\lambda^-}{96}\bigg|X_{j,k-1}\Bigg) \geq 1 - p_{f_2}(\alpha,\zeta)\nn\\
\label{B_k_2}
\eea
where
\bea
p_{f_2}&&(\alpha,\zeta)=n\digamma\left(\alpha, \frac{c\zeta\lambda^-}{192}, 2\frac{1+b-2\sqrt{b^{\alpha+1}(1+b-b^{\alpha+1})}}{1+b}\sqrt{cr}\gamma_*\gamma_{\new,k}(\zeta_*^++\phi^+\zeta_*^+)(1+(\kappa_s^+)^2\zeta_{k-1}^+\frac{\phi^+}{\sqrt{1-(\zeta_*^+)^2}})\right) + \nn\\
&&n\digamma\left(\alpha, \frac{c\zeta\lambda^-}{192}, 4\frac{b(1-b^{2\alpha})}{1+b}\sqrt{cr}\gamma_*\gamma_{\new,k}(\zeta_*^++\phi^+\zeta_*^+)\bigg(1+(\kappa_s^+)^2\zeta_{k-1}^+\frac{\phi^+}{\sqrt{1-(\zeta_*^+)^2}}\bigg)\right)\nn
\eea
and
\bea
\Pb&&\Bigg(\| {E_{\new,\perp}}'(- \Phi_{0} {I_{T_t}}[{(\Phi_{k-1})_{T_t}}'(\Phi_{k-1})_{T_t}]^{-1}I_{T_t}' D_{\new,k-1} ) a_{t,\new}a_{t,*}'(D_*'{E_{\new}} - P_*' \Phi_{k-1}{I_{T_t}} [{(\Phi_{k-1})_{T_t}}'(\Phi_{k-1})_{T_t}]^{-1} {I_{T_t}}'D_\new \nn\\ &&(R_\new)^{-1})\|_2 \leq
\frac{b^2(1-b^{2\alpha})}{1-b^2}\sqrt{cr}\gamma_*\gamma_{\new,k}\phi^+\kappa_s^+\zeta_{k-1}^+\bigg(\zeta_*^++\zeta_*^+\phi^+\frac{\kappa_s^+}{\sqrt{1-(\zeta_*^+)^2}}\bigg) + \frac{c\zeta\lambda^-}{96}\bigg|X_{j,k-1}\Bigg) \geq 1 - p_{f_3}(\alpha,\zeta)
\label{B_k_3}
\eea
where
\bea
p_{f_3}&&(\alpha,\zeta)=n\digamma\left(\alpha, \frac{c\zeta\lambda^-}{192}, 2\frac{1+b-2\sqrt{b^{\alpha+1}(1+b-b^{\alpha+1})}}{1+b}\sqrt{cr}\gamma_*\gamma_{\new,k}\phi^+\kappa_s^+\zeta_{k-1}^+\bigg(\zeta_*^++\zeta_*^+\phi^+\frac{\kappa_s^+}{\sqrt{1-(\zeta_*^+)^2}}\bigg)\right) + \nn\\
&&n\digamma\left(\alpha, \frac{c\zeta\lambda^-}{192}, 4\frac{b(1-b^{2\alpha})}{1+b}\sqrt{cr}\gamma_*\gamma_{\new,k}\phi^+\kappa_s^+\zeta_{k-1}^+\bigg(\zeta_*^++\zeta_*^+\phi^+\frac{\kappa_s^+}{\sqrt{1-(\zeta_*^+)^2}}\bigg)\right)\nn
\eea

(c) Using (\ref{err_bound}) and similar procedure to get (\ref{a_t_new_upper_bound}), we have
\bea
\Pb&&\Bigg(\bigg\|\frac{1}{\alpha}\sum_{t\in t_{j,k}}{E_{\new,\perp}}'(- \Phi_{0} {I_{T_t}}[{(\Phi_{k-1})_{T_t}}'(\Phi_{k-1})_{T_t}]^{-1}I_{T_t}' D_{\new,k-1} ) a_{t,\new}a_{t,\new}(R_{\new}'-D_{\new,k-1}'{I_{T_t}} [{(\Phi_{k-1})_{T_t}}'(\Phi_{k-1})_{T_t}]^{-1} {I_{T_t}}'D_\new \nn\\
&&(R_\new)^{-1})\bigg\|_2 \leq\zeta_{k-1}^+\phi^+\kappa_s^+\bigg(1+\phi^+\zeta_{k-1}^+\frac{(\kappa_s^+)^2}{\sqrt{1-(\zeta_*^+)^2}}\bigg)\bigg(\lambda_{\new,k}^+ + \frac{b^2(1-b^{2\alpha})}{1-b^2}c\gamma_{\new,k-1}^2\bigg) + \frac{c\zeta\lambda^-}{96}\bigg|X_{j,k-1}\Bigg)\geq 1 - p_{f_4}(\alpha,\zeta)\nn\\
\label{B_k_4}
\eea
where
\bea
p_{f_4}&&(\alpha,\zeta)=n\digamma\left(\alpha, \frac{c\zeta\lambda^-}{384}, \frac{(1-b^{2\alpha})(1-b)}{1+b}c\gamma_{\new,k}^2\zeta_{k-1}^+\phi^+\kappa_s^+\bigg(1+\phi^+\zeta_{k-1}^+\frac{(\kappa_s^+)^2}{\sqrt{1-(\zeta_*^+)^2}}\bigg)\right) \nn\\
&&+ n\digamma\Bigg(\alpha, \frac{c\zeta\lambda^-}{384}, 2\frac{b(1-b)(1-b^{2(\alpha-1)})}{1+b}c\gamma_{\new,k}^2(\zeta_{k-1}^+\phi^+)^2\frac{(\kappa_s^+)^3}{\sqrt{1-(\zeta_*^+)^2}}\Bigg) \nn\\
&&+n\digamma\left(\alpha, \frac{c\zeta\lambda^-}{384}, 2\frac{(1-b^{\alpha})^2}{1+b}c\gamma_{\new,k}^2\zeta_{k-1}^+\phi^+\kappa_s^+\bigg(1+\phi^+\zeta_{k-1}^+\frac{(\kappa_s^+)^2}{\sqrt{1-(\zeta_*^+)^2}}\bigg)\right) \nn\\&&+n\digamma\Bigg(\alpha, \frac{c\zeta\lambda^-}{384},4\frac{b(1-b^{2\alpha})}{1+b}c\gamma_{\new,k}\gamma_{\new,k-1}\zeta_{k-1}^+\phi^+\kappa_s^+\bigg(1+\phi^+\zeta_{k-1}^+\frac{(\kappa_s^+)^2} {\sqrt{1-(\zeta_*^+)^2}}\bigg)\Bigg)
\eea

Using (\ref{B_k_1}), (\ref{B_k_2}), (\ref{B_k_3}) and (\ref{B_k_4}) and the union bound,  for any $X_{j,k-1} \in  \Gamma_{j,k-1}$,
\bea
\mathbf{P} &&\Bigg(\|B_k\|_2 \leq (\zeta_*^++\phi^+\zeta_*^+)\bigg(\zeta_*^++\zeta_*^+\phi^+\frac{\kappa_s^+}{\sqrt{1-(\zeta_*^+)^2}}\bigg) \bigg(\lambda^++\frac{b^2(1-b^{2\alpha})}{1-b^2}r\gamma_{*}^2\bigg)+\zeta_{k-1}^+\phi^+\kappa_s^+\bigg(1+\phi^+\zeta_{k-1}^+\frac{(\kappa_s^+)^2}{\sqrt{1-(\zeta_*^+)^2}}\bigg)\nn\\
&&\bigg(\lambda_{\new,k}^+ + \frac{b^2(1-b^{2\alpha})}{1-b^2}c\gamma_{\new,k-1}^2\bigg)+ \frac{b^2(1-b^{2\alpha})}{1-b^2}\sqrt{cr}\gamma_*\gamma_{\new,k}\bigg((\zeta_*^++\phi^+\zeta_*^+)\Big(1+(\kappa_s^+)^2\zeta_{k-1}^+\frac{\phi^+}{\sqrt{1-(\zeta_*^+)^2}}\Big) \nn\\&& + \phi^+\kappa_s^+\zeta_{k-1}^+\Big(\zeta_*^++\zeta_*^+\phi^+\frac{\kappa_s^+}{\sqrt{1-(\zeta_*^+)^2}}\Big)\bigg)+
\frac{c\zeta\lambda^-}{24}\bigg|X_{j,k-1}\Bigg)\geq 1-p_f(\alpha,\zeta)
\label{B_k_bound}
\eea
where $p_f(\alpha,\zeta)=p_{f_1}(\alpha,\zeta)+p_{f_2}(\alpha,\zeta)+p_{f_3}(\alpha,\zeta)+p_{f_4}(\alpha,\zeta)$.

Using (\ref{e_t_bound}), (\ref{T_2_4_bound}) and (\ref{B_k_bound}) and the union bound,  for any $X_{j,k-1} \in  \Gamma_{j,k-1}$,
\bea
\mathbf{P} \left(\|\mathcal{H}_k\|_2 \leq b_{\Hc} + \frac{c\zeta\lambda^-}{8}|X_{j,k-1}\right) \geq 1- p_c(\alpha,\zeta) - p_f(\alpha,\zeta)- \max\{p_d(\alpha,\zeta), p_e(\alpha,\zeta)\}
\label{H_k_bound}
\eea
where

\bea
b_{\Hc} :=&&
(\phi^+)^2(\zeta_*^+)^2\lambda^+ +(\phi^+)^2(\kappa_s^+\zeta_{k-1}^+)^2\lambda_{\new,k}^++2\frac{b^2(1-b^{2\alpha})}{1-b^2}\sqrt{cr}\gamma_*\gamma_{\new,k}\zeta_*^+\zeta_{k-1}^+\kappa_s^+(\phi^+)^2 + \frac{b^2(1-b^{2\alpha})}{1-b^2}r\gamma_*^2(\phi^+\zeta_*^+)^2 + \nn\\&&\frac{b^2(1-b^{2\alpha})}{1-b^2}c\gamma_{\new,k-1}^2(\phi^+\kappa_s^+\zeta_{k-1}^+)^2 + 2\phi^+\kappa_s^+\frac{(\zeta_*^+)^2}{\sqrt{1-(\zeta_*^+)^2}}\left(\lambda^++\frac{b^2(1-b^{2\alpha})}{1-b^2}r\gamma_{*}^2\right)+ \nn\\ && 2\phi^+\max\{0.15c\zeta,\zeta_{k-1}^+\}\frac{(\kappa_s^+)^2}{\sqrt{1-(\zeta_*^+)^2}}\left(\lambda_{\new,k}^+ +\frac{b^2(1-b^{2\alpha})}{1-b^2}\eta\lambda^+_{\new,k}\right) +  4\frac{b^2(1-b^{2\alpha})}{1-b^2}r\gamma_*^2\zeta_*^+\phi^+\frac{\kappa_s^+}{\sqrt{1-(\zeta_*^+)^2}}+ \nn\\ && (\zeta_*^++\phi^+\zeta_*^+)\left(\zeta_*^++\zeta_*^+\phi^+\frac{\kappa_s^+}{\sqrt{1-(\zeta_*^+)^2}}\right) \left(\lambda^++\frac{b^2(1-b^{2\alpha})}{1-b^2}r\gamma_{*}^2\right)+\zeta_{k-1}^+\phi^+\kappa_s^+\left(1+\phi^+\zeta_{k-1}^+\frac{(\kappa_s^+)^2}{\sqrt{1-(\zeta_*^+)^2}}\right)\nn\\
&&\left(\lambda_{\new,k}^+ + \frac{b^2(1-b^{2\alpha})}{1-b^2}c\gamma_{\new,k-1}^2\right)+ \frac{b^2(1-b^{2\alpha})}{1-b^2}\sqrt{cr}\gamma_*\gamma_{\new,k}\Bigg((\zeta_*^++\phi^+\zeta_*^+)\bigg(1+(\kappa_s^+)^2\zeta_{k-1}^+\frac{\phi^+}{\sqrt{1-(\zeta_*^+)^2}}\bigg) \nn\\&& + \phi^+\kappa_s^+\zeta_{k-1}^+\Big(\zeta_*^++\zeta_*^+\phi^+\frac{\kappa_s^+}{\sqrt{1-(\zeta_*^+)^2}}\Big)\Bigg)
\eea

\begin{remark}
As shown in Remark \ref{a_t_new_k_1}, when $k=1$, there is no $a_{t_j+(k-1)\alpha-1,\new}$, leading to changes in bounds $b_{A_k}, b_{A_{k,\perp}}, b_{\Hc}$, in which case $\zeta_{k}^+$ decreases exponentially with same parameters given above with larger $b$.
\end{remark}

\end{proof}

\appendix

\subsection{Proof of Lemma \ref{lemma0}}
\begin{proof} 
Because $P$, $Q$ and $\Phat$ are basis matrix, $P'P=I$, ${Q}'Q = I$ and $\Phat'\Phat=I$.
\ben
\item Using $P'P = I$ and $\|M\|_2^2 = \|MM'\|_2$, $\|(I-\Phat{\Phat}')P P'\|_2 =\|(I-\Phat{\Phat}')P\|_2$. Similarly, $\|(I-PP')\Phat {\Phat}'\|_2=\|(I-PP')\Phat\|_2$. Let $D_1 = (I-\Phat{\Phat}')P P'$ and let $D_2=(I-PP')\Phat {\Phat}'$. Notice that $\|D_1\|_2 = \sqrt{\lambda_{\max}(D_1'D_1)} = \sqrt{\|D_1'D_1\|_2}$ and $\|D_2\|_2 = \sqrt{\lambda_{\max}(D_2'D_2)} = \sqrt{\|D_2'D_2\|_2}$. So, in order to show $\|D_1\|_2 = \|D_2\|_2$, it suffices to show that $\|D_1' D_1\|_2 = \|D_2'D_2\|_2$.
    Let $P'\Phat\overset{SVD}{=} U\Sigma V'$. Then, $D_1'D_1 = P (I - P'\Phat{\Phat}'P)P' = P U (I-\Sigma^2) U' P' $ and $D_2'D_2 = \Phat ( I - {\Phat}'PP'\Phat){\Phat}' = \Phat V (I-\Sigma^2) V'{\Phat}'$ are the compact SVD's of $D_1' D_1$ and $D_2' D_2$ respectively. Therefore, $\|D_1' D_1\| = \|D_2'D_2\|_2 = \|I-\Sigma^2\|_2$ and hence $\|(I-\Phat{\Phat}')PP'\|_2 =\|(I - P{P}')\Phat{\Phat}'\|_2$.

\item $\|P{P}' -\Phat {\Phat}'\|_2 = \| PP - \Phat{\Phat}'PP' + \Phat{\Phat}'PP'-\Phat {\Phat}'\|_2 \leq \| (I- \Phat{\Phat}')PP'\|_2 +
    \|(I-PP')\Phat {\Phat}'\|_2 = 2 \zeta_*$.

\item Since ${Q}'P = 0$, then $\|{Q}'\Phat\|_2 = \|{Q}'(I-P P')\Phat\|_2 \leq \|(I-P P')\Phat\|_2  = \zeta_*$.

\item 
Let $M = (I-\Phat {\Phat}') Q)$. Then $M'M = Q'(I-\Phat {\Phat}')Q$ and so $\sigma_i ((I-\Phat {\Phat}') Q) = \sqrt{\lambda_i (Q'(I-\Phat {\Phat}')Q)}$. Clearly, $\lambda_{\max} (Q'(I-\Phat {\Phat}')Q) \leq 1$. By Weyl's Theorem, $\lambda_{\min} (Q'(I-\Phat {\Phat}')Q) \geq 1 - \lambda_{\max} (Q'\Phat {\Phat}'Q) = 1- \|{Q}'\Phat\|_2^2 \geq 1-\zeta_*^2$. Therefore, $\sqrt{1-\zeta_*^2} \leq \sigma_{i}((I-\Phat {\Phat}') Q) \leq 1$.

\een
\end{proof}

\subsection{Proof of Lemma \ref{rem_prob}}
\begin{proof}
It is easy to see that $\mathbf{P}(\ecalb, \ecalc) = \E[\mathbb{I}_\calb(X,Y) \mathbb{I}_\calc(X)].$
If $\E[\mathbb{I}_{\calb}(X,Y)|X] \ge p$ for all $X \in \calc$, this means that $\E[\mathbb{I}_{\calb}(X,Y)|X] \mathbb{I}_{\calc}(X) \ge p  \mathbb{I}_{\calc}(X) $. This, in turn, implies that
$$\mathbf{P}(\ecalb, \ecalc) = \E[\mathbb{I}_\calb(X,Y) \mathbb{I}_\calc(X)] = \E[ \E[\mathbb{I}_{\calb}(X,Y)|X] \mathbb{I}_{\calc}(X) ] \ge p \E[\mathbb{I}_{\calc}(X) ].$$
Recall from Definition \ref{probdefs} that $\mathbf{P}(\ecalb|X) = \E[\mathbb{I}_{\calb}(X,Y)|X]$ and $\mathbf{P}(\ecalc)= \E[\mathbb{I}_{\calc}(X) ]$.
Thus, we conclude that if $\mathbf{P}(\ecalb|X) \ge p$ for all $X \in \calc$, then $\mathbf{P}(\ecalb, \ecalc) \ge p \mathbf{P}(\ecalc)$. Using the definition of $\mathbf{P}(\ecalb|\ecalc)$, the claim follows. 
\end{proof}

\subsection{Proof of Theorem \ref{azuma}}
\begin{proof}[Proof of Theorem]
The matrix Laplace transform method, Proposition 3.1, states that
\begin{equation} \label{eqn:azuma-lt}
\Pb\left( \lambda_{\max}\left( \sum_k \mathbf{X}_k \right) \geq t \right)
	\leq \inf_{\theta > 0} \left\{ e^{-\theta t} \cdot
	\E \trace \exp\left( \sum\nolimits_k \theta \mathbf{X}_k \right) \right\}.
\end{equation}
The main difficulty in the proof is to bound the matrix mgf, which we accomplish by an iterative argument that alternates between symmetrization and cumulant bounds.

Let us detail the first step of the iteration.  Denote conditional expectation $\E(XY|g(Y))$ as $\E_{X,Y|g(Y)}(XY)$, then
\begin{align*}
\E \trace \exp\left( \sum\nolimits_k \theta \mathbf{X}_k \right)
	&= \E_{\mathbf{X}_1,\cdots,\mathbf{X}_{n-1} } \E_{\mathbf{X_n}|\mathbf{X}_1,\cdots,\mathbf{X}_{n-1} }\left[ \trace \exp\left( \sum\nolimits_{k=1}^{n-1}
		\theta \mathbf{X}_k + \theta \mathbf{X}_n \right)\right] \\
    &\leq \E_{\mathbf{X}_1,\cdots,\mathbf{X}_{n-1} } \E_{\epsilon, \mathbf{X_n}|\mathbf{X}_1,\cdots,\mathbf{X}_{n-1} }\left[ \trace \exp\left(
	\sum\nolimits_{k=1}^{n-1} \theta\mathbf{X}_k
	+ 2\eps \theta \mathbf{X}_n \right) \right] \\
	&= \E_{\mathbf{X}_1,\cdots,\mathbf{X}_{n} } \E_{\epsilon|\mathbf{X}_1,\cdots,\mathbf{X}_{n} }\left[ \trace \exp\left(
	\sum\nolimits_{k=1}^{n-1} \theta\mathbf{X}_k
	+ 2\eps \theta \mathbf{X}_n \right) \right] \\
	&\leq \E_{\mathbf{X}_1,\cdots,\mathbf{X}_{n} } \trace \exp\left( \sum\nolimits_{k=1}^{n-1} \theta \mathbf{X}_k
	+ \log \E_{\epsilon|\mathbf{X}_1,\cdots,\mathbf{X}_{n} }\big[ e^{2\eps \theta \mathbf{X}_n}\big] \right) \\
	&\leq \E_{\mathbf{X}_1,\cdots,\mathbf{X}_{n} } \trace \exp\left( \sum\nolimits_{k=1}^{n-1} \theta \mathbf{X}_k
	+ 2\theta^2 \mathbf{A}_n^2 \right).
\end{align*}

The first identity is the tower property of conditional expectation.  In the second line, we invoke the symmetrization method, Lemma 7.6, conditional on $\mathbf{X}_1,\cdots,\mathbf{X}_{n-1}$, and then we relax the conditioning on the inner expectation to $\mathbf{X}_1,\cdots,\mathbf{X}_{n}$.  By construction, the Rademacher variable $\eps$ is independent from $\mathbf{X}_1,\cdots,\mathbf{X}_{n}$, so we can apply the concavity result, Corollary 3.3, conditional on $\mathbf{X}_1,\cdots,\mathbf{X}_{n}$.  Finally, we use the fact (2.5) that the trace exponential is monotone to introduce the Azuma cgf bound, Lemma 7.7, in the last inequality.

By iteration, we achieve
\begin{equation} \label{eqn:azuma-mgf-pf}
\E \trace \exp\left( \sum\nolimits_k \theta \mathbf{X}_k \right)
	\leq \trace \exp\left( 2\theta^2 \sum\nolimits_k \mathbf{A}_k^2 \right).
\end{equation}
Note that this procedure relies on the fact that the sequence $\{\mathbf{A}_k\}$ of upper bounds does not depend on the values of the random sequence $\{\mathbf{X}_k\}$.  Substitute the mgf bound~\eqref{eqn:azuma-mgf-pf} into the Laplace transform bound~\eqref{eqn:azuma-lt}, and observe that the infimum is achieved when $\theta = t/4\sigma^2$.
\end{proof}

\subsection{Proof of Corollary \ref{azuma_rec}}
\begin{proof}
\label{proof_azuma_rec}
Define the dilation of an $n_1 \times n_2$ matrix $M$ as $\text{dilation} (M) := \left[\begin{array}{cc}0 & {M}' \\ M & 0 \\\end{array} \right]$. Notice that this is an $(n_1+n_2) \times (n_1 +n_2)$ Hermitian matrix \cite{tropp2012user}. As shown in \cite[equation 2.12]{tropp2012user}, 
\bea
\lambda_{\max}(\text{dilation}(M)) = \|\text{dilation} (M)\|_2 = \|M\|_2
\label{dilM}
\eea
Thus, the corollary assumptions imply that $\mathbf{P}(\|\text{dilation} (Z_t)\|_2 \leq b_1 |X) = 1$ for all $X \in \calc$. 
%
Using (\ref{dilM}), the corollary assumptions also imply that $\E_{t-1}( \text{dilation}(Z_t) |X) = \text{dilation} ( \E_{t-1}(Z_t|X)) =0$ for all $X \in \calc$.
%
%
Thus, applying Corollary \ref{azuma} for the sequence $\{\text{dilation} (Z_t)\}$, we get that, 
$$\mathbf{P} \left(\lambda_{\max}\left(\frac{1}{\alpha}\sum_{t=1}^{\alpha} \text{dilation}(Z_t)\right)\leq \epsilon \big|X\right) \geq 1- (n_1+n_2) \exp\left(-\frac{\alpha \epsilon^2}{32 b_1^2}\right) \ \text{for all} \ X \in \calc$$
Using (\ref{dilM}), $\lambda_{\max}(\frac{1}{\alpha}\sum_{t=1}^{\alpha} \text{dilation}(Z_t)) = \lambda_{\max}(\text{dilation}(\frac{1}{\alpha}\sum_{t=1}^{\alpha} Z_t))  = \|\frac{1}{\alpha}\sum_{t=1}^{\alpha} Z_t\|_2$ and this gives the final result. 
\end{proof}

\subsection{Proof of Corollary \ref{azuma_nonzero}}

\begin{proof}
\ben
\item Since, for any $X \in {\cal C}$, conditioned on $X$, the $Z_t$'s are adapted, the same is also true for $Z_t - g(X, Z_{1:t-1})$ for any function of $X$ and $Z_{1:t-1}$.
Let $Y_t := Z_t - \E_{t-1}(Z_t|X)$. Thus, for any $X \in {\cal C}$, conditioned on $X$, the $Y_t$'s are adapted.
Also, clearly $\E_{t-1}(Y_t|X) = 0$. Since for all $X \in \calc$, $\mathbf{P}(b_1 I \preceq Z_t \preceq b_2 I|X)=1$ and since $\lambda_{\max}(.)$ is a convex function, and $\lambda_{\min}(.)$ is a concave function, of a Hermitian matrix, thus $b_1 I \preceq \E_{t-1}(Z_t|X) \preceq b_2 I$ w.p. one for all $X \in \calc$. Therefore, $\mathbf{P}(Y_t^2 \preceq (b_2 -b_1)^2 I|X) = 1$ for all $X \in \calc$.
Thus, for Theorem \ref{azuma}, $\sigma^2 = \|\sum_{t=1}^{\alpha} (b_2 - b_1)^2I\|_2 = \alpha (b_2-b_1)^2$. For any $X \in \calc$, applying Theorem \ref{azuma} for $\{Y_t\}$'s conditioned on $X$, we get that, for any $\eps > 0$,
    $$\mathbf{P}\left( \lambda_{\max} \left(\frac{1}{\alpha} \sum_{t=1}^{\alpha} Y_t\right) \leq \epsilon\bigg|X\right) > 1- n\exp\left(-\frac{\alpha \epsilon^2}{8  (b_2 -b_1)^2}\right) \ \text{for all} \ X \in \calc$$
    By Weyl's theorem, $\lambda_{\max} (\frac{1}{\alpha} \sum_{t=1}^{\alpha} Y_t) = \lambda_{\max} (\frac{1}{\alpha} \sum_{t=1}^{\alpha} (Z_t - \E_{t-1}(Z_t|X))
    \geq \lambda_{\max} (\frac{1}{\alpha} \sum_{t=1}^{\alpha} Z_t) +  \lambda_{\min} (\frac{1}{\alpha} \sum_{t=1}^{\alpha} -\E_{t-1}(Z_t|X))$.
    Since $\lambda_{\min} (\frac{1}{\alpha} \sum_{t=1}^{\alpha} -\E_{t-1}(Z_t|X)) = - \lambda_{\max} (\frac{1}{\alpha} \sum_{t=1}^{\alpha} \E_{t-1}(Z_t|X))\geq -b_4$, thus $ \lambda_{\max} (\frac{1}{\alpha} \sum_{t=1}^{\alpha} Y_t) \geq \lambda_{\max} (\frac{1}{\alpha} \sum_{t=1}^{\alpha} Z_t) - b_4$.
    Therefore,
    $$\mathbf{P}\left( \lambda_{\max} \left(\frac{1}{\alpha} \sum_{t=1}^{\alpha} Z_t\right) \leq b_4 + \epsilon|X\right) > 1- n\exp\left(-\frac{\alpha \epsilon^2}{8  (b_2 -b_1)^2}\right) \ \text{for all} \ X \in \calc$$

\item Let $Y_t = \E_{t-1}(Z_t|X) - Z_t$. As before, $\E_{t-1}(Y_t|X) = 0$ and conditioned on any $X \in {\cal C}$, the $Y_t$'s are independent and $\mathbf{P}(Y_t^2 \preceq (b_2 -b_1)^2 I|X) = 1$.  As before, applying Theorem \ref{azuma}, we get that for any $\epsilon >0$,
    $$\mathbf{P}\left( \lambda_{\max} \left(\frac{1}{\alpha} \sum_{t=1}^{\alpha} Y_t\right) \leq \epsilon|X\right) > 1- n\exp\left(-\frac{\alpha \epsilon^2}{8  (b_2 -b_1)^2}\right) \ \text{for all} \ X \in \calc $$
    By Weyl's theorem, $\lambda_{\max}(\frac{1}{\alpha}\sum_{t=1}^{\alpha} Y_t) = \lambda_{\max}(\frac{1}{\alpha} \sum_{t=1}^{\alpha}(\E_{t-1}(Z_t|X) - Z_t)) \geq \lambda_{\min} (\frac{1}{\alpha} \sum_{t=1}^{\alpha} \E_{t-1}(Z_t|X)) + \lambda_{\max} (\frac{1}{\alpha} \sum_{t=1}^{\alpha} -Z_t) = \lambda_{\min} (\frac{1}{\alpha} \sum_{t=1}^{\alpha} \E_{t-1}(Z_t|X)) - \lambda_{\min} (\frac{1}{\alpha} \sum_{t=1}^{\alpha} Z_t) \ge b_3 - \lambda_{\min} (\frac{1}{\alpha} \sum_{t=1}^{\alpha} Z_t)$
Therefore, for any $\epsilon >0$,
    $$\mathbf{P} \left(\lambda_{\min}\left(\frac{1}{\alpha}\sum_{t=1}^{\alpha} Z_t\right) \geq b_3 -\epsilon|X\right) \geq  1- n \exp\left(-\frac{\alpha \epsilon^2}{8(b_2-b_1)^2}\right) \ \text{for all} \ X \in \calc$$
\een
\end{proof}

\subsection{Proof of Lemma \ref{ind_expec}}
\begin{proof}
Denote conditional expectation $\E(Xh(Y)|g(Y))$ as $\E_{X,Y|g(Y)}(Xh(Y))$, then
\bea
\E(Xh(Y)|g(Y))=&&\E_{Y|g(Y)}\left(\E_{X|Y,g(Y)}\left(Xh(Y)\right)\right)\nn\\=&&\E_{Y|g(Y)}\left(\E_{X|Y}\left(Xh(Y)\right)\right)\nn\\=&&\E_{Y|g(Y)}\left(\E_{X|Y}\left(X\right)h(Y)\right) \nn\\=&&\E_{Y|g(Y)}\left(\E(X)h(Y)\right) \nn\\=&&\E(X)\E(h(Y)|g(Y)).
\eea
\end{proof}

\subsection{Proof of Lemma \ref{delta_kappa}}
\begin{proof}
Let $A = I - PP'$. By definition, $\delta_s(A) := \max\{ \max_{|T| \leq s}(\lambda_{\max}(A_T'A_T) -1),\max_{|T| \leq s} ( 1 - \lambda_{\min} (A_T' A_T))) \} $.
Notice that $A_T'A_T = I - I_T' PP'I_T$. Since $I_T' PP'I_T$ is p.s.d., by Weyl's theorem, $\lambda_{\max}(A_T'A_T) \leq1$.
Since $\lambda_{\max}(A_T'A_T)- 1\leq 0$ while $1 - \lambda_{\min}(A_T'A_T) \geq 0$, thus, 
\beq
\delta_s(I - PP') = \max_{|T| \leq s}\Big(1 - \lambda_{\min} ( I - I_T' PP'I_T)\Big) \label{defn_kappa_1}
\eeq
By Definition, $\kappa_s(P) = \max_{|T| \leq s} \frac{\|I_T' P\|_2}{\|P\|_2} =\max_{|T| \leq s} \|I_T' P\|_2$.
Notice that $\|I_T' P\|_2^2 = \lambda_{\max} (I_T' PP'I_T)  =  1-\lambda_{\min} (I - I_T'PP'I_T)$ \footnote{This follows because $B=I_T'PP'I_T$ is a Hermitian matrix.  Let $B = U \Sigma U'$ be its EVD. Since $UU'=I$, $\lambda_{\min}(I-B) = \lambda_{\min}(U(I - \Sigma)U') =\lambda_{\min}(I - \Sigma) = 1 - \lambda_{\max}(\Sigma) = 1-\lambda_{\max}(B)$.}, and so
\beq
\kappa_s^2(P) =\max_{|T| \leq s} \Big(1 - \lambda_{\min} (I - I_T'PP'I_T)\Big)\label{defn_kappa_2}
\eeq
From (\ref{defn_kappa_1}) and (\ref{defn_kappa_2}), we get $ \delta_s(I-PP') = \kappa_s^2 (P) $.
\end{proof}

\bibliographystyle{IEEEtran}
\bibliography{tipnewpfmt}
\end{document}

%% file: zcom.tex
\setlength{\arraycolsep}{0.03cm}
\newcommand{\xhat}{\hat{x}}
\newcommand{\xpred}{\hat{x}_{t|t-1}}
\newcommand{\xupd}{\hat{x}_{t|t}}
\newcommand{\Ppred}{P_{t|t-1}}
\newcommand{\ty}{\tilde{y}_t}
\newcommand{\tty}{\tilde{y}_{t,\text{res}}}
\newcommand{\tw}{\tilde{w}_t}
\newcommand{\ttw}{\tilde{w}_{t,f}}
\newcommand{\betahat}{\hat{\beta}}

\newcommand{\range}{\operatorname{range}}
\newcommand{\rank}{\operatorname{rank}}
\newcommand{\supp}{\operatorname{supp}}
\newcommand{\R}{\mathbb{R}}
\newcommand{\calb}{\mathcal{B}}
\newcommand{\ecalb}{\mathcal{B}^e}
\newcommand{\calc}{\mathcal{C}}
\newcommand{\ecalc}{\mathcal{C}^e}
\newcommand{\cald}{\mathcal{D}}
\newcommand{\ecald}{\mathcal{D}^e}

\newcommand{\what}{\widehat}

\newtheorem{theorem}{Theorem}[section]
\newtheorem{lem}[theorem]{Lemma}
\newtheorem{sigmodel}[theorem]{Signal Model}
\newtheorem{corollary}[theorem]{Corollary}
\newtheorem{defin}[theorem]{Definition}
\newtheorem{definition}[theorem]{Definition}
\newtheorem{remark}[theorem]{Remark}
\newtheorem{example}[theorem]{Example}
\newtheorem{ass}[theorem]{Assumption}
\newtheorem{proposition}[theorem]{Proposition}
\newtheorem{fact}[theorem]{Fact}
\newtheorem{heuristic}[theorem]{Heuristic}
\newtheorem{alg}[theorem]{Algorithm}

\newcommand{\ypast}{y_{1:t-1}}
\newcommand{\sone}{S_{*}}
\newcommand{\sinf}{{S_{**}}}
\newcommand{\smax}{S_{\max}}
\newcommand{\smin}{S_{\min}}
\newcommand{\samax}{S_{a,\max}}
\newcommand{\Nhat}{{\hat{N}}}

\newcommand{\Dnum}{D_{num}}
\newcommand{\pss}{p^{**,i}}
\newcommand{\fr}{f_{r}^i}

\newcommand{\A}{{\cal A}}
\newcommand{\Z}{{\cal Z}}
\newcommand{\B}{{\cal B}}

\newcommand{\reg}{{\cal G}}
\newcommand{\const}{\mbox{const}}

\newcommand{\trace}{\mbox{tr}}

\newcommand{\hsim}{{\hspace{0.0cm} \sim  \hspace{0.0cm}}}
\newcommand{\he}{{\hspace{0.0cm} =  \hspace{0.0cm}}}

\newcommand{\vect}[2]{\left[\begin{array}{cccccc}
     #1 \\
     #2
   \end{array}
  \right]
  }

\newcommand{\matr}[2]{ \left[\begin{array}{cc}
     #1 \\
     #2
   \end{array}
  \right]
  }
\newcommand{\vc}[2]{\left[\begin{array}{c}
     #1 \\
     #2
   \end{array}
  \right]
  }

\newcommand{\gdot}{\dot{g}}
\newcommand{\Cdot}{\dot{C}}
\newcommand{\re}{\mathbb{R}}
\newcommand{\n}{{\cal N}}  
\newcommand{\N}{{\overrightarrow{\bf N}}}  
\newcommand{\chat}{\tilde{C}_t}
\newcommand{\chati}{\chat^i}

\newcommand{\cmin}{C^*_{min}}
\newcommand{\twi}{\tilde{w}_t^{(i)}}
\newcommand{\twj}{\tilde{w}_t^{(j)}}
\newcommand{\wi}{{w}_t^{(i)}}
\newcommand{\twio}{\tilde{w}_{t-1}^{(i)}}

\newcommand{\tWi}{\tilde{W}_n^{(m)}}
\newcommand{\tWj}{\tilde{W}_n^{(k)}}
\newcommand{\Wi}{{W}_n^{(m)}}
\newcommand{\tWio}{\tilde{W}_{n-1}^{(m)}}

\newcommand{\ds}{\displaystyle}

\newcommand{\SAR}{S$\!$A$\!$R }
\newcommand{\MAR}{MAR}
\newcommand{\MMRF}{MMRF}
\newcommand{\AR}{A$\!$R }
\newcommand{\GMRF}{G$\!$M$\!$R$\!$F }
\newcommand{\DTM}{D$\!$T$\!$M }
\newcommand{\MSE}{M$\!$S$\!$E }
\newcommand{\RCS}{R$\!$C$\!$S }
\newcommand{\uomega}{\underline{\omega}}
\newcommand{\lu}{\mu}
\newcommand{\g}{g}
\newcommand{\s}{{\bf s}}
\newcommand{\bft}{{\bf t}}
\newcommand{\refmap}{{\cal R}}
\newcommand{\totrefl}{{\cal E}}
\newcommand{\beq}{\begin{equation}}
\newcommand{\eeq}{\end{equation}}
\newcommand{\bdm}{\begin{displaymath}}
\newcommand{\edm}{\end{displaymath}}
\newcommand{\hatz}{\hat{z}}
\newcommand{\hatu}{\hat{u}}
\newcommand{\tilz}{\tilde{z}}
\newcommand{\tilu}{\tilde{u}}
\newcommand{\hhatz}{\hat{\hat{z}}}
\newcommand{\hhatu}{\hat{\hat{u}}}
\newcommand{\tilc}{\tilde{C}}
\newcommand{\hatc}{\hat{C}}
\newcommand{\tim}{n}

\newcommand{\ssp}{\renewcommand{\baselinestretch}{1.0}}
\newcommand{\defd}{\mbox{$\stackrel{\mbox{$\triangle$}}{=}$}}
\newcommand{\goes}{\rightarrow}
\newcommand{\tends}{\rightarrow}
\newcommand{\se}{&=&}
\newcommand{\sdefn}{& :=  &}
\newcommand{\sle}{& \le &}
\newcommand{\sge}{& \ge &}
\newcommand{\plusminus}{\stackrel{+}{-}}
\newcommand{\Ey}{E_{Y_{1:t}}}
\newcommand{\ey}{E_{Y_{1:t}}}

\newcommand{\equivto}{\mbox{~~~which is equivalent to~~~}}
\newcommand{\nonzero}{i:\pi^n(x^{(i)})>0}
\newcommand{\nonzeroc}{i:c(x^{(i)})>0}

\newcommand{\supn}{\sup_{\phi:||\phi||_\infty \le 1}}

\newcommand{\eps}{\epsilon}
\newcommand{\bd}{\begin{definition}}
\newcommand{\ed}{\end{definition}}
\newcommand{\udq}{\underline{D_Q}}
\newcommand{\td}{\tilde{D}}
\newcommand{\epsinv}{\epsilon_{inv}}
\newcommand{\al}{\mathcal{A}}

\newcommand{\bfx} {\bf X}
\newcommand{\bfy} {\bf Y}
\newcommand{\bfz} {\bf Z}
\newcommand{\ddas}{\mbox{${d_1}^2({\bf X})$}}
\newcommand{\ddbs}{\mbox{${d_2}^2({\bfx})$}}
\newcommand{\dda}{\mbox{$d_1(\bfx)$}}
\newcommand{\ddb}{\mbox{$d_2(\bfx)$}}
\newcommand{\xinc}{{\bfx} \in \mbox{$C_1$}}
\newcommand{\eqa}{\stackrel{(a)}{=}}
\newcommand{\eqb}{\stackrel{(b)}{=}}
\newcommand{\eqe}{\stackrel{(e)}{=}}
\newcommand{\leqc}{\stackrel{(c)}{\le}}
\newcommand{\leqd}{\stackrel{(d)}{\le}}

\newcommand{\leqa}{\stackrel{(a)}{\le}}
\newcommand{\leqb}{\stackrel{(b)}{\le}}
\newcommand{\leqe}{\stackrel{(e)}{\le}}
\newcommand{\leqf}{\stackrel{(f)}{\le}}
\newcommand{\leqg}{\stackrel{(g)}{\le}}
\newcommand{\leqh}{\stackrel{(h)}{\le}}
\newcommand{\leqi}{\stackrel{(i)}{\le}}
\newcommand{\leqj}{\stackrel{(j)}{\le}}

\newcommand{\w}{{W^{LDA}}}
\newcommand{\halpha}{\hat{\alpha}}
\newcommand{\hsigma}{\hat{\sigma}}
\newcommand{\slmax}{\sqrt{\lambda_{max}}}
\newcommand{\slmin}{\sqrt{\lambda_{min}}}
\newcommand{\lmax}{\lambda_{max}}
\newcommand{\lmin}{\lambda_{min}}

\newcommand{\da} {\frac{\alpha}{\sigma}}
\newcommand{\chka} {\frac{\check{\alpha}}{\check{\sigma}}}
\newcommand{\sumo}{\sum _{\underline{\omega} \in \Omega}}
\newcommand{\distance}{d\{(\hatz _x, \hatz _y),(\tilz _x, \tilz _y)\}}
\newcommand{\col}{{\rm col}}
\newcommand{\rcs}{\sigma_0}
\newcommand{\CalR}{{\cal R}}
\newcommand{\df}{{\delta p}}
\newcommand{\dq}{{\delta q}}
\newcommand{\dZ}{{\delta Z}}
\newcommand{\pprime}{{\prime\prime}}

\newcommand{\vn}{N}

\newcommand{\bv}{\begin{vugraph}}
\newcommand{\ev}{\end{vugraph}}
\newcommand{\bi}{\begin{itemize}}
\newcommand{\ei}{\end{itemize}}
\newcommand{\ben}{\begin{enumerate}}
\newcommand{\een}{\end{enumerate}}
\newcommand{\be}{\protect\[}
\newcommand{\ee}{\protect\]}
\newcommand{\bean}{\begin{eqnarray*} }
\newcommand{\eean}{\end{eqnarray*} }
\newcommand{\bea}{\begin{eqnarray} }
\newcommand{\eea}{\end{eqnarray} }
\newcommand{\nn}{\nonumber}
\newcommand{\ba}{\begin{array} }
\newcommand{\ea}{\end{array} }
\newcommand{\ep}{\mbox{\boldmath $\epsilon$}}
\newcommand{\epp}{\mbox{\boldmath $\epsilon '$}}
\newcommand{\Lep}{\mbox{\LARGE $\epsilon_2$}}
\newcommand{\und}{\underline}
\newcommand{\pdif}[2]{\frac{\partial #1}{\partial #2}}
\newcommand{\odif}[2]{\frac{d #1}{d #2}}
\newcommand{\dt}[1]{\pdif{#1}{t}}
\newcommand{\urho}{\underline{\rho}}

\newcommand{\spc}{{\cal S}}
\newcommand{\tspc}{{\cal TS}}

\newcommand{\uv}{\underline{v}}
\newcommand{\us}{\underline{s}}
\newcommand{\uc}{\underline{c}}
\newcommand{\utheta}{\underline{\theta}^*}
\newcommand{\ualpha}{\underline{\alpha^*}}

\newcommand{\uxy}{\underline{x}^*}
\newcommand{\uxyj}{[x^{*}_j,y^{*}_j]}
\newcommand{\arcl}[1]{arclen(#1)}
\newcommand{\one}{{\mathbf{1}}}

\newcommand{\uxyjt}{\uxy_{j,t}}
\newcommand{\E}{\mathbf{E}}

\newcommand{\rhomat}{\left[\begin{array}{c}
                        \rho_3 \ \rho_4 \\
                        \rho_5 \ \rho_6
                        \end{array}
                   \right]}
\newcommand{\deltat}{\tau} 
\newcommand{\deltatt}{\Delta t_1}
\newcommand{\ceil}[1]{\ulcorner #1 \urcorner}

\newcommand{\xxi}{x^{(i)}}
\newcommand{\txi}{\tilde{x}^{(i)}}
\newcommand{\txj}{\tilde{x}^{(j)}}

\newcommand{\mi}[1]{{#1}^{(m,i)}}

%% file: ReProCSCommands.tex
\newcommand{\ttrain}{t_{\text{train}}}
\newcommand{\new}{\text{new}}
\newcommand{\cs}{\text{cs}}
\newcommand{\bigO}{\mathcal{O}}
\newcommand{\newset}{\text{new-set}}
\newcommand{\old}{\text{old}}
\newcommand{\Shat}{\hat{S}}
\newcommand{\Lhat}{\hat{L}}
\newcommand{\Ltilde}{\tilde{L}}
\newcommand{\Phat}{\hat{P}}
\newcommand{\Qhat}{\hat{Q}}
\newcommand{\Span}{\text{span}}
\newcommand{\del}{\text{del}}
\newcommand{\diag}{\text{diag}}
\newcommand{\add}{\text{add}}

\newcommand{\train}{\text{train}}
\newcommand{\thresh}{\hat{\lambda}^-}
\newcommand{\betathresh}{\eta} 
\newcommand{\That}{\hat{T}}
\newcommand{\zetasp}{{\zeta_*^+}}
\newcommand{\SE}{\text{SE}}
\newcommand{\egam}{\Gamma^e}